\documentclass{article}

    \PassOptionsToPackage{numbers, compress}{natbib}

\usepackage[preprint]{neurips_2025}

\usepackage[utf8]{inputenc} 
\usepackage[T1]{fontenc}    
\usepackage{hyperref}       
\usepackage{url}            
\usepackage{booktabs}       
\usepackage{amsfonts}       
\usepackage{nicefrac}       
\usepackage{microtype}      
\usepackage{xcolor}         
\usepackage{enumitem}
\usepackage{amsmath}
\usepackage{amssymb}
\usepackage{mathtools}
\usepackage{amsthm}
\usepackage{nccmath}
\usepackage{multirow}
\usepackage[ruled,linesnumbered, vlined]{algorithm2e}
\usepackage{algorithmic}

\theoremstyle{plain}
\newtheorem{theorem}{Theorem}[section]

\newtheorem{lemma}[theorem]{Lemma}

\theoremstyle{definition}
\newtheorem{definition}[theorem]{Definition}
\newtheorem{assumption}[theorem]{Assumption}
\theoremstyle{remark}

\newcommand{\fullset}{D_{\mathbb{N}}}
\newcommand{\datasubset}{\mathcal{S}}
\newcommand{\subsetsize}{|\datasubset|}

\newcommand{\fullmodel}{\Theta_{D_\clientset}}

\newcommand{\subsetmodel}{\Theta_{\datasubset}}
\newcommand{\reportedsubsetmodel}{\Theta_{\reportedsubset}}

\newcommand{\sv}{\phi^{SV}}
\newcommand{\tsv}{\phi^{TSV}}

\newcommand{\empiricaldatavalue}{\widehat{\phi}}
\newcommand{\weightsv}{w^{SV}}
\newcommand{\betasv}{\beta^{SV}_i}

\newcommand{\clientlevelgranddataset}{\mathbb{D}_{\mathbb{N}}}
\newcommand{\clientleveldatasubset}{\mathbb{D}_{\mathcal{C}}}
\newcommand{\stoi}{D^{\mathcal{S}}_{i}}

\newcommand{\reportedstoi}{\widehat{D}^{\mathcal{S}}_{i}}

\newcommand{\reportedstominusi}{\widehat{D}^{\mathcal{S}}_{-i}}
\newcommand{\reportedsubset}{\widehat{\datasubset}}

\newcommand{\splus}{\mathcal{S}^{+}}

\newcommand{\stominusi}{D^{\mathcal{S}}_{-i}}

\newcommand{\clientsins}{\mathbb{N}(\mathcal{S})}

\newcommand{\clientset}{\mathbb{N}}
\newcommand{\clientsubset}{\mathcal{C}}

\newcommand{\realrepresentationi}{\theta^{(t)}_{i|\datasubset}}

\newcommand{\realrepresentationsi}{\mathbf{\theta}_{i|\datasubset}}
\newcommand{\reportedrepresentationsi}{\mathbf{\theta}_{i|\reportedsubset}}

\DeclareRobustCommand{\set}[1][]{\{#1\}}

\title{Data Overvaluation Attack and Truthful Data Valuation in Federated Learning}

\newcommand{\osaka}{$^{1}$}
\newcommand{\kyoto}{$^{2}$}
\newcommand{\scityo}{$^{3}$}
\newcommand{\shenzhen}{$^{4}$}
\newcommand{\oseikei}{$^{5}$}

\author{%
  \begin{tabular}{c}
    Shuyuan Zheng\osaka \qquad Sudong Cai\kyoto \qquad Chuan Xiao\osaka \qquad Yang Cao\scityo \\[0.5em]
    Jianbin Qin\shenzhen\qquad Masatoshi Yoshikawa\oseikei\qquad Makoto Onizuka\osaka\\[0.8em]
  \end{tabular}\\
  \begin{tabular}{c}
      \osaka The University of Osaka, \kyoto Kyoto University,\\ 
      \scityo Institute of Science Tokyo, \shenzhen Shenzhen University, \oseikei Osaka Seikei University\\
  \end{tabular}\\
    \texttt{zheng@ist.osaka-u.ac.jp}
}

\begin{document}

\maketitle

\begin{abstract}
In collaborative machine learning (CML), data valuation, i.e., evaluating the contribution of each client's data to the machine learning model, has become a critical task for incentivizing and selecting positive data contributions. However, existing studies often assume that clients engage in data valuation truthfully, overlooking the practical motivation for clients to exaggerate their contributions. To unlock this threat, this paper introduces the \textit{data overvaluation attack}, enabling strategic clients to have their data significantly overvalued in federated learning, a widely adopted paradigm for decentralized CML. Furthermore, we propose a Bayesian truthful data valuation metric, named \textit{Truth-Shapley}. Truth-Shapley is the unique metric that guarantees some promising axioms for data valuation while ensuring that clients' optimal strategy is to perform truthful data valuation under certain conditions. Our experiments demonstrate the vulnerability of existing data valuation metrics to the proposed attack and validate the robustness and effectiveness of Truth-Shapley.
\end{abstract}

\section{Introduction}

As data protection regulations become increasingly stringent, it is becoming more challenging for enterprises to collect sufficient high-quality training data and share the data with each other for machine learning (ML). 
To address this issue, federated learning (FL)~\citep{mcmahan2017communication}, a decentralized paradigm for collaborative machine learning (CML), has emerged as a promising solution, which enables enterprises to train accurate ML models by uploading some transformed representations of data, such as local models and gradients, instead of sharing the raw data. 
Given that enterprises' datasets are often highly heterogeneous, a critical task in CML is data valuation, that is, how to reasonably evaluate the contribution of different heterogeneous datasets to model performance improvement. 
Based on the datasets' data values, i.e., the outcome of data valuation, enterprises can select higher-quality data to further enhance model performance and fairly allocate rewards among themselves, such as the revenue made by deploying the model.

In the literature, \textit{marginal contribution}-based valuation metrics, represented by the leave-one-out (LOO)~\citep{cook1977detection} and the Shapley value (SV)~\citep{shapley1953value}, have been widely adopted for data valuation in CML. 
These metrics evaluate data value by measuring the impact of including or excluding a dataset on model performance. 
For example, the SV requires iterating over all possible combinations of the datasets and computing the model utility improvement contributed by each dataset to each combination, leading to significant computational costs for repeated model retraining. 
Consequently, extensive research efforts (e.g., \citep{ghorbani2019data, jia2019towards, jia2019efficient, kwon2021efficient}) have been devoted to improving the computational efficiency of data valuation to enhance its practicality.
However, existing studies overlook a critical trust vulnerability in data overvaluation within FL scenarios: \textit{Clients may misreport their uploaded data representations to manipulate the utilities of the retrained models, thereby inflating their data value for illicit gains in data selection and reward allocation.}
This gap motivates us to conduct the first exploration of data overvaluation and truthful data valuation in FL.

In this paper, we propose a novel attack method grounded in a strong theoretical foundation, targeting data valuation in FL scenarios: the data overvaluation attack. 
This attack strategically manipulates the attacker's local data during model retraining to selectively enhancing or degrading the utilities of the retrained models, thereby inflating the attacker's data value.
Notably, the attack works against all linear data valuation metrics, which cover most of the state-of-the-arts (SOTAs) including the LOO and the SV. 
Our experimental results demonstrate that the data overvaluation attack can increase the attacker’s SV and LOO value by \textbf{77\% and 97\% on average} across various FL tasks, respectively.

Next, we explore how to ensure \textit{truthful data valuation}.
We theoretically characterize the subclass of \textit{linear} data valuation metrics that can resist the data overvaluation attack under certain conditions. 
This characterization is fundamental since most of mainstream data valuation metrics are linear, inlcuding the LOO, the SV, Beta Shapley~\citep{kwon2022beta}, and Banzhaf value~\cite{wang2023data}.
From this subclass, we identify a novel valuation metric, named \textit{Truth-Shapley}. 
Similar to the SV, Truth-Shapley uniquely satisfies a set of promising axioms for valuation, thereby ensuring effective data selection and fair reward allocation. 
Therefore, regardless of whether a data overvaluation attack occurs, Truth-Shapley serves as an excellent metric for robust and effective data valuation.

We summarize our contributions as follows. 
\textbf{(1)} We propose the data overvaluation attack.
This attack reveals the unexplored vulnerability of existing data valuation metrics to strategic manipulation during the model retraining process and thus opens up a new research direction toward truthful data valuation in FL.
\textbf{(2)} We theoretically analyze and characterize the necessary and sufficient conditions for linear data valuation metrics to ensure \textit{Bayesian truthful} data valuation. 
This characterization facilitates a rigorous assessment of their robustness against data overvaluation.
\textbf{(3)} We propose Truth-Shapley, which is the unique Bayesian truthful data valuation metric that satisfies some key fairness axioms for valuation.
\textbf{(4)} We conduct extensive experiments across various FL scenarios. 
Our results demonstrate the vulnerability of existing data valuation metrics to the proposed attack and validate the robustness and effectiveness of Truth-Shapley in data selection and reward allocation.

\section{Preliminaries}

We consider a \textit{horizontal} CML setting with $N$ clients $\clientset=\set[1,\dots, N]$, whose datasets share the same feature space.
Each client $i\in \clientset$ holds a local dataset $D_i=\{D_{i, 1},\dots, D_{i, M_i}\}$, which is partitioned into $M_i$ data blocks--subsets of data samples with distinct distributions.
The clients pool their data and utilize a CML algorithm $\mathcal{A}$ to train an ML model $\fullmodel =\mathcal{A}(\fullset)$ under the coordination of a central server, where $\fullset = \{D_{i,j} \mid i\in \clientset, j\in[M_i]\}$ denotes the set of all data blocks across the clients.
We refer to $\fullmodel$ as the \textit{grand model} as it is the final product of CML.
Note that our work can be extended to the \textit{vertical} setting where data blocks differ in feature space (see \textbf{Appendix~\ref{sec:discuss}} for a detailed discussion).

\textbf{Data valuation.}
After training the grand model, the server performs data valuation to evaluate each data block $D_{i,j}$'s \textit{data value} $\phi_{i,j}(\fullset, v)$, which reflects the contribution of $D_{i,j}$ to improving the utility $v(\fullset)$ of the grand model $\fullmodel$.
We simply write $\phi_{i,j}(\fullset, v)$ as $\phi_{i,j}$ when there is no ambiguity.
Consequently, the data valuation task is to design a data valuation metric $\phi$ to determine data values for all data blocks involved in the CML (see Definition~\ref{def:data_valuation}).
Data valuation facilitates the following two downstream tasks, which ensure fairness and incentivize collaboration. 
\textbf{(1) Data selection:} The server selects data blocks $D_{i,j}$ with high(er) data values $\phi_{i,j}$ to enhance the performance of CML next time, which is critical when there exist clients who contribute trivial data or outliers~\citep{cohen2005feature, nagalapatti2021game}.
\textbf{(2) Reward allocation:} For each data block $D_{i,j}$, the server allocates a reward $R_{i,j}(\phi)$ to its owner based on its data value $\phi_{i,j}$.
The rewards may be revenue obtained from commercializing the grand model~\citep{nguyen2022trade} or fees collected from the clients for participating in CML~\citep{ohrimenko2019collaborative}.
We follow prior work~\citep{agarwal2019marketplace, ohrimenko2019collaborative, song2019profit, nguyen2022trade} to assume each client's reward $R_i(\phi) = \sum_{j\in[M_i]} R_{i,j}(\phi)$ increases with their own data value $\phi_i$ while decreases with the other clients' data values $\phi_{-i} = \sum_{i'\in \mathbb{N} \setminus \{i\}} \phi_{i'}$ (see \textbf{Appendix~\ref{appendix:reward}} for examples). 
Consequently, a reward-seeking client $i$ aims to maximize $\phi_i$ and minimizes $\phi_{-i}$ to optimize their reward $R(\phi_i)$.

\begin{definition}[Data Valuation]
\label{def:data_valuation}
    A utility metric $v: 2^{\fullset}\rightarrow \mathbb{R}$ maps a set of data blocks $S \subseteq \fullset$ to the utility $v(S)$ of the model $\mathcal{A}(S)$ with $v(\emptyset) = 0$. 
    A data valuation metric $\phi: \fullset \times \mathcal{G}(\fullset)$ allocates data values $\phi(\fullset, v) = \{\phi_{i,j}(\fullset, v) \mid i\in \clientset, j\in [M_i]\}$, where $\mathcal{G}(\fullset) = \{v\mid v:2^{\fullset} \rightarrow \mathbb{R}\}$.
\end{definition}


\textbf{Shapley value.}
As a classic contribution metric in cooperative game theory, the SV has been widely adopted for data valuation in CML because it uniquely satisfies some promising properties, including linearity (LIN), efficiency (EFF), dummy actions (DUM), and symmetry (SYM) (see \textbf{Appendix~\ref{appendix:sv}} for details).
Specifically, the SV computes the data value $\sv_{i,j}$ of each data block $D_{i,j}$ as follows.
\begin{align}
\label{eq:sv}
\sv_{i,j}(\fullset, v) \coloneqq \sum_{\datasubset\subseteq \fullset \setminus \set[D_{i,j}]} \weightsv(\subsetsize) \big(v(\splus) - v(\datasubset)\big)
\end{align}
where $\weightsv(\subsetsize) \coloneqq  \frac{\subsetsize!(|\fullset|- \subsetsize - 1)!}{|\fullset|!}$ and $\splus = \datasubset \cup \{D_{i,j}\}$.
The SV enumerates all possible subsets $\datasubset$ of $\fullset$ excluding the data block $D_{i,j}$ and calculates the utility improvement $\big(v(\splus) - v(\datasubset)\big)$ achieved by adding $D_{i,j}$ to subset $\datasubset$.
$\weightsv(\subsetsize)$ is a coefficient that weights the importance of $\datasubset$.
The data value $\sv_{i,j}$ thus is the weighted aggregation of all utility improvements attributed to  $D_{i,j}$.
Note that for each block subset $\datasubset \subset \fullset$, computing $v(\datasubset)$ requires retraining an ML model $\Theta_{\datasubset} = \mathcal{A}(\datasubset)$, which is referred to as a \textit{subset model}.

\textbf{Federated learning.}
To protect the clients' data privacy/property rights, we consider that Algorithm $\mathcal{A}$ is an FL algorithm, which allows the clients to perform CML without sharing local data.
Given a block set $\datasubset \subseteq \fullset$, let $\stoi$ denote the subset of data blocks in $\datasubset$ that belong to client $i$, let $\stominusi$ denote the other clients' data blocks in $\datasubset$, and let $\clientsins$ denote the set of clients with at least one data block in $\datasubset$. 
An FL algorithm $\mathcal{A}(\datasubset)$ consists of $T$ rounds of decentralized model training. 
In each round $t\in[T]$, given a global model $\Theta^{(t-1)}_{\datasubset}$ from the server, each client $i\in \clientsins$ employs a local training algorithm $\mathcal{A}_i$ to learn a representation $\realrepresentationi=\mathcal{A}_i(\stoi; \Theta^{(t-1)}_{\datasubset})$ from their local data $\stoi$;
the server then collects all learned presentations and calls a fusion algorithm $\mathcal{A}_0$ to derive a new global model $\Theta^{(t)}_{\datasubset} =\mathcal{A}_0(\{\realrepresentationi | i \in \clientsins \})$.
Finally, $\mathcal{A}(\datasubset)$ outputs the global model at the final round $T$ as the grand or subset model, i.e., $\subsetmodel = \Theta^{(T)}_{\datasubset}$.
Note that the above abstraction encompasses most FL algorithms;
the learned data representation $\realrepresentationi$ is a general concept that may be instantiated as a local model/gradient~\citep{mcmahan2017communication}, embeddings of local data~\citep{thapa2022splitfed}, predicted logits~\citep{li2019fedmd}, among others.

\section{Data Overvaluation Attack}

\textbf{Data overvaluation against the SV.}
Although the SV fairly allocates data values to honest clients, a strategic client can manipulate their SVs by misreporting their data representations when retraining subset models in FL.
Specifically, given a subset of data blocks $\datasubset \subset \fullset$, when evaluating the utility $v(\datasubset)$ for calculating $\sv_{i,j}$, attacker $i$ can misreport their data representations $\realrepresentationsi\coloneqq \{\theta^{(t)}_{i|\datasubset}\}_{t\in[T]}$ to alter the subset model $\Theta_{\datasubset} = \mathcal{A}(\datasubset)$.
Consequently, the utility $v(\datasubset)$ of the subset model $\Theta_{\datasubset}$ might be altered, which affects their block-level SVs $\sv_{i, j}$.
Similarly, attacker $i$ can also manipulate their client-level SV $\sv_{i}$ by altering the model utility $v(\datasubset)$, as the client-level SV $\sv_{i}=\sum_{j\in[M_i]}\sv_{i,j}$ can be written in the following form:
\begin{align*}
     &\sv_{i}(\fullset, v) = \sum_{\datasubset\subseteq \fullset} \betasv(\datasubset) \cdot v(\datasubset), \text{ where} \\
      & \betasv(\datasubset) \coloneqq 
   \begin{cases}
            (|\stoi|\fullset - |D_{i}||\datasubset|) \cdot \frac{(|\datasubset |- 1)! (\fullset-|\datasubset| - 1)!}{\fullset!}, & \datasubset \subset \fullset, \datasubset \neq \emptyset, \\
            |D_i|/ \fullset, & \datasubset = \fullset, \\
            - |D_i|/ \fullset, & \datasubset = \emptyset.
   \end{cases}
\end{align*}
Because $\frac{\partial \sv_i}{\partial (v(\datasubset))} \!=\! \betasv(\datasubset)$, when $\betasv(\datasubset)\! > \!0$, increasing $v(\datasubset)$ enhances $\sv_{i}$; when $\betasv(\datasubset) \!<\! 0$, decreasing $v(\datasubset)$ improves $\sv_{i}$; when $\betasv(\datasubset) \!=\! 0$, changing $v(\datasubset)$ has no effect on $\sv_{i}$.

\textbf{Generalization.}
We further generalize the data overvaluation attack in Definition~\ref{def:overvaluation} to manipulate all linear data valuation metrics (i.e., $\phi(\fullset, v_1+v_2) = \phi(\fullset, v_1)  + \phi(\fullset, v_2)$ for any two utility metrics $v_1, v_2$). 
Specifically, Lemma \ref{lem:data_value_form} indicates that any linear data value $\phi_i$ can be expressed as a weighted sum of model utilities $\{v(\datasubset)\}_{\datasubset\subseteq \fullset}$. 
Consequently, similar to the case of the SV, for each subset $\datasubset \subset \fullset$, and for each client $i\in\clientsins$, when $\beta_i(\datasubset)$ is positive (negative), they can increase (decrease) $v(\datasubset)$ to enhance their linear data value $\phi_i$.
However, unlike the SV, some data valuation metrics such as Beta Shapley and Banzhaf value do not satisfy EFF, meaning that an increase in $\phi_{i}$ does not necessarily lead to a decrease in $\phi_{-i}$.
As a result, attacker $i$ may not always receive a higher reward.
Therefore, when $\beta_{i}(\datasubset)$ is positive (negative), we should also ensure that $\beta_{-i}(\datasubset)=\sum_{i' \in \mathbb{N}(\datasubset) \setminus \{i\}} \beta_{i'}(\datasubset)$ is non-positive to prevent $\phi_{-i}$ from increasing.

\begin{lemma}
\label{lem:data_value_form}
    If a data valuation metric $\phi$ satisfies LIN, then there exist functions $\beta_i: 2^{\fullset} \rightarrow \mathbb{R}$ and $\beta_{i,j}: 2^{\fullset} \rightarrow \mathbb{R}$ such that $\phi_{i,j}(\fullset, v) \equiv \sum_{\datasubset \subseteq \fullset} \beta_{i,j}(\datasubset) \cdot v(\datasubset)$ and $\phi_{i}(\fullset, v) \equiv \sum_{\datasubset \subseteq \fullset} \beta_i(\datasubset) \cdot v(\datasubset)$ for all $i \in \clientset$ and $j \in [M_i]$.
\end{lemma}

\begin{definition}[Data Overvaluation Attack]
\label{def:overvaluation}
    A data overvaluation attack against a linear data valuation metric $\phi$ is to misreport data representations $\realrepresentationsi$ such that $v(\datasubset)$ increases when $\beta_i(\datasubset) > 0$ and $\beta_{-i}(\datasubset) \leq 0$, and that $v(\datasubset)$ decreases when $\beta_i(\datasubset) < 0$ and $\beta_{-i}(\datasubset) \geq 0$, for all subsets $\datasubset \subset \fullset$.
\end{definition}

\begin{algorithm}[t]
    \caption{Data Overvaluation through Data Augmentation}
    \label{alg:data_overvaluation}
    \For{each subset $\datasubset\subset \fullset$}{
        \If{$i\in \clientsins$}{
            \If{$\beta_i(\datasubset) > 0$ and $\beta_{-i}(\datasubset) \leq 0$}{
                Attacker $i$: Positively augment $\stoi$ to obtain a dataset $\reportedstoi$\;
            }
            \ElseIf{$\beta_i(\datasubset) < 0$ and $\beta_{-i}(\datasubset) \geq 0$}{
                Attacker $i$: Negatively augment $\stoi$ to generate a dataset $\reportedstoi$\;
            }
            \Else{
                Attacker $i$: Honestly use $\reportedstoi = \stoi$\;
            }
        }
        Clients $\clientsins$: Retrain a surrogate model $\reportedsubsetmodel = \mathcal{A}(\reportedsubset)$ of $\subsetmodel$ using the FL algorithm $\mathcal{A}$\;
        Server: Evaluate the utility $v(\reportedsubset)$ of $\reportedsubsetmodel$\;
    }
    Server: Calculate $\empiricaldatavalue_{i,j}(\fullset, v), \forall i, \forall j$ and return them\;
\end{algorithm}

\textbf{Implementation.}
We propose Algorithm \ref{alg:data_overvaluation} to implement data overvaluation through \textit{data augmentation}.
Let $\reportedstoi$ and $\reportedstominusi$ denote the data used by attacker $i$ and the other clients for retraining a subset model $\reportedsubsetmodel = \mathcal{A}(\reportedsubset)$ as a surrogate of $\subsetmodel$, respectively, where $\reportedsubset = \reportedstoi \cup \reportedstominusi$.
That means, instead of truthfully using dataset $\stoi$, attacker $i$ may untruthfully employ dataset $\reportedstoi \neq \stoi$ to train their representations $\realrepresentationsi$, altering the utility of $\datasubset$ from $v(\datasubset)$ to $v(\reportedsubset)$.
\textbf{This untruthful manipulation is hard to observe since the server cannot collect the local data in FL.} 
More specifically, in Algorithm \ref{alg:data_overvaluation}, to compute data values, the server performs model retraining and utility evaluation for every subset $\datasubset \subset \fullset$.
Note that model retraining for the full set $\fullset$ is not needed since the grand model $\fullmodel$ has already been trained before data valuation. 
When a subset $\datasubset$ includes attacker $i$'s data blocks, if $\beta_i(\datasubset)$ is nonzero, attacker $i$ has the incentive to manipulate $v(\datasubset)$; thus, in this case, attacker $i$ positively/negatively augments their dataset $\stoi$ to derive a new dataset $\reportedstoi$ (Lines 4 and 6).
The positive augmentation can be implemented by employing the entire local dataset $D_i$ for retraining, while the negative augmentation can be achieved by perturbing the data $\stoi$ (see \textbf{Appendix~\ref{appendix:overvaluation}} for more details).
Then, the clients $\clientsins$ perform an FL process $\mathcal{A}(\reportedsubset)$ to retrain a surrogate model $\reportedsubsetmodel$, where attacker $i$ misreports their augmented data representations $\reportedrepresentationsi$ as $\realrepresentationsi$, resulting in an enhanced/reduced utility $v(\reportedsubset)$ for subset $\datasubset$ (Lines 9--10).
\textbf{Note that positive augmentation during model retraining is detrimental, as it does not change the final performance of CML (i.e., the utility of the grand model) but merely inflates the contribution of the block subset $\datasubset$.}
After evaluating all the block subsets, the server calculates the \textit{empirical} data values $\empiricaldatavalue_{i,j}$ and $\empiricaldatavalue_{i}$:
\begin{align*}
    & \empiricaldatavalue_{i,j}(\fullset, v) =  \beta_{i,j}(\fullset) \cdot v(\fullset) + \sum_{\datasubset\subset \fullset \setminus \set[D_{i,j}]} \beta_{i,j}(\datasubset) \cdot v(\reportedsubset),\\
    & \empiricaldatavalue_{i}(\fullset, v) = \beta_i(\fullset) \cdot v(\fullset) + \sum_{\datasubset\subset \fullset} \beta_i(\datasubset) \cdot v(\reportedsubset).
\end{align*}

\textbf{Theoretical guarantee.}
Lemma \ref{lem:attack_success} shows that by successfully implementing the attack, attacker $i$ will derive an inflated data value $\empiricaldatavalue_{i}(\fullset, v)$, which exceeds the value $\empiricaldatavalue_{i}(\fullset, v \mid \forall \datasubset \subset \fullset, \reportedstoi =  \stoi)$ obtained when attacker $i$ truthfully using $\reportedstoi = \stoi$ for model retraining.
Meanwhile, the sum of the other clients' data values is non-increasing compared with the truth-telling case, i.e., $\phi_{-i}(\fullset, v) \!\leq \! \phi_{-i}(\fullset, v \!\mid\! \forall \datasubset \subset \fullset, \reportedstoi =  \stoi)$.
Consequently, attacker $i$ receives a higher reward.


\begin{lemma}
\label{lem:attack_success}
    Algorithm \ref{alg:data_overvaluation} ensures that $\empiricaldatavalue_{i}(\fullset, v) \geq \empiricaldatavalue_{i}(\fullset, v \mid \forall \datasubset \subset \fullset, \reportedstoi =  \stoi)$ while $\empiricaldatavalue_{-i}(\fullset, v) \leq \empiricaldatavalue_{-i}(\fullset, v \mid \forall \datasubset \subset \fullset, \reportedstoi =  \stoi)$ under the following condition: If $\forall \datasubset \subset \fullset, i \in \mathbb{N}(\datasubset)$, if $\beta_i(\datasubset) > 0$ and $\beta_{-i}(\datasubset) \leq 0$, we have $v(\reportedstoi\cup \reportedstominusi) > v(\stoi\cup \reportedstominusi)$; if $\beta_i(\datasubset) < 0$ and $\beta_{-i}(\datasubset) \geq 0$, we have $v(\reportedstoi\cup \reportedstominusi) < v(\stoi\cup \reportedstominusi)$; otherwise, $v(\reportedstoi\cup \reportedstominusi) = v(\stoi\cup \reportedstominusi)$.
\end{lemma}

\textbf{Defenses.}
In centralized CML scenarios, the server can defend data overvaluation by simply collecting all local datasets and retraining models based on the collected datasets. 
However, in FL scenarios, this defense is infeasible since the server is not allowed to collect the clients' raw data.

Moreover, since implementing the proposed attack requires knowing the signs of $\beta_i(\datasubset)$ and $\beta_{-i}(\datasubset)$, a possible defense is to prevent the clients from knowing the current subset $\datasubset$ during model retraining.
However, attacker $i$ still can estimate the expected values $\mathbb{E}[\beta_i(\datasubset)]$ and $\mathbb{E}[\beta_{-i}(\datasubset)]$ and successfully perform data overvaluation based on the signs of $\mathbb{E}[\beta_i(\datasubset)]$ and $\mathbb{E}[\beta_{-i}(\datasubset)]$.
See \textbf{Appendices~\ref{appendix:incomplete_knowledge} and~\ref{appendix:exp_incomplete_knowledge}}  for technical details and experimental results.

\section{Truthful Data Valuation for CML}
In this section, we characterize the subclass of data valuation metrics that can prevent the data overvaluation attack and select a special metric from this class, named \textit{Truth-Shapley}. 

\subsection{Characterization of Truthful Data Valuation}
The issue of data overvaluation arises from strategic clients misreporting their data representations through data augmentation, which is highly analogous to the problem of untruthful bidding in auctions. 
Specifically, for each client $i$, if we regard their data $\reportedstoi$ used for model retraining as their bid, which is reported in the form of the data representations, and the empirical data value $\empiricaldatavalue_i$ as their payoff, the problem of preventing data overvaluation can be viewed as ensuring a truthful auction. 

\begin{definition}[Bayesian Incentive Compatibility for Truthful Data Valuation]
\label{def:truthfulness}
    A data valuation metric $\phi$ is Bayesian incentive compatible (BIC) if for any game $(\fullset, v)$, for any client $i$, and for any reported data $\{\reportedstoi\mid \datasubset \subset \fullset, i \in \mathbb{N}(\datasubset) \}$, we have
    \begin{align}
    \label{formula:expected_data_value_i}
        &\mathbb{E}_{\reportedstominusi \sim \sigma_i(\cdot \mid \datasubset), \forall \datasubset \subset \fullset}[\empiricaldatavalue_{i}(\fullset, v)] \leq \mathbb{E}_{\reportedstominusi \sim \sigma_i(\cdot \mid \datasubset), \forall \datasubset \subset \fullset}[\empiricaldatavalue_{i}(\fullset, v \!\mid\! \forall \datasubset \!\subset \!\fullset,\! \reportedstoi\! =  \!\stoi)],\\
        \label{formula:expected_data_value_minusi}
        &\mathbb{E}_{\reportedstominusi \sim \sigma_i(\cdot \mid \datasubset), \forall \datasubset \subset \fullset}[\empiricaldatavalue_{-i}(\fullset, v)] \geq \medmath{\mathbb{E}_{\reportedstominusi \sim \sigma_i(\cdot \mid \datasubset), \forall \datasubset \subset \fullset}[\empiricaldatavalue_{-i}(\fullset, v \!\mid \!\forall \datasubset \subset \fullset, \reportedstoi \!= \! \stoi)]},
    \end{align}
    where $\sigma_i(\cdot \mid \datasubset)$ denotes the distribution of $\reportedstominusi$ estimated by client $i$ in their belief.
\end{definition}

Accordingly, we draw on the concept of \textit{Bayesian incentive compatibility} (BIC) from auction theory~\cite{d1982bayesian}, defined in Definition \ref{def:truthfulness}, to ensure truthful data valuation.
Intuitively, BIC ensures that for each client $i$, truthfully reporting their data $\reportedstoi =  \stoi$ is the optimal strategy that not only maximizes their expected data value in Inequality (\ref{formula:expected_data_value_i}) but also minimizes the sum of the other clients' data values in Inequality (\ref{formula:expected_data_value_minusi}). 
Note that the expected data values in Inequalities (\ref{formula:expected_data_value_i})-(\ref{formula:expected_data_value_minusi}) are based on client 
$i$’s prior beliefs $\{\sigma_i(\cdot \mid \datasubset)\}_{\forall \datasubset \subset \fullset}$ about the other clients' reported data. 
This implies that truthful reporting is subjectively optimal based on their beliefs, rather than objectively optimal.
For simplicity, we will omit the conditional subscript of the expectation operator $\mathbb{E}$ where there is no ambiguity.

\begin{assumption}[Subjectively Optimal Data]
\label{ass:optimal_dataset}
    For any data subset $\datasubset \subset \fullset$ with $\stoi = D_i$, and for any $\reportedstoi$, we have $\mathbb{E}_{\reportedstominusi \sim \sigma_i(\cdot \mid \datasubset), \forall \datasubset \subset \fullset} [v(\stoi \cup \reportedstominusi)] \geq \mathbb{E}_{\reportedstominusi \sim \sigma_i(\cdot \mid \datasubset), \forall \datasubset \subset \fullset} [v(\reportedstoi \cup \reportedstominusi)]$.
\end{assumption}

Next, we characterize the subclass of linear data valuation metrics that ensure BIC, based on Assumption \ref{ass:optimal_dataset}.
This assumption means that according to client $i$’s belief, $D_i$ is the dataset known to them that can best optimize the grand model’s utility.
\textbf{Note that it does not require clients to have complete knowledge of their own dataset's potential; it only requires them to submit the data they subjectively believe to be the best for CML.} 
This is because clients are assumed to be rational and self-interested—if they are subjectively aware of better data, they would directly use it in CML to obtain greater rewards, rather than holding it back for data overvaluation.

Subsequently, according to Lemma \ref{lem:data_value_form}, for any linear data valuation metric $\phi$, the expected data values in Inequalities (\ref{formula:expected_data_value_i})-(\ref{formula:expected_data_value_minusi}) can be expressed in the following form:
\begin{align*}
    & \medmath{\mathbb{E}[\empiricaldatavalue_{i}(\fullset, v)] \!=\!  \beta_i(\fullset) \!\cdot\! v(\fullset) \!+\! \sum_{\datasubset\subset \fullset} \!\beta_i(\datasubset) \!\cdot\! \mathbb{E}[v(\reportedstoi \!\cup\! \reportedstominusi)]},\\
    & \medmath{\mathbb{E}[\empiricaldatavalue_{-i}(\fullset, v)] \!=\!  \beta_{-i}(\fullset) \!\cdot\! v(\fullset) \!+\! \sum_{\datasubset\subset \fullset} \!\beta_{-i}(\datasubset) \!\cdot\! \mathbb{E}[v(\reportedstoi \!\cup\! \reportedstominusi)]}.
\end{align*}
Therefore, for any data subset $\datasubset \subset \fullset$, if $\stoi \neq D_i$, we should ensure that $\beta_i(\datasubset) = \beta_{-i}(\datasubset)= 0$ such that varying $\reportedstoi$ has no impact on client $i$'s (the other clients') expected data value $\mathbb{E}[\empiricaldatavalue_{i}(\fullset, v)]$ ($\mathbb{E}[\empiricaldatavalue_{-i}(\fullset, v)]$), thereby ensuring no incentive to report $\reportedstoi \neq \stoi$.
If $\stoi = D_i$, we can set $\beta_i(\datasubset) \geq 0$ and $\beta_{-i}(\datasubset) \leq 0$ to incentivize client $i$ to optimize the utility $\mathbb{E}[v(\widehat{D}_i \cup \reportedstominusi)]$;
because of Assumption \ref{ass:optimal_dataset}, reporting their subjectively optimal data $\widehat{D}_i = D_i$ is their best strategy.
Consequently, we derive the following characterization of linear and BIC data valuation metrics.
\begin{theorem}[Characterization 1]
\label{thm:characterization1}
    Consider a linear data valuation metric $\phi_i(\fullset, v) \coloneqq \sum_{\datasubset \subseteq \fullset} \beta_i(\datasubset) \cdot v(\datasubset)$ where $\beta_i: 2^{\fullset} \rightarrow \mathbb{R}$.
    Under Assumption \ref{ass:optimal_dataset}, $\phi$ satisfies BIC iff $\beta_i$ satisfies that: $\forall \datasubset \subset \fullset$, if $\stoi = D_i$, $\beta_i(\datasubset) \geq 0$ and $\beta_{-i}(\datasubset) \leq 0$;
    otherwise, $\beta_i(\datasubset) = \beta_{-i}(\datasubset) = 0$.
\end{theorem}

\subsection{Truth-Shapley}
Next, we attempt to select a strong member from the subclass of linear and BIC data valuation metrics. 
Our idea is to satisfy the four axioms enjoyed by the SV as much as possible, even though full compliance is impossible since the SV uniquely satisfies all four axioms. 
The first axiom we prioritize is \textit{EFF} (i.e., $\sum_{i\in \clientset}\sum_{j\in [M_i]} \phi_{i,j} = v(\fullset)$), as it ensures that the model utility $v(\fullset)$ is fully attributed to all data blocks.
We write $\clientsubset \subseteq \clientset$ to denote a subset of clients and $D_{\clientsubset}$ to denote the set of data blocks possessed by clients $\clientsubset$, i.e., $D_{\clientsubset} = \cup_{i \in \clientsubset} D_i = \{D_{i,j} \mid i\in \clientsubset, j\in[M_i]\}$.
Consequently, we propose Theorem 4.4 to characterize linear, efficient, and BIC valuation metrics.

\begin{theorem}[Characterization 2]
\label{thm:characterization2}
    Consider a linear, efficient data valuation metric $\phi_i(\fullset, v) \coloneqq \sum_{\datasubset \subseteq \fullset} \beta_i(\datasubset) \cdot v(\datasubset)$ where $\beta_i: 2^{\fullset} \rightarrow \mathbb{R}$.
    Under Assumption \ref{ass:optimal_dataset}, $\phi$ satisfies BIC iff: $\phi_i(\fullset, v) \equiv \sum_{\mathcal{C} \subseteq \mathbb{N}} \beta_i(D_{\mathcal{C}}) \cdot v(D_{\mathcal{C}})$ where $\beta_i(D_{\mathcal{C}}) \geq 0$ for all $ \mathcal{C} \subset \mathbb{N}$ with $i\in \mathcal{C}$.
\end{theorem}

Based on Theorem \ref{thm:characterization2}, we know that under a data valuation metric satisfying LIN, EFF, and BIC, each client's client-level data value $\phi_i$ should be determined only by the utilities $v(D_\mathcal{C})$ of all combinations $D_\mathcal{C}$ of client's full datasets $D_1,\dots,D_N$. 
Accordingly, we propose Truth-Shapley (simply TSV) $\tsv$, which uses an SV-style approach to (1) compute the client-level data value $\tsv_i$ based on the clients' full datasets and then (2) divide $\tsv_i$ among individual data blocks to derive $\tsv_{i,1}, \dots, \tsv_{i, M_i}$.

Specifically, let $\clientleveldatasubset = \{D_i\}_{\forall i \in \mathcal{C}}$ for all $\mathcal{C} \subseteq \mathbb{N}$ and $\mathbb{D}_{-i} = \mathbb{D}_{\mathbb{N}} \setminus \{D_i\}$.
Note that $\clientleveldatasubset$ is mathematically distinct from $D_{\mathcal{C}}$, but it holds the same physical meaning and thus corresponds to the same utility $v(\clientleveldatasubset) = v(D_{\mathcal{C}})$.
Then, we apply the approach of the SV to calculate the client-level TSV:
\begin{align*}
    &\tsv_i(\fullset, v) \coloneqq \sv_i(\clientlevelgranddataset, v) =  \sum_{\mathcal{C} \subseteq \mathbb{N} \setminus \{i\}} \weightsv(\mathcal{C} \mid \mathbb{N})\big(v(\clientleveldatasubset \cup \{D_i\}) - v(\clientleveldatasubset)\big),
\end{align*}
where $\weightsv(\mathcal{C} \mid \mathbb{N}) \coloneqq  \frac{|\mathcal{C}|!(|\mathbb{N}|- |\mathcal{C}| - 1)!}{|\mathbb{N}|!}$.
Next, we employ the SV to calculate the block-level TSV:
\begin{align*}
    &\tsv_{i,j}(\fullset, v) \coloneqq \sv_j(D_i, v^{\tsv_i}) =  \! \sum_{\datasubset \subseteq D_i \setminus \{D_{i,j}\}} \weightsv(\datasubset \!\mid\! D_i)\big(v^{\tsv_i}(\datasubset \cup \{D_{i,j}\}) \!-\! v^{\tsv_i}(\datasubset )\big),
\end{align*}
where $\weightsv(\mathcal{S} \mid D_i) \coloneqq  \frac{|\mathcal{S}|!(|D_i|- |\mathcal{S}| - 1)!}{|D_i|!}$ and $v^{\tsv_i}(\datasubset )\coloneqq \tsv_i(\mathbb{D}_{-i} \cup \{\stoi\}, v)$.
Intuitively, the utility $v^{\tsv_i}(\datasubset)$ represents the client-level TSV $\tsv_i$ when client $i$ contributes dataset $\stoi$.
Consequently, the block-level TSV $\tsv_{i,j}$ measures the expected marginal contribution of the data block $D_{i,j}$ to improving the client-level TSV $\tsv_i$.
Due to the use of the SV-style approach for defining both $\tsv_i$ and $\tsv_{i, j}$, we ensure that Truth-Shapley is linear, efficient, and perfectly complies with the characterization in Theorem \ref{thm:characterization2}.
Therefore, we conclude that it satisfies BIC.

\begin{theorem}
\label{thm:tsv_bic}
    Truth-Shapley $\tsv$ satisfies BIC.
\end{theorem}

Moreover, we find that 
apart from LIN and EFF, Truth-Shapley satisfies the following axioms. 
\textbf{(1)} Client-level dummy actions (DUM-C): If for all $\mathcal{C}\subseteq \mathbb{N} \setminus \set[i]$, we have $v(D_{\mathcal{C}}\cup D_{i}) - v(D_{\mathcal{C}}) = v(D_{i})$, then $\phi_{i}(\fullset, v) = v(D_{i})$. 
\textbf{(2)} Inner block-level dummy actions (DUM-IB): If for all $\mathcal{S}\subseteq D_{i} \setminus \set[D_{i,j}]$, we have $v(\mathcal{S}\cup \set[D_{i,j}]) - v(\mathcal{S}) = v(\set[D_{i,j}])$, then $\phi_{i,j}(\fullset, v) = v(\set[D_{i,j}])$. 
\textbf{(3)} Client-level symmetry (SYM-C): For any two clients $i_1, i_2$, if for any subset of the other clients $\mathcal{C} \subseteq \mathbb{N} \setminus \{i_1, i_2 \}$, we have $v(D_{\mathcal{C}} \cup D_{i_1}) = v(D_{\mathcal{C}} \cup D_{i_2})$, then we have $\phi_{i_1}(\fullset, v) =\phi_{i_2}(\fullset, v)$. 
\textbf{(4)} Inner block-level symmetry (SYM-IB): For any client $i$, if for any two data blocks $D_{i, j_1}, D_{i, j_2}$, and for any subset of their other blocks $\mathcal{S} \subseteq D_i \setminus \{D_{i, j_1}, D_{i, j_2} \}$, we have $v(\mathcal{S} \cup \set[D_{i, j_1}]) = v(\mathcal{S} \cup \set[D_{i, j_2}])$, then we have $\phi_{i, j_1}(\fullset, v) =\phi_{i, j_2}(\fullset, v)$.
DUM-C and SYM-C are client-level adaptations of the SV's axioms DUM and SYM (see definitions in \textbf{Appendix~\ref{appendix:sv})}, tailored for a fair allocation of data value across clients.
Similarly, DUM-IB and SYM-IB are block-level variants of these axioms, designed to ensure a fair distribution of data value among the data blocks of a single client.
More importantly, \textbf{Truth-Shapley uniquely satisfies EFF, LIN, DUM-C, DUM-IB, SYM-C, and SYM-IB, highlighting its distinctiveness among all BIC data valuation metrics.}

\begin{theorem}
\label{thm:tsv_unique}
    Truth-Shapley is the unique data valuation metric that satisfies EFF, LIN, DUM-C, DUM-IB, SYM-C, and SYM-IB.
\end{theorem}

\begin{table*}[t]
\centering
\caption{Reward allocation results under the data overvaluation attack. The percentages report the Symmetric Percentage Change in the attacker $i$'s data value $\empiricaldatavalue_i$ and the other clients' total data value $\empiricaldatavalue_{-i}$ after performing data overvaluation. The client-level valuation error is defined as the normalized MAE between the vectors $[\empiricaldatavalue_1,\dots,\empiricaldatavalue_N]$ and $[\phi_1,\dots,\phi_N]$.}
\label{tab:reward_alloc_full_knowledge}
\resizebox{\textwidth}{!}
{
\begin{tabular}{cl|ccccc|ccccc|ccccc}
\hline
\multicolumn{2}{c|}{\multirow{2}{*}{\textbf{\begin{tabular}[c]{@{}c@{}}CML \\ Setting\end{tabular}}}} & \multicolumn{5}{c|}{\textbf{Change in $\empiricaldatavalue_i$ (\%)}}     & \multicolumn{5}{c|}{\textbf{Change in $\empiricaldatavalue_{-i}$ (\%)}}  & \multicolumn{5}{c}{\textbf{Client-level valuation error}}                    \\ \cline{3-17} 
\multicolumn{2}{c|}{}                                                                                 & \textbf{SV}  & \textbf{TSV} & \textbf{LOO} & \textbf{BSV} & \textbf{BV}  & \textbf{SV}  & \textbf{TSV} & \textbf{LOO} & \textbf{BSV} & \textbf{BV}  & \textbf{SV}   & \textbf{TSV} & \textbf{LOO}  & \textbf{BSV}  & \textbf{BV}   \\ \hline
\multirow{3}{*}{\textit{Bank}}                            & \textit{FedAvg}                           & +82          & 0.0          & +119         & +63          & +48          & -66          & 0.0          & 0.0          & -13          & -31          & 0.83          & 0.0          & 3.18          & 0.54          & 0.44          \\
                                                          & \textit{SplitFed}                         & +89          & 0.0          & +95          & +81          & +73          & -66          & 0.0          & 0.0          & -16          & -38          & 0.85          & 0.0          & 1.79          & 0.67          & 0.43          \\
                                                          & \textit{FedMD}                            & +73          & 0.0          & +6.5         & +66          & +46          & -37          & 0.0          & 0.0          & -13          & -18          & 0.57          & 0.0          & 0.41          & 0.50          & 0.30          \\ \hline
\multirow{3}{*}{\textit{Rent}}                            & \textit{FedAvg}                           & +69          & 0.0          & +177         & +43          & +72          & -78          & 0.0          & 0.0          & -7.9         & -112         & 4.34          & 0.0          & 15.4          & 1.63          & 3.45          \\
                                                          & \textit{SplitFed}                         & +83          & 0.0          & +183         & +52          & +95          & -122         & 0.0          & 0.0          & -10          & -160         & 5.56          & 0.0          & 9.09          & 2.12          & 5.03          \\
                                                          & \textit{FedMD}                            & +80          & 0.0          & +187         & +45          & +79          & -108         & 0.0          & 0.0          & -8.3         & -154         & 6.22          & 0.0          & 17.5          & 1.56          & 4.72          \\ \hline
\multirow{3}{*}{\textit{MNIST}}                           & \textit{FedAvg}                           & +48          & 0.0          & +58          & +63          & +22          & -37          & 0.0          & 0.0          & -14          & -20          & 2.31          & 0.0          & 0.70          & 2.38          & 0.90          \\
                                                          & \textit{SplitFed}                         & +68          & 0.0          & +66          & +77          & +48          & -68          & 0.0          & 0.0          & -19          & -50          & 3.35          & 0.0          & 0.56          & 3.29          & 1.87          \\
                                                          & \textit{FedMD}                            & +53          & 0.0          & +29          & +77          & +21          & -43          & 0.0          & 0.0          & -19          & -20          & 2.49          & 0.0          & 0.87          & 2.14          & 1.12          \\ \hline
\multirow{3}{*}{\textit{FMNIST}}                          & \textit{FedAvg}                           & +107         & 0.0          & +44          & +116         & +22          & -132         & 0.0          & 0.0          & -50          & -130         & 0.85          & 0.0          & 0.36          & 3.35          & 0.28          \\
                                                          & \textit{SplitFed}                         & +120         & 0.0          & +44          & +138         & +102         & -175         & 0.0          & 0.0          & -102         & -159         & 1.72          & 0.0          & 0.38          & 4.75          & 0.80          \\
                                                          & \textit{FedMD}                            & +67          & 0.0          & +86          & +123         & +30          & -71          & 0.0          & 0.0          & -63          & -31          & 3.04          & 0.0          & 1.99          & 3.57          & 1.88          \\ \hline
\multirow{3}{*}{\textit{CIFAR10}}                         & \textit{FedAvg}                           & +68          & 0.0          & +62          & +75          & +33          & -123         & 0.0          & +0.0         & -46          & -87          & 1.00          & 0.0          & 0.36          & 2.82          & 0.39          \\
                                                          & \textit{SplitFed}                         & +109         & 0.0          & +65          & +98          & +104         & -142         & 0.0          & +0.0         & -46          & -131         & 1.51          & 0.0          & 0.38          & 3.83          & 0.77          \\
                                                          & \textit{FedMD}                            & +62          & 0.0          & +107         & +98          & +28          & -70          & 0.0          & 0.0          & -37          & -33          & 1.33          & 0.0          & 0.80          & 3.60          & 0.53          \\ \hline
\multirow{3}{*}{\textit{AGNews}}                          & \textit{FedAvg}                           & +63          & 0.0          & +153         & +79          & +35          & -60          & 0.0          & 0.0          & -20          & -34          & 10.4          & 0.0          & 7.23          & 6.85          & 4.67          \\
                                                          & \textit{SplitFed}                         & +73          & 0.0          & +154         & +90          & +46          & -81          & 0.0          & 0.0          & -26          & -53          & 11.8          & 0.0          & 5.53          & 6.15          & 6.55          \\
                                                          & \textit{FedMD}                            & +60          & 0.0          & +58          & +84          & +41          & -55          & 0.0          & 0.0          & -24          & -37          & 5.43          & 0.0          & 0.62          & 2.76          & 5.81          \\ \hline
\multirow{3}{*}{\textit{Yahoo}}                           & \textit{FedAvg}                           & +61          & 0.0          & +110         & +83          & +25          & -78          & 0.0          & 0.0          & -29          & -22          & 2.58          & 0.0          & 5.42          & 2.88          & 0.88          \\
                                                          & \textit{SplitFed}                         & +110         & 0.0          & +110         & +147         & +76          & -196         & 0.0          & 0.0          & -142         & -135         & 6.22          & 0.0          & 4.06          & 9.20          & 3.32          \\
                                                          & \textit{FedMD}                            & +67          & 0.0          & +129         & +92          & +29          & -89          & 0.0          & 0.0          & -37          & -29          & 2.54          & 0.0          & 4.21          & 2.31          & 1.07          \\ \hline
\multicolumn{2}{c|}{\textbf{Average}}                                                                 & \textbf{+77} & \textbf{0.0} & \textbf{+97} & \textbf{+85} & \textbf{+51} & \textbf{-90} & \textbf{0.0} & \textbf{0.0} & \textbf{-35} & \textbf{-71} & \textbf{3.57} & \textbf{0.0} & \textbf{3.85} & \textbf{3.19} & \textbf{2.15} \\ \hline
\end{tabular}
}
\end{table*}

\section{Experiments}
\subsection{Setup}
\textbf{Research questions.}
Our experiments answer the following questions:
How do Truth-Shapley and existing data valuation metrics perform in reward allocation and data selection under data valuation attacks?
In scenarios without attacks, is Truth-Shapley an effective metric for data selection?

\textbf{Baselines.}
We include four SOTA \textbf{linear} data valuation metrics as baselines: the SV, the LOO, Beta Shapley (BSV)~\citep{kwon2022beta}, and Banzhaf value (BV)~\citep{wang2023data}.
Since they satisfy linearity, the data overvaluation attack can be applied to them. 
For BSV, we select the Beta(16, 1) distribution to set up its parameters, which demonstrates the best performance in the original paper.

\textbf{CML settings.}
We consider seven CML tasks, including one tabular data classification task (\textit{Bank}~\citep{bank_marketing_222}), one tabular data regression task (\textit{Rent}~\citep{apartment_for_rent_classified_555}), three image classification tasks (\textit{MNIST}~\citep{lecun1998gradient}, \textit{FMNIST}~\citep{xiao2017fashion}, \textit{CIFAR10}~\citep{krizhevsky2009learning}), and two text classification tasks (AGNews and Yahoo).
For each task, we partition the data into \textbf{non-IID} data blocks and distribute the blocks to three clients to simulate an FL scenario.
Additionally, since clients report different types of data representations under various FL algorithms, we implement three representative FL algorithms to evaluate the performance of data overvaluation under three common data representation settings. 
These algorithms include FedAvg~\citep{mcmahan2017communication}, SplitFed~\citep{thapa2022splitfed}, and FedMD~\citep{li2019fedmd}, where the attacker performs data overvaluation by misreporting local models, embeddings of local data, and predicted logits, respectively.
We conduct 15 repeated experiments, resampling local data for the clients and randomly selecting one client as the attacker in each run.
More details can be found in \textbf{Appendix \ref{appendix:cml_settings}}.

\textbf{Utility metric.}
The utility metric for classification tasks is classification accuracy, while for regression tasks it is R-squared. 
Additional experiments using other utility metrics are reported in \textbf{Appendix~\ref{appendix:exp_utility_metrics}}.


\begin{table*}[t]
\centering
\caption{Data selection results under the data overvaluation attack. The reported values indicate the decline rate in model utility when data blocks are selected based on inflated data values. }
\label{tab:data_sel_decline_full}
\small
\resizebox{\textwidth}{!}
{
\begin{tabular}{cl|ccccccccccccccc}
\hline
\multicolumn{2}{c|}{\multirow{3}{*}{\textbf{\begin{tabular}[c]{@{}c@{}}CML \\ setting\end{tabular}}}} & \multicolumn{15}{c}{\textbf{Decline (\%) in model utility due to data overvaluation}}                                                                                                                                                                                                                                                                                                    \\ \cline{3-17} 
\multicolumn{2}{c|}{}                                                                                 & \multicolumn{5}{c|}{\textbf{Top-k Selection}}                                                     & \multicolumn{5}{c|}{\textbf{Above-average Selection}}                                             & \multicolumn{5}{c}{\textbf{Above-median Selection}}                                                                                                                              \\ \cline{3-17} 
\multicolumn{2}{c|}{}                                                                                 & \textbf{SV}   & \textbf{TSV} & \textbf{LOO}  & \textbf{BSV}  & \multicolumn{1}{c|}{\textbf{BV}}   & \textbf{SV}   & \textbf{TSV} & \textbf{LOO}  & \textbf{BSV}  & \multicolumn{1}{c|}{\textbf{BV}}   & \multicolumn{1}{c}{\textbf{SV}}   & \multicolumn{1}{c}{\textbf{TSV}} & \multicolumn{1}{c}{\textbf{LOO}}  & \multicolumn{1}{c}{\textbf{BSV}}  & \multicolumn{1}{c}{\textbf{BV}}   \\ \hline
\multirow{3}{*}{\textit{Bank}}                            & \textit{FedAvg}                           & 2.29          & 0.0          & 2.61          & 0.19          & \multicolumn{1}{c|}{0.19}          & 0.19          & 0.0          & 1.03          & 0.18          & \multicolumn{1}{c|}{0.02}          & 0.10                              & 0.0                              & 0.35                              & 0.08                              & 0.07                              \\
                                                          & \textit{SplitFed}                         & 1.97          & 0.0          & 1.94          & 0.34          & \multicolumn{1}{c|}{0.14}          & 0.31          & 0.0          & 1.26          & 0.33          & \multicolumn{1}{c|}{0.03}          & 0.35                              & 0.0                              & 0.08                              & 0.39                              & 0.12                              \\
                                                          & \textit{FedMD}                            & 0.95          & 0.0          & 0.40          & 0.17          & \multicolumn{1}{c|}{0.12}          & 0.15          & 0.0          & 0.49          & 0.11          & \multicolumn{1}{c|}{0.04}          & 0.12                              & 0.0                              & 0.21                              & 0.32                              & 0.04                              \\ \hline
\multirow{3}{*}{\textit{Rent}}                            & \textit{FedAvg}                           & 10.3          & 0.0          & 11.0          & 1.08          & \multicolumn{1}{c|}{5.45}          & 2.27          & 0.0          & 1.05          & 1.80          & \multicolumn{1}{c|}{2.59}          & 2.75                              & 0.0                              & 1.28                              & 2.75                              & 1.77                              \\
                                                          & \textit{SplitFed}                         & 10.0          & 0.0          & 11.1          & 0.96          & \multicolumn{1}{c|}{9.20}          & 1.82          & 0.0          & 2.92          & 1.73          & \multicolumn{1}{c|}{1.04}          & 1.69                              & 0.0                              & 1.73                              & 1.69                              & 1.69                              \\
                                                          & \textit{FedMD}                            & 12.8          & 0.0          & 12.7          & 1.16          & \multicolumn{1}{c|}{9.41}          & 3.14          & 0.0          & 4.56          & 2.64          & \multicolumn{1}{c|}{2.52}          & 3.19                              & 0.0                              & 2.20                              & 2.39                              & 3.12                              \\ \hline
\multirow{3}{*}{\textit{MNIST}}                           & \textit{FedAvg}                           & 3.53          & 0.0          & 1.56          & 2.48          & \multicolumn{1}{c|}{2.11}          & 1.45          & 0.0          & 1.95          & 1.13          & \multicolumn{1}{c|}{2.47}          & 2.40                              & 0.0                              & 1.38                              & 0.49                              & 1.14                              \\
                                                          & \textit{SplitFed}                         & 3.23          & 0.0          & 2.39          & 2.52          & \multicolumn{1}{c|}{4.43}          & 0.65          & 0.0          & 5.52          & 1.24          & \multicolumn{1}{c|}{0.52}          & 2.58                              & 0.0                              & 2.75                              & 0.04                              & 3.37                              \\
                                                          & \textit{FedMD}                            & 3.78          & 0.0          & 2.00          & 2.44          & \multicolumn{1}{c|}{1.93}          & 1.52          & 0.0          & 2.52          & 1.34          & \multicolumn{1}{c|}{2.54}          & 5.47                              & 0.0                              & 1.70                              & 2.08                              & 3.05                              \\ \hline
\multirow{3}{*}{\textit{FMNIST}}                          & \textit{FedAvg}                           & 4.60          & 0.0          & 6.34          & 12.6          & \multicolumn{1}{c|}{3.51}          & 0.06          & 0.0          & 0.06          & 21.0          & \multicolumn{1}{c|}{0.06}          & 0.25                              & 0.0                              & 0.08                              & 25.4                              & 0.10                              \\
                                                          & \textit{SplitFed}                         & 14.8          & 0.0          & 6.70          & 12.5          & \multicolumn{1}{c|}{7.80}          & 19.1          & 0.0          & 4.77          & 0.0           & \multicolumn{1}{c|}{4.91}          & 17.4                              & 0.0                              & 0.11                              & 26.2                              & 3.05                              \\
                                                          & \textit{FedMD}                            & 5.34          & 0.0          & 3.55          & 7.56          & \multicolumn{1}{c|}{2.27}          & 1.53          & 0.0          & 1.93          & 0.0           & \multicolumn{1}{c|}{5.29}          & 13.2                              & 0.0                              & 8.33                              & 6.45                              & 1.14                              \\ \hline
\multirow{3}{*}{\textit{CIFAR10}}                         & \textit{FedAvg}                           & 1.06          & 0.0          & 0.32          & 1.77          & \multicolumn{1}{c|}{3.15}          & 1.28          & 0.0          & 0.0           & 0.69          & \multicolumn{1}{c|}{1.60}          & 1.22                              & 0.0                              & 0.0                               & 11.6                              & 0.0                               \\
                                                          & \textit{SplitFed}                         & 6.83          & 0.0          & 0.22          & 2.19          & \multicolumn{1}{c|}{3.25}          & 12.7          & 0.0          & 0.0           & 0.0           & \multicolumn{1}{c|}{9.79}          & 11.7                              & 0.0                              & 0.0                               & 12.8                              & 1.51                              \\
                                                          & \textit{FedMD}                            & 3.84          & 0.0          & 5.45          & 5.35          & \multicolumn{1}{c|}{0.86}          & 8.62          & 0.0          & 15.8          & 0.0           & \multicolumn{1}{c|}{8.08}          & 8.05                              & 0.0                              & 8.18                              & 11.2                              & 4.56                              \\ \hline
\multirow{3}{*}{\textit{AGNews}}                          & \textit{FedAvg}                           & 0.75          & 0.0          & 0.85          & 0.85          & \multicolumn{1}{c|}{0.54}          & 3.73          & 0.0          & 2.33          & 0.20          & \multicolumn{1}{c|}{2.88}          & 2.16                              & 0.0                              & 0.08                              & 0.79                              & 2.05                              \\
                                                          & \textit{SplitFed}                         & 1.69          & 0.0          & 0.58          & 0.57          & \multicolumn{1}{c|}{1.02}          & 3.63          & 0.0          & 0.23          & 1.03          & \multicolumn{1}{c|}{4.17}          & 4.03                              & 0.0                              & 0.06                              & 1.03                              & 4.44                              \\
                                                          & \textit{FedMD}                            & 1.78          & 0.0          & 1.33          & 1.18          & \multicolumn{1}{c|}{1.26}          & 4.27          & 0.0          & 0.0           & 3.06          & \multicolumn{1}{c|}{6.35}          & 5.01                              & 0.0                              & 0.0                               & 1.99                              & 6.05                              \\ \hline
\multirow{3}{*}{\textit{Yahoo}}                           & \textit{FedAvg}                           & 5.16          & 0.0          & 5.50          & 4.37          & \multicolumn{1}{c|}{2.06}          & 33.5          & 0.0          & 19.7          & 31.1          & \multicolumn{1}{c|}{8.03}          & 5.32                              & 0.0                              & 9.33                              & 5.78                              & 1.63                              \\
                                                          & \textit{SplitFed}                         & 5.88          & 0.0          & 5.93          & 4.47          & \multicolumn{1}{c|}{4.46}          & 35.1          & 0.0          & 25.8          & 33.3          & \multicolumn{1}{c|}{32.9}          & 6.81                              & 0.0                              & 6.63                              & 5.46                              & 6.09                              \\
                                                          & \textit{FedMD}                            & 4.81          & 0.0          & 4.51          & 4.50          & \multicolumn{1}{c|}{4.07}          & 38.7          & 0.0          & 41.4          & 29.4          & \multicolumn{1}{c|}{14.9}          & 5.55                              & 0.0                              & 6.26                              & 5.25                              & 5.56                              \\ \hline
\multicolumn{2}{c|}{\textbf{Average}}                                                                 & \textbf{5.02} & \textbf{0.0} & \textbf{4.14} & \textbf{3.30} & \multicolumn{1}{c|}{\textbf{3.20}} & \textbf{8.27} & \textbf{0.0} & \textbf{6.35} & \textbf{6.20} & \multicolumn{1}{c|}{\textbf{5.27}} & \multicolumn{1}{c}{\textbf{4.73}} & \multicolumn{1}{c}{\textbf{0.0}} & \multicolumn{1}{c}{\textbf{2.42}} & \multicolumn{1}{c}{\textbf{5.91}} & \multicolumn{1}{c}{\textbf{2.41}} \\ \hline
\end{tabular}
}
\end{table*}

\begin{table*}[t]
\centering
\caption{Data selection results without data overvaluation.}
\label{tab:data_sel_full}
\small
\resizebox{\textwidth}{!}
{
\begin{tabular}{cl|ccccccccccccccc}
\hline
\multicolumn{2}{c|}{\multirow{3}{*}{\textbf{\begin{tabular}[c]{@{}c@{}}CML \\ setting\end{tabular}}}} & \multicolumn{15}{c}{\textbf{Model utility without data overvaluation}}                                                                                                                                                                                                                  \\ \cline{3-17} 
\multicolumn{2}{c|}{}                                                                                 & \multicolumn{5}{c|}{\textbf{Top-k Selection}}                                                      & \multicolumn{5}{c|}{\textbf{Above-average Selection}}                                              & \multicolumn{5}{c}{\textbf{Above-median Selection}}                           \\ \cline{3-17} 
\multicolumn{2}{c|}{}                                                                                 & \textbf{SV}   & \textbf{TSV}  & \textbf{LOO}  & \textbf{BSV}  & \multicolumn{1}{c|}{\textbf{BV}}   & \textbf{SV}   & \textbf{TSV}  & \textbf{LOO}  & \textbf{BSV}  & \multicolumn{1}{c|}{\textbf{BV}}   & \textbf{SV}   & \textbf{TSV}  & \textbf{LOO}  & \textbf{BSV}  & \textbf{BV}   \\ \hline
\multirow{3}{*}{\textit{Bank}}                            & \textit{FedAvg}                           & 82.1          & 82.0          & 82.0          & 82.1          & \multicolumn{1}{c|}{82.1}          & 82.7          & 82.7          & 82.4          & 82.7          & \multicolumn{1}{c|}{82.7}          & 82.7          & 82.6          & 82.9          & 82.7          & 82.7          \\
                                                          & \textit{SplitFed}                         & 82.1          & 82.1          & 82.0          & 82.1          & \multicolumn{1}{c|}{82.0}          & 82.5          & 82.6          & 82.5          & 82.6          & \multicolumn{1}{c|}{82.5}          & 82.7          & 82.6          & 82.6          & 82.7          & 82.8          \\
                                                          & \textit{FedMD}                            & 82.2          & 82.2          & 82.0          & 82.2          & \multicolumn{1}{c|}{82.2}          & 82.7          & 82.5          & 82.8          & 82.7          & \multicolumn{1}{c|}{82.4}          & 82.9          & 82.8          & 82.9          & 82.9          & 82.9          \\ \hline
\multirow{3}{*}{\textit{Rent}}                            & \textit{FedAvg}                           & 83.8          & 83.0          & 84.4          & 83.7          & \multicolumn{1}{c|}{83.8}          & 84.1          & 83.7          & 82.3          & 84.8          & \multicolumn{1}{c|}{83.7}          & 84.5          & 83.5          & 85.8          & 84.5          & 84.4          \\
                                                          & \textit{SplitFed}                         & 83.8          & 83.1          & 85.6          & 83.8          & \multicolumn{1}{c|}{83.5}          & 84.0          & 83.3          & 85.6          & 84.9          & \multicolumn{1}{c|}{82.8}          & 84.7          & 84.0          & 86.8          & 84.7          & 84.7          \\
                                                          & \textit{FedMD}                            & 86.5          & 86.1          & 87.0          & 87.0          & \multicolumn{1}{c|}{86.1}          & 89.4          & 88.8          & 87.3          & 89.4          & \multicolumn{1}{c|}{89.2}          & 89.4          & 89.4          & 89.5          & 88.4          & 89.4          \\ \hline
\multirow{3}{*}{\textit{MNIST}}                           & \textit{FedAvg}                           & 76.7          & 75.2          & 75.9          & 76.2          & \multicolumn{1}{c|}{77.4}          & 77.9          & 67.7          & 71.2          & 78.5          & \multicolumn{1}{c|}{78.9}          & 82.8          & 80.6          & 83.5          & 76.3          & 82.2          \\
                                                          & \textit{SplitFed}                         & 76.5          & 75.0          & 77.0          & 74.5          & \multicolumn{1}{c|}{78.3}          & 78.6          & 64.4          & 76.0          & 78.4          & \multicolumn{1}{c|}{77.6}          & 81.2          & 79.2          & 83.5          & 74.0          & 82.9          \\
                                                          & \textit{FedMD}                            & 77.5          & 78.3          & 80.2          & 77.8          & \multicolumn{1}{c|}{77.9}          & 84.1          & 78.0          & 82.0          & 83.7          & \multicolumn{1}{c|}{81.5}          & 84.1          & 85.8          & 85.6          & 82.2          & 85.0          \\ \hline
\multirow{3}{*}{\textit{FMNIST}}                          & \textit{FedAvg}                           & 52.9          & 52.0          & 49.2          & 46.5          & \multicolumn{1}{c|}{53.8}          & 52.4          & 52.4          & 52.4          & 39.8          & \multicolumn{1}{c|}{52.4}          & 52.5          & 52.5          & 52.5          & 53.4          & 52.5          \\
                                                          & \textit{SplitFed}                         & 52.1          & 49.9          & 49.2          & 46.3          & \multicolumn{1}{c|}{54.0}          & 52.4          & 52.4          & 53.3          & 39.2          & \multicolumn{1}{c|}{52.4}          & 52.3          & 52.3          & 52.3          & 53.6          & 52.2          \\
                                                          & \textit{FedMD}                            & 59.0          & 57.0          & 54.7          & 59.2          & \multicolumn{1}{c|}{55.8}          & 66.7          & 68.0          & 42.0          & 63.4          & \multicolumn{1}{c|}{65.9}          & 68.9          & 68.9          & 62.5          & 59.2          & 62.2          \\ \hline
\multirow{3}{*}{\textit{CIFAR10}}                         & \textit{FedAvg}                           & 30.0          & 30.4          & 25.4          & 25.0          & \multicolumn{1}{c|}{31.9}          & 37.1          & 34.5          & 27.3          & 18.1          & \multicolumn{1}{c|}{42.2}          & 32.4          & 32.4          & 27.3          & 22.3          & 31.9          \\
                                                          & \textit{SplitFed}                         & 29.3          & 30.5          & 25.2          & 25.1          & \multicolumn{1}{c|}{32.0}          & 36.7          & 36.7          & 27.2          & 18.0          & \multicolumn{1}{c|}{36.7}          & 32.0          & 32.0          & 27.2          & 22.9          & 32.0          \\
                                                          & \textit{FedMD}                            & 39.4          & 39.6          & 37.9          & 38.8          & \multicolumn{1}{c|}{40.6}          & 38.5          & 39.6          & 42.1          & 30.3          & \multicolumn{1}{c|}{40.1}          & 41.9          & 43.4          & 43.9          & 39.2          & 44.3          \\ \hline
\multirow{3}{*}{\textit{AGNews}}                          & \textit{FedAvg}                           & 69.5          & 73.3          & 69.3          & 69.1          & \multicolumn{1}{c|}{70.0}          & 78.1          & 80.3          & 70.5          & 66.9          & \multicolumn{1}{c|}{77.9}          & 79.6          & 80.9          & 80.1          & 77.3          & 80.2          \\
                                                          & \textit{SplitFed}                         & 70.6          & 75.3          & 69.3          & 69.9          & \multicolumn{1}{c|}{69.6}          & 77.8          & 81.4          & 64.2          & 75.4          & \multicolumn{1}{c|}{79.1}          & 80.3          & 81.6          & 80.3          & 76.8          & 80.7          \\
                                                          & \textit{FedMD}                            & 73.6          & 75.0          & 72.1          & 71.7          & \multicolumn{1}{c|}{71.1}          & 81.8          & 80.4          & 64.9          & 72.9          & \multicolumn{1}{c|}{82.6}          & 80.9          & 81.3          & 82.1          & 81.5          & 82.2          \\ \hline
\multirow{3}{*}{\textit{Yahoo}}                           & \textit{FedAvg}                           & 43.1          & 43.0          & 43.1          & 43.1          & \multicolumn{1}{c|}{43.1}          & 50.3          & 49.5          & 41.7          & 46.9          & \multicolumn{1}{c|}{50.0}          & 44.6          & 44.6          & 45.5          & 44.5          & 44.6          \\
                                                          & \textit{SplitFed}                         & 43.2          & 43.1          & 43.3          & 43.0          & \multicolumn{1}{c|}{43.2}          & 48.2          & 46.4          & 44.4          & 47.2          & \multicolumn{1}{c|}{46.7}          & 45.0          & 44.4          & 44.6          & 44.6          & 44.6          \\
                                                          & \textit{FedMD}                            & 44.0          & 44.1          & 43.8          & 44.0          & \multicolumn{1}{c|}{44.1}          & 51.2          & 52.4          & 53.5          & 45.5          & \multicolumn{1}{c|}{49.9}          & 45.8          & 46.0          & 45.6          & 45.5          & 45.8          \\ \hline
\multicolumn{2}{c|}{\textbf{Average}}                                                                 & \textbf{63.7} & \textbf{63.8} & \textbf{62.8} & \textbf{62.4} & \multicolumn{1}{c|}{\textbf{63.9}} & \textbf{67.5} & \textbf{66.1} & \textbf{62.6} & \textbf{62.4} & \multicolumn{1}{c|}{\textbf{67.5}} & \textbf{67.2} & \textbf{67.2} & \textbf{67.0} & \textbf{64.7} & \textbf{67.2} \\ \hline
\end{tabular}
}
\end{table*}

\subsection{Main Results}
\label{sec:reward_alloc}

\textbf{Reward allocation.}
Table~\ref{tab:reward_alloc_full_knowledge} presents the performance of the data overvaluation attack against various data valuation metrics in the reward allocation task. 
Since the reward of attacker $i$ depends on both their own data value $\empiricaldatavalue_i$ and the other clients' data values $\empiricaldatavalue_{-i}$, we measure the impact of the data overvaluation attack on reward allocation by assessing the changes in both $\empiricaldatavalue_i$ and $\empiricaldatavalue_{-i}$ relative to the truthful values.
Additionally, we compute the \textit{client-level valuation error}, which measures the absolute change in each client's data value, to evaluate the robustness of different data valuation metrics against data overvaluation.

As shown in Table \ref{tab:reward_alloc_full_knowledge}, Truth-Shapley is completely robust against the data overvaluation attack, meaning that such attacks have no impact on client-level TSVs. 
In contrast, the attack significantly affects other data valuation metrics, leading to a substantial increase in the attacker’s data value. 
For SV, BSV, and BV, data overvaluation even results in a notable decrease in the other clients' data values, benefiting the attacker at their expense. 
Moreover, the attack takes effect across all three FL algorithms, indicating that attackers can misreport the three representative types of data presentation to overvalue their local data.

\textbf{Data selection.}
Tables~\ref{tab:data_sel_decline_full} and~\ref{tab:data_sel_full} report the performance of the data overvaluation attack against various data valuation metrics in the data selection task.
Three data selection strategies are considered: Top-k Selection, Above-average Selection, and Above-median Selection, which select data blocks with the highest $k$ data values, data values above the average, and data values above the median, respectively.
For the Top-k Selection, $k$ is randomly selected for times, and the average result is reported.

As shown in Table~\ref{tab:data_sel_decline_full}, when the data overvaluation attack occurs, all data valuation metrics except Truth-Shapley suffer a significant decline in model utility. 
This is because the attacker’s data becomes overvalued, thereby distorting the ranking of the block-level data values and reducing their effectiveness in guiding data selection.
Moreover, Table~\ref{tab:data_sel_full} shows that even without data overvaluation, the three data selection strategies based on Truth-Shapley still demonstrate the best or near-best performance. 
This indicates that Truth-Shapley is not only highly robust against the data overvaluation attack but also serves as an effective metric for data selection in non-adversarial settings.

\section{Related Work}

Most of existing studies designed data valuation methods for CML based on two valuation metrics: LOO~\cite{cook1977detection} and the SV~\cite{shapley1953value}. 
Since computing these metrics requires evaluating the model utilities for a large number of data subsets, substantial research efforts have been devoted to improving the efficiency of the computation. 
Their approaches include downsampling data subsets~\cite{ghorbani2019data, jia2019efficient, jia2019towards, luo2024fast, luo2022shapley, lin2022measuring, kwon2021efficient}, designing training-free utility functions~\cite{wang2024helpful, pruthi2020estimating, koh2017understanding}, and approximating retrained models~\cite{wu2022davinz, just2023lava, nohyun2022data}.

Another line of research focuses on enhancing the robustness and reliability of data valuation. 
Xu et al.~\cite{xu2021validation} designed a new utility function that is more robust to clients’ data replication behavior. 
Lin et al.~\cite{lin2024distributionally} provided a validation-free utility function for clients without a joinly-agreed validation dataset. 
Some studies~\cite{schoch2022cs, xu2024model, xia2024p} have designed utility functions that capture a model’s predictive capability at a finer granularity than prediction accuracy. 
Tian et al.~\cite{tian2024derdava} and Xia et al.~\cite{xia2023equitable} proposed methods to accelerate recomputing data values in machine unlearning scenarios. 
Zheng et al.~\cite{zheng2023secure} and Wang et al.~\cite{wang2024privacy} proposed methods to ensure privacy and security in data valuation.
Wang et al.~\cite{wang2023data} introduced the Banzhaf value as a data valuation metric, which is robust to the randomness of model retraining. 
Kwon et al.~\cite{kwon2022beta} extended the SV to Beta Shapley, improving the detection of noisy data points.
\textit{Our work reveals a new vulnerability in data valuation, i.e., data overvaluation, and proposes Truth-Shapley to enhance robustness/reliability against data overvaluation.}

A research question that also addresses data misreporting but focuses on a significantly different scenario is truthful data acquisition~\cite{chen2020truthful}.
Their focus is on designing data pricing mechanisms to incentivize sellers to provide their true and accurate datasets, whereas our focus is on preventing sellers from inflating the value of their data during the data valuation stage.

\section{Conclusion}
This paper introduces the first data overvaluation attack in CML scenarios.
We characterized the subclass of linear and BIC data valuation metrics that can resist this attack and selected Truth-Shapley from the subclass as a promising solution to truthful data valuation. 
Through both theoretical analysis and empirical experiments, we demonstrated the vulnerability of existing linear data valuation metrics to data overvaluation and the robustness and effectiveness of Truth-Shapley. 
Research opportunities in this direction (see \textbf{Appendix~\ref{sec:discuss}}) include implementing data overvaluation in other CML scenarios, such as vertical FL, and designing defense mechanisms for those vulnerable data valuation metrics.


\bibliography{ref}
\bibliographystyle{plainnat}

\appendix
\onecolumn

\section*{Broader Impact}
For the industry, this paper identifies a new attack method that poses a trust crisis for data valuation in FL. 
For the academia, this paper opens up a new research direction: truthful data valuation for FL.

\section{Technical Details}
\label{appendix:technical_details}

\subsection{Axioms of Shapley Value}
\label{appendix:sv}
The SV is considered an ideal solution to data valuation because it has been proven be to the unique valuation metric that satisfies the following axioms~\citep{shapley1953value}. \textbf{(1)} Linearity (LIN): The server can linearly combine the data values evaluated on any two utility metrics $v_1$ and $v_2$, i.e., $\phi_{i,j}(\fullset, v_1+v_2) = \phi_{i,j}(\fullset, v_1)  + \phi_{i,j}(\fullset, v_2)$. 
\textbf{(2)} Efficiency (EFF): The sum of all data blocks' data values equals the utility improved by the grand dataset $\fullset$, i.e., $\sum_{i\in \clientset} \sum_{j\in[M_i]} \phi_{i,j}(\fullset, v) = v(\fullset)$.
\textbf{(3)} Dummy actions (DUM): If a data block $D_{i,j}$ does not have any synergy with the other blocks, its data value $\sv_{i,j}$ equals the utility $v(D_{i,j})$ of the model trained only on itself.
That is, if for all $\datasubset\subseteq \fullset \setminus \set[D_{i,j}]$, we have $v(\datasubset\cup D_{i, j}) - v(\datasubset) = v(D_{i, j})$, then $\phi_{i,j}(\fullset, v) = v(D_{i,j})$.
\textbf{(4)} Symmetry (SYM): If two data blocks have the same effect on the model utility, they should obtain the same block-level data values.
In other words, for two data blocks $D_{i_1,j_1}, D_{i_2,j_2} \in \fullset$, if for any subset of data blocks $\datasubset \subseteq \fullset \setminus \set[D_{i_1,j_1}, D_{i_2,j_2}]$, we have $v(\datasubset \cup D_{i_1,j_1}) = v(\datasubset \cup D_{i_2,j_2})$, then we have $\phi_{i_1,j_1}(\fullset, v) =\phi_{i_2, j_2}(\fullset, v)$.

\begin{theorem}[Uniqueness of SV~\citep{shapley1953value}]
\label{thm:sv_unique}
    The SV $\sv$ is the unique data valuation metric that satisfies DUM, SYM, LIN, and EFF.
\end{theorem}

\subsection{Examples of Reward Functions}
\label{appendix:reward}

In the literature, most existing works on data valuation and data marketplaces (e.g.,~\citep{agarwal2019marketplace, song2019profit, ohrimenko2019collaborative}) adopt the normalized data value $\frac{\phi_{i, j}}{\sum_{i'\in \clientset, j' \in [M_{i'}]} \phi_{i', j'}}$ to allocate rewards among data blocks. Specifically, the reward function is defined as $R_{i, j}(\phi) = \frac{\phi_{i, j}}{\sum_{i'\in \clientset, j' \in [M_{i'}]} \phi_{i', j'}} \cdot C$, where $C$ denotes the total reward to be distributed, such as the revenue generated from monetizing the grand model.
An alternative approach~\citep{nguyen2022trade} is the reward function $R_{i, j}(\phi) = |\fullset| \cdot \phi_{i,j} - \sum_{i'\in \clientset, j' \in [M_{i'}]} \phi_{i', j'}$, which ensures reward balance, i.e., $\sum_{i \in \clientset, j \in [M_i]} R_{i,j} = 0$. Under this scheme, if the data value $\phi_{i, j}$ exceeds the average value $\frac{\sum_{i'\in \clientset, j' \in [M_{i'}]} \phi_{i', j'}}{|\fullset|}$, client $i$ receives a positive reward for contributing $D_{i,j}$. Conversely, if $\phi_{i, j}$ falls below the average, client $i$ incurs a negative reward (i.e., a payment), as $D_{i,j}$ is considered of low value.
We refer to these two reward functions as the proportional reward and the balanced reward, respectively, and provide experiment results on them in Appendix~\ref{appendix:exp_reward}.

\subsection{Data Overvaluation with Incomplete Knowledge}
\label{appendix:incomplete_knowledge}
When retraining a model on a block subset $\datasubset \subset \datasubset$, the server has to inform the attacker $i$ the blocks $\stoi = o \in \datasubset$ they should retrain data representations on.  
Then, attacker $i$ can estimate the expected value of $\beta_i(\datasubset)$ by computing $\mathbb{E}_{\datasubset, \stoi = o}[\beta_i(\datasubset)] = Pr[\datasubset = s \mid \stoi = o] \cdot \beta_i(s)$, as well as the expected value of $\beta_{-i}(\datasubset)$ by computing $\mathbb{E}_{\datasubset, \stoi = o}[\beta_{-i}(\datasubset)] = Pr[\datasubset = s \mid \stoi = o] \cdot \beta_{-i}(s)$, where $Pr[\datasubset = s \mid \stoi = o] \cdot \beta_i(s)$ is the probability distribution of $\datasubset$ when $\stoi$ is given.
Then, in Algorithm~\ref{alg:data_overvaluation}, if $\mathbb{E}_{\datasubset, \stoi = o}[\beta_i(\datasubset)] > 0$ and $\mathbb{E}_{\datasubset, \stoi = o}[\beta_i(\datasubset)] \leq 0$, attacker $i$ positively augments $\stoi$ to obtain a dataset $\reportedstoi$;
if $\mathbb{E}_{\datasubset, \stoi = o}[\beta_i(\datasubset)] < 0$ and $\mathbb{E}_{\datasubset, \stoi = o}[\beta_i(\datasubset)] \geq 0$, attacker $i$ negatively augments $\stoi$ to derive $\reportedstoi$.
The augmented dataset $\reportedstoi$ is then used for model retraining.
Both the reward functions satisfy that $R_{i}(\phi)$ increases with $\phi_{i}$ while decreasing with $\phi_{-i}$.

\section{Experimental Implementation Details}
\label{appendix:exp_details}

\subsection{Computational Resources}
\label{appendix:compute_resources}
Our experiments were conducted on a workstation equipped with two NVIDIA RTX 4090 GPUs, an Intel Core i9-13900K CPU, and 192 GB RAM. 
Since computing data values requires extensive model retraining, we parallelized the model retraining of different block subsets using multiprocessing.
Since all linear data valuation metrics are functions of model utilities $\{v(\datasubset)\}_{\datasubset \subseteq \fullset}$, we can first compute and record the utilities of all block subsets, and then compute all data valuation metrics together to reduce the computation time.

\subsection{CML Settings}
\label{appendix:cml_settings}
Tables~\ref{tab:common_settings},~\ref{tab:specific_settings},~\ref{tab:data_blocks} and~\ref{tab:model} describe our CML setting for different datasets. 
For each dataset, we randomly sample training, validation, and distillation subsets from the raw datasets. 
The training set is further divided into 9 to 11 data blocks, which are then allocated to three clients. 
The data blocks are partitioned based on one of the following strategies: grouping by the value of a specific attribute, grouping by labels, or clustering via Principal Component Analysis (PCA)~\citep{pearson1901liii}. 
The specific partitioning details are as follows.
\begin{itemize}[leftmargin=*]
    \item \textbf{Attribute-based division:} For the Bank and Rent datasets, we partition the training data based on the attribute "job" and the attribute "state", respectively, so that different clients possess bank marketing data associated with different jobs and housing rental data from different regions. Each data block may contain multiple states for the Rent dataset.
    \item \textbf{PCA clustering:} For the MNIST dataset, we apply PCA clustering to partition the training data into nine groups, with each group forming a data block.
    \item \textbf{Label-based division:} For the FMNIST, CIFAR10, AGNews, and Yahoo datasets, we partition the data into data blocks based on their labels. Each data block may contain multiple classes, but the samples are non-overlapping across blocks.
\end{itemize}

\begin{table}[h]
\caption{Common settings used for all datasets.}
\label{tab:common_settings}
\resizebox{\textwidth}{!}
{
\begin{tabular}{l|ccccccc}
\hline
 & \textbf{Train size} & \textbf{Validation size} & \textbf{Distillation size} & \textbf{FL rounds $T$} & \textbf{Batch size} & \textbf{Learning rate} & \textbf{Optimizer} \\ \hline
\textbf{All datasets}     & 8000                & 2000                     & 2000                       & 3                      & 64                  & 0.001                 & Adam               \\ \hline
\end{tabular}
}
\end{table}

\begin{table}[h]
\caption{Dataset-specific settings used in our experiments.}
\label{tab:specific_settings}
\resizebox{\textwidth}{!}
{
\begin{tabular}{l|cccccc}
\hline
\textbf{Dataset} & \textbf{Task}                                                    & \textbf{Data blocks} & \textbf{Block division}    & \textbf{Model architecture}                                            & \textbf{Local epochs} & \textbf{Loss function}                                  \\ \hline
\textbf{Bank}    & \begin{tabular}[c]{@{}c@{}}Tabular\\ classification\end{tabular} & 11                   & Attribute-based ("job")  & MLP                                                                    & 3                     & \begin{tabular}[c]{@{}c@{}}Cross\\ Entropy\end{tabular} \\
\textbf{Rent}    & \begin{tabular}[c]{@{}c@{}}Tabular\\ regression\end{tabular}     & 9                    & Attribute-based ("state") & MLP                                                                    & 10                    & MSE                                                     \\
\textbf{MNIST}   & \begin{tabular}[c]{@{}c@{}}Image\\ classification\end{tabular}   & 9                    & PCA clustering             & CNN                                                                    & 5                     & \begin{tabular}[c]{@{}c@{}}Cross\\ Entropy\end{tabular} \\
\textbf{FMNIST}  & \begin{tabular}[c]{@{}c@{}}Image\\ classification\end{tabular}   & 9                    & Label-based            & CNN                                                                    & 10                    & \begin{tabular}[c]{@{}c@{}}Cross\\ Entropy\end{tabular} \\
\textbf{CIFAR10} & \begin{tabular}[c]{@{}c@{}}Image\\ classification\end{tabular}   & 9                    & Label-based            & CNN                                                                    & 10                    & \begin{tabular}[c]{@{}c@{}}Cross\\ Entropy\end{tabular} \\
\textbf{AGNews}  & \begin{tabular}[c]{@{}c@{}}Text\\ classification\end{tabular}    & 9                    & Label-based            & \begin{tabular}[c]{@{}c@{}}Transformer-based\\ classifier\end{tabular} & 5                     & \begin{tabular}[c]{@{}c@{}}Cross\\ Entropy\end{tabular} \\
\textbf{Yahoo}   & \begin{tabular}[c]{@{}c@{}}Text\\ classification\end{tabular}    & 10                   & Label-based            & \begin{tabular}[c]{@{}c@{}}Transformer-based\\ classifier\end{tabular} & 10                    & \begin{tabular}[c]{@{}c@{}}Cross\\ Entropy\end{tabular} \\ \hline
\end{tabular}
}
\end{table}

\begin{table}[h]
\caption{Distribution of data blocks across three clients.}
\label{tab:data_blocks}
\resizebox{\textwidth}{!}
{
\begin{tabular}{|c|clll|cll|clll|}
\hline
\multirow{2}{*}{\textbf{\begin{tabular}[c]{@{}c@{}}Dataset\\ (Division Method)\end{tabular}}} & \multicolumn{4}{c|}{\multirow{2}{*}{\textbf{Blocks of Client 1}}}                                                                                                                                                         & \multicolumn{3}{c|}{\multirow{2}{*}{\textbf{Blocks of Client 2}}}                                                                                                                                                                     & \multicolumn{4}{c|}{\multirow{2}{*}{\textbf{Blocks of Client 3}}}                                                                                                                        \\
                                                                                              & \multicolumn{4}{c|}{}                                                                                                                                                                                                     & \multicolumn{3}{c|}{}                                                                                                                                                                                                                 & \multicolumn{4}{c|}{}                                                                                                                                                                    \\ \hline
\textbf{\begin{tabular}[c]{@{}c@{}}Bank\\ (Attribute-based)\end{tabular}}                     & \multicolumn{4}{c|}{\begin{tabular}[c]{@{}c@{}}Block({[}"management"{]}),\\ Block({[}"entrepreneur"{]}),\\ Block({[}"admin."{]}),\\ Block({[}"self-employed"{]})\end{tabular}}                                            & \multicolumn{3}{c|}{\begin{tabular}[c]{@{}c@{}}Block({[}"technician"{]}),\\ Block({[}"blue-collar"{]}),\\ Block({[}"services"{]})\end{tabular}}                                                                                       & \multicolumn{4}{c|}{\begin{tabular}[c]{@{}c@{}}Block({[}"retired"{]}),\\ Block({[}"unemployed"{]}),\\ Block({[}"housemaid"{]}),\\ Block({[}"student"{]})\end{tabular}}                   \\ \hline
\textbf{\begin{tabular}[c]{@{}c@{}}Rent\\ (Attribute-based)\end{tabular}}                     & \multicolumn{4}{c|}{\begin{tabular}[c]{@{}c@{}}Block({[}"ME", "NH", "VT",\\ "MA", "RI", "CT"{]}),\\ Block({[}"NY", "NJ", "PA", \\ "DE", "MD", "DC"{]}),\\ Block({[}"VA", "WV", "NC",\\ "SC", "GA", "FL"{]})\end{tabular}} & \multicolumn{3}{c|}{\begin{tabular}[c]{@{}c@{}}Block({[}"OH", "IN", "IL",\\ "MI", "WI"{]}),\\ Block({[}"MN", "IA", "MO",\\ "ND", "SD", "NE", "KS"{]}),\\ Block({[}"KY", "TN",\\ "AL", "MS", "AR",\\ "LA", "OK", "TX"{]}\end{tabular}} & \multicolumn{4}{c|}{\begin{tabular}[c]{@{}c@{}}Block({[}"MT", "WY", "ID"{]}),\\ Block({[}"CO", "UT",\\ "NV", "NM", "AZ"{]}),\\ Block({[}"WA", "OR",\\ "CA", "AK", "HI"{]})\end{tabular}} \\ \hline
\textbf{\begin{tabular}[c]{@{}c@{}}MNIST\\ (PCA Clustering)\end{tabular}}                     & \multicolumn{4}{c|}{3 clusters}                                                                                                                                                                                           & \multicolumn{3}{c|}{3 clusters}                                                                                                                                                                                                       & \multicolumn{4}{c|}{3 clusters}                                                                                                                                                          \\ \hline
\textbf{\begin{tabular}[c]{@{}c@{}}FMNIST\\ (Label-based)\end{tabular}}                       & \multicolumn{4}{c|}{\begin{tabular}[c]{@{}c@{}}Block({[}0, 2, 3, 4, 6{]}),\\ Block({[}1{]}),\\ Block({[}5, 9{]})\end{tabular}}                                                                                            & \multicolumn{3}{c|}{\begin{tabular}[c]{@{}c@{}}Block({[}0, 3, 4, 6{]}),\\ Block({[}1{]}),\\ Block({[}5, 8, 9{]})\end{tabular}}                                                                                                        & \multicolumn{4}{c|}{\begin{tabular}[c]{@{}c@{}}Block({[}0, 2{]}),\\ Block({[}1{]}),\\ Block({[}7, 8{]})\end{tabular}}                                                                    \\ \hline
\textbf{\begin{tabular}[c]{@{}c@{}}CIFAR10\\ (Label-based)\end{tabular}}                      & \multicolumn{4}{c|}{\begin{tabular}[c]{@{}c@{}}Block({[}0, 8{]}),\\ Block({[}1, 9{]}),\\ Block({[}2, 3, 4, 5, 6{]})\end{tabular}}                                                                                         & \multicolumn{3}{c|}{\begin{tabular}[c]{@{}c@{}}Block({[}0, 8{]}),\\ Block({[}1, 9{]}),\\ Block({[}2, 3, 4, 5, 7{]})\end{tabular}}                                                                                                     & \multicolumn{4}{c|}{\begin{tabular}[c]{@{}c@{}}Block({[}0, 8{]}),\\ Block({[}1, 9{]}),\\ {[}2, 3, 5, 6, 7{]}\end{tabular}}                                                               \\ \hline
\textbf{\begin{tabular}[c]{@{}c@{}}AGNews\\ (Label-based)\end{tabular}}                       & \multicolumn{4}{c|}{\begin{tabular}[c]{@{}c@{}}Block({[}0{]}),\\ Block({[}2{]}),\\ Block({[}3{]})\end{tabular}}                                                                                                           & \multicolumn{3}{c|}{\begin{tabular}[c]{@{}c@{}}Block({[}0{]}),\\ Block({[}1{]}),\\ Block({[}3{]})\end{tabular}}                                                                                                                       & \multicolumn{4}{c|}{\begin{tabular}[c]{@{}c@{}}Block({[}0{]}),\\ Block({[}2{]}),\\ Block({[}3{]})\end{tabular}}                                                                          \\ \hline
\textbf{\begin{tabular}[c]{@{}c@{}}Yahoo\\ (Label-based)\end{tabular}}                        & \multicolumn{4}{c|}{\begin{tabular}[c]{@{}c@{}}Block({[}0{]}),\\ Block({[}3{]}),\\ Block({[}9{]})\end{tabular}}                                                                                                           & \multicolumn{3}{c|}{\begin{tabular}[c]{@{}c@{}}Block({[}1{]}),\\ Block({[}4{]}),\\ Block({[}6{]})\end{tabular}}                                                                                                                       & \multicolumn{4}{c|}{\begin{tabular}[c]{@{}c@{}}Block({[}2{]}),\\ Block({[}5{]}),\\ Block({[}7{]}),\\ Block({[}8{]})\end{tabular}}                                                        \\ \hline
\end{tabular}
}
\end{table}

\begin{table}[h]
\caption{Model architectures used in our experiments. The implementations are based on PyTorch 2.5.1. We adopt the Sentence Transformer named "all-MinLM-L6-v2" from the sentence-transformers Python package~\cite{reimers-2019-sentence-bert} to implement the transformer-based classifier. Note that in our implementations, the clients fix the Sentence Transformer and only update the Linear layer through FL.}
\label{tab:model}
\resizebox{\textwidth}{!}
{
\begin{tabular}{l|ccc}
\hline
   & \textbf{MLP}                                                                                                                      & \textbf{CNN}                                                                                                                                                                                                                                                             & \textbf{\begin{tabular}[c]{@{}c@{}}Transformer-based\\ classifier\end{tabular}}                        \\ \hline
\textbf{Architecture} & \begin{tabular}[c]{@{}c@{}}Linear(input\_size, 32),\\ ReLU(),\\ Linear(32, 32),\\ ReLU(),\\ Linear(32, output\_size)\end{tabular} & \begin{tabular}[c]{@{}c@{}}Conv2d(in\_channels, 32, kernel\_size=3, padding=1),\\ ReLU(),\\ MaxPool2d(2, 2)\\ Conv2d(32, 64, kernel\_size=3, padding=1),\\ ReLU(),\\ MaxPool2d(2, 2)\\ Flatten(),\\ Linear(4$\times$input\_size/in\_channels, output\_size)\end{tabular} & \begin{tabular}[c]{@{}c@{}}SentenceTransformer(),\\ Linear(embedding\_size, output\_size)\end{tabular} \\ \hline
\end{tabular}
}
\end{table}

\subsection{FL ALgorithms and Implementations of Data Overvaluation}
\label{appendix:overvaluation}
In our experiments, we consider three representative FL algorithms—FedAvg, SplitFed, and FedMD—in which data overvaluation can be achieved by misreporting three key types of data representations in FL: local models, embeddings of training data, and predicted logits over distillation data, respectively.
We introduce the three algorithms and the corresponding implementations of the data overvaluation attack as follows.
\begin{itemize}[leftmargin=*]
    \item \textbf{FedAvg (see Algorithm~\ref{alg:positive_aug_fedavg_fedmd}):} In each training round of FedAvg, the clients train the global model using their local data to derive their new local models.
    Then, the server collects the new local models and aggregates them into a new global model. 
    Therefore, when performing the data overvaluation attack, for each block subset $\datasubset \subset \fullset$, attacker $i$ can use the entire dataset $D_i$ for positive augmentation, i.e., $\reportedstoi = D_i$, and flips the features of the real data $\stoi$ for negative augmentation.
    Since the data representations are misreported in the form of local models, it is hard for the server to observe the augmentation.
    \item \textbf{SplitFed (see Algorithm~\ref{alg:positive_aug_splitfed}):} In SplitFed, a global model is split into two parts: the embedding sub-model resides on the client side, while the classification sub-model stays on the server side. 
    In each training round, each client generates embeddings of their local data using the local embedding sub-model and uploads them to the server. 
    The server then inputs these embeddings into the classification sub-model for forward propagation and backpropagates the gradients to the client side to update both the embedding and classification sub-models. 
    At the end of each round, all clients' local models are aggregated into a new global model.
    For data overvaluation, the attacker $i$ still can perform negative augmentation by flipping the features of the real data $\stoi$. 
    However, in the case of positive augmentation, using the entire dataset $D_i$, which contains more samples than the real data $\stoi$ in the size of samples, leads to an inconsistent number of embeddings uploaded to the server, thereby facilitating detection of the attack. 
    To address this issue, attacker $i$ can randomly redivide their dataset $D_i$ into $M_i$ data blocks $\widetilde{D}_i = \{\widetilde{D}_{i, 1}, \dots, \widetilde{D}_{i, M_i}\}$ such that the new block $|\widetilde{D}_{i, j}|$ contains the same number of samples as the original block $|D_{i, j}|$ for all $j\in[M_i]$.
    Then, attacker $i$ can use $\reportedstoi = \widetilde{D}^{\datasubset}_i$ for positive augmentation, where $\widetilde{D}^{\datasubset}_i$ denotes the new data blocks corresponding to $\stoi$.
    Since the randomly redivided data $\widetilde{D}^{\datasubset}_i$ approximates $D_i$ in distribution, it typically yields better model utility.
    \item \textbf{FedMD (see Algorithm~\ref{alg:positive_aug_fedavg_fedmd}):} In each training round of FedMD, each client uses its local model to perform inference on the public distillation dataset. The server then collects the predicted logits, averages them to construct a consensus dataset, and broadcasts it to the clients for updating their local models.
    Subsequently, each client updates its local model using both the consensus dataset and its own local dataset.
    The implementation of positive and negative augmentation in FedMD is identical to that in FedAvg.
    Notably, when performing data overvaluation in FedMD, attacker $i$ trains the local model on the augmented dataset to improve the accuracy of the predicted logits.
    As a result, the use of the augmented dataset remains difficult for the server to detect.
\end{itemize}

\begin{algorithm}[t]
    \caption{Data Augmentation for Data Overvaluation in FedAvg and FedMD}
    \label{alg:positive_aug_fedavg_fedmd}
   \KwIn{block subset $\datasubset$}
   \If{$\beta_i(\datasubset) > 0$ and $\beta_{-i}(\datasubset) \leq 0$ (positive augmentation)}{
        Attacker $i$: Use the whole local dataset as the positively augmented dataset, i.e, $\reportedstoi = D_i$\; 
   }
    \ElseIf{$\beta_i(\datasubset) < 0$ and $\beta_{-i}(\datasubset) \geq 0$ (negative augmentation)}{
        Attacker $i$: Flip the features of $\stoi$ to generate the negatively augmented dataset $\reportedstoi$\;
    }
\end{algorithm}

\begin{algorithm}[t]
    \caption{Data Augmentation for Data Overvaluation in SplitFed}
    \label{alg:positive_aug_splitfed}
   \KwIn{block subset $\datasubset$}
   \If{$\beta_i(\datasubset) > 0$ and $\beta_{-i}(\datasubset) \leq 0$ (positive augmentation)}{
        Attacker $i$: Randomly redivide the local dataset $D_i$ into $M_i$ data blocks $\widetilde{D}_i = \{\widetilde{D}_{i, 1}, \dots, \widetilde{D}_{i, M_i}\}$ such that $|\widetilde{D}_{i, j}|=|D_{i, j}|$ for all $j\in[M_i]$\;
        Attacker $i$: Use $\widetilde{D}^{\datasubset}_i$ as the positively augmented dataset, i.e., $\reportedstoi = \widetilde{D}^{\datasubset}_i$\;
   }
    \ElseIf{$\beta_i(\datasubset) < 0$ and $\beta_{-i}(\datasubset) \geq 0$ (negative augmentation)}{
        Attacker $i$: Flip the features of $\stoi$ to generate the negatively augmented dataset $\reportedstoi$\;
    }
\end{algorithm}

\section{Limitations and Opportunities}
\label{sec:discuss}

\textbf{Data overvaluation in vertical federated learning.}
The data overvaluation attack proposed in this paper, along with Truth-Shapley, is also applicable to vertical federated learning (VFL), where different data blocks possess distinct attributes or feature spaces in VFL. 
In the VFL setting, data overvaluation can also be achieved through the negative augmentation in Algorithm~\ref{alg:data_overvaluation}. 
However, since the input space of a subset model $\Theta_{\datasubset}$ is defined as the feature space of the given block subset $\datasubset$, it is difficult to enhance the model utility $v(\datasubset)$ by positively augmenting the data representations $\realrepresentationsi$ with other blocks. 
This poses a challenge for implementing the positive augmentation in Algorithm~\ref{alg:data_overvaluation}. 
To address this issue, future work may explore the implementation of data overvaluation against VFL.

\textbf{Computational efficiency.}
Similar to computing the SV, computing Truth-Shapley is time-consuming, as it requires $O(2^{N+\max_i M_i})$ times of model retraining. 
Since Truth-Shapley utilizes the SV-style approach to define both its client-level data value and block-level data value, existing techniques for accelerating SV computation can be applied to computing these two levels of data value.
Also, designing more efficient acceleration methods specifically for Truth-Shapley is a promising direction for future research.

\textbf{Poisoning attacker.}
Assumption 4.2 is the core assumption of this paper, which implicitly assumes that the attacker's dataset $D_i$ is not a poisoning dataset. 
This assumption is based on the premise that client $i$ aims to maximize their reward and thus will not poison the grand model $\mathcal{A}(\fullset)$, as doing so would reduce the reward derived from monetizing/utilizing $\mathcal{A}(\fullset)$.
However, in certain scenarios, client $i$ may pursue dual objectives: both attacking the grand model and conducting data overvaluation. Addressing this dual-objective scenario requires further exploration.
Note that in privacy-preserving CML scenarios~\cite{zheng2022fl, zheng2023secure}, if clients use differential privacy techniques~\cite{dwork2006calibrating, ge2024privbench} to add noise to their local datasets, it does not necessarily mean they are poisoning attackers, as their goal remains to maximize their reward while preserving privacy.


\section{Proofs}
\label{appendix:proofs}

\begin{proof}[Proof of Lemma \ref{lem:data_value_form}]
    For every data subset $\datasubset \subseteq \fullset$, we define the basis game $\delta_\datasubset(T) =$ by 
    \begin{align*}
        \delta_\datasubset(T) =
        \begin{cases}
            1, & T=S, \\
            0, & T\neq S.
   \end{cases}
\end{align*}
The set $\{\delta_\datasubset \mid \datasubset \subseteq \fullset\}$ is a natural basis for $\mathcal{G}(\fullset)$.
Then, consider a linear data valuation metric $\phi(\fullset, v)$.
For any two utility functions $v_1$ and $v_2$ and scalars $w_1, w_2 \in \mathbb{R}$, we have:
\begin{align*}
    \phi_{i, j}(\fullset, w_1v_1+w_2v_2) = w_1\phi_{i,j}(\fullset, v_1) + w_2\phi_{i,j}(\fullset, v_2).
\end{align*}
Therefore, $\phi_{i,j}(\fullset, \cdot): \mathcal{G}(\fullset) \rightarrow \mathbb{R}$ is a linear functional on the vector space $\mathcal{G}(\fullset)$.
Then, because any utility function $v\in \mathcal{G}(\fullset)$ can be written as $v=\sum_{\datasubset \subseteq \fullset} v(\datasubset)\cdot \delta_\datasubset$, we have:
\begin{align*}
    \phi_{i,j}(\fullset, v) = \phi_{i,j}(\fullset, \sum_{\datasubset \subseteq \fullset} v(\datasubset)\cdot \delta_\datasubset) = \sum_{\datasubset \subseteq \fullset} v(\datasubset)\cdot \phi_{i,j}(\fullset, \delta_\datasubset).
\end{align*}
Defining $\beta_{i,j}(\datasubset) \coloneqq \phi_{i,j}(\fullset, \delta_\datasubset)$ and $\beta_{i}(\datasubset) \coloneqq \sum_{j\in[M_i]}\beta_{i,j}(\datasubset)$, we conclude the proof.

\end{proof}

\begin{proof}[Proof of Lemma \ref{lem:attack_success}]
According to Definition \ref{def:overvaluation}, under a data overvaluation attack, we have
\begin{align*}
    & \empiricaldatavalue_{i}(\fullset, v) = \beta_i(\fullset) \cdot v(\fullset) + \sum_{\datasubset\subset \fullset} \beta_i(\datasubset) \cdot v(\reportedsubset) \\
    = & \beta_i(\fullset) \cdot v(\fullset) + \sum_{\datasubset\subset \fullset, \beta_i(\mathcal{S}) > 0, \beta_{-i}(\mathcal{S}) \leq 0} \beta_i(\datasubset) \cdot v(\reportedsubset) + \sum_{\datasubset\subset \fullset, \beta_i(\mathcal{S}) \leq 0, \beta_{-i}(\mathcal{S}) > 0} \beta_i(\datasubset) \cdot v(\reportedsubset) \\
   & + \sum_{\datasubset\subset \fullset, \beta_i(\mathcal{S}) > 0, \beta_{-i}(\mathcal{S}) > 0} \beta_i(\datasubset) \cdot v(\reportedsubset) + \sum_{\datasubset\subset \fullset, \beta_i(\mathcal{S}) \leq 0, \beta_{-i}(\mathcal{S}) \leq 0} \beta_i(\datasubset) \cdot v(\reportedsubset)\\
   \geq & \beta_i(\fullset) \cdot v(\fullset) + \sum_{\datasubset\subset \fullset, \beta_i(\mathcal{S}) > 0, \beta_{-i}(\mathcal{S}) \leq 0} \beta_i(\datasubset) \cdot v(\stoi\cup \reportedstominusi) + \sum_{\datasubset\subset \fullset, \beta_i(\mathcal{S}) \leq 0, \beta_{-i}(\mathcal{S}) > 0} \beta_i(\datasubset) \cdot v(\stoi\cup \reportedstominusi) \\
   & + \sum_{\datasubset\subset \fullset, \beta_i(\mathcal{S}) > 0, \beta_{-i}(\mathcal{S}) > 0} \beta_i(\datasubset) \cdot v(\stoi\cup \reportedstominusi) + \sum_{\datasubset\subset \fullset, \beta_i(\mathcal{S}) \leq 0, \beta_{-i}(\mathcal{S}) \leq 0} \beta_i(\datasubset) \cdot v(\stoi\cup \reportedstominusi)\\
   = & \beta_i(\fullset) \cdot v(\fullset) + \sum_{\datasubset\subset \fullset} \beta_i(\datasubset) \cdot v(\stoi \cup \reportedstominusi) = \empiricaldatavalue_{i}(\fullset, v \mid \forall \datasubset \subset \fullset, \reportedstoi =  \stoi).
\end{align*}
Similarly, under a data overvaluation attack, we have
\begin{align*}
    & \empiricaldatavalue_{-i}(\fullset, v) = \beta_{-i}(\fullset) \cdot v(\fullset) + \sum_{\datasubset\subset \fullset} \beta_{-i}(\datasubset) \cdot v(\reportedsubset) \\
    = & \beta_{-i}(\fullset) \cdot v(\fullset) + \sum_{\datasubset\subset \fullset, \beta_i(\mathcal{S}) > 0, \beta_{-i}(\mathcal{S}) \leq 0} \beta_{-i}(\datasubset) \cdot v(\reportedsubset) + \sum_{\datasubset\subset \fullset, \beta_i(\mathcal{S}) \leq 0, \beta_{-i}(\mathcal{S}) > 0} \beta_{-i}(\datasubset) \cdot v(\reportedsubset) \\
   & + \sum_{\datasubset\subset \fullset, \beta_{i}(\mathcal{S}) > 0, \beta_{-i}(\mathcal{S}) > 0} \beta_{-i}(\datasubset) \cdot v(\reportedsubset) + \sum_{\datasubset\subset \fullset, \beta_i(\mathcal{S}) \leq 0, \beta_{-i}(\mathcal{S}) \leq 0} \beta_{-i}(\datasubset) \cdot v(\reportedsubset)\\
   \leq & \beta_{-i}(\fullset) \cdot v(\fullset) + \sum_{\datasubset\subset \fullset, \beta_i(\mathcal{S}) > 0, \beta_{-i}(\mathcal{S}) \leq 0} \beta_{-i}(\datasubset) \cdot v(\stoi\cup \reportedstominusi) + \sum_{\datasubset\subset \fullset, \beta_i(\mathcal{S}) \leq 0, \beta_{-i}(\mathcal{S}) > 0} \beta_{-i}(\datasubset) \cdot v(\stoi\cup \reportedstominusi) \\
   & + \sum_{\datasubset\subset \fullset, \beta_i(\mathcal{S}) > 0, \beta_{-i}(\mathcal{S}) > 0} \beta_{-i}(\datasubset) \cdot v(\stoi\cup \reportedstominusi) + \sum_{\datasubset\subset \fullset, \beta_i(\mathcal{S}) \leq 0, \beta_{-i}(\mathcal{S}) \leq 0} \beta_{-i}(\datasubset) \cdot v(\stoi\cup \reportedstominusi)\\
   = & \beta_{-i}(\fullset) \cdot v(\fullset) + \sum_{\datasubset\subset \fullset} \beta_{-i}(\datasubset) \cdot v(\stoi \cup \reportedstominusi) = \empiricaldatavalue_{-i}(\fullset, v \mid \forall \datasubset \subset \fullset, \reportedstoi =  \stoi).
\end{align*}
\end{proof}

\begin{proof}[Proof of Theorem \ref{thm:characterization1}]
    $\Rightarrow$: Under Assumption \ref{ass:optimal_dataset}, for any game $(\fullset, v)$, for any client $i$, and for any reported data subsets $\{\reportedstoi\mid \datasubset \subset \fullset, i \in \mathbb{N}(\datasubset) \}$, we have
    \begin{align*}
        &\mathbb{E}[\empiricaldatavalue_{i}(\fullset, v)] = \beta_i(\fullset) \cdot v(\fullset) +\sum_{\datasubset\subset \fullset} \beta_i(\datasubset) \cdot\! \mathbb{E}[v(\reportedstoi \cup \reportedstominusi)]\\
    = & \beta_i(\fullset) \cdot v(\fullset) +\sum_{\datasubset\subset \fullset, \stoi = D_i} \beta_i(\datasubset) \cdot\! \mathbb{E}[v(\reportedstoi \cup \reportedstominusi)] +\sum_{\datasubset\subset \fullset, \stoi \neq D_i} \beta_i(\datasubset) \cdot\! \mathbb{E}[v(\reportedstoi \cup \reportedstominusi)] \\
    = & \beta_i(\fullset) \cdot v(\fullset) +\sum_{\datasubset\subset \fullset, \stoi = D_i} \beta_i(\datasubset) \cdot\! \mathbb{E}[v(\reportedstoi \cup \reportedstominusi)] \\
    \leq &\beta_i(\fullset) \cdot v(\fullset) +\sum_{\datasubset\subset \fullset, \stoi = D_i} \beta_i(\datasubset) \cdot\! \mathbb{E}[v(\stoi \cup \reportedstominusi)] \\
    = &\beta_i(\fullset) \cdot v(\fullset) +\sum_{\datasubset\subset \fullset, \stoi = D_i} \beta_i(\datasubset) \cdot\! \mathbb{E}[v(\stoi \cup \reportedstominusi)] +\sum_{\datasubset\subset \fullset, \stoi \neq D_i} \beta_i(\datasubset) \cdot\! \mathbb{E}[v(\stoi \cup \reportedstominusi)]\\
    = & \mathbb{E}[\empiricaldatavalue_{i}(\fullset, v \mid \forall \datasubset \subset \fullset, \reportedstoi = \stoi)].
    \end{align*}
Similarly, under Assumption \ref{ass:optimal_dataset}, for any game $(\fullset, v)$, for any client $i$, and for any reported data subsets $\{\reportedstoi\mid \datasubset \subset \fullset, i \in \mathbb{N}(\datasubset) \}$, we have
    \begin{align*}
        &\mathbb{E}[\empiricaldatavalue_{-i}(\fullset, v)] = \beta_{-i}(\fullset) \cdot v(\fullset) +\sum_{\datasubset\subset \fullset} \beta_{-i}(\datasubset) \cdot\! \mathbb{E}[v(\reportedstoi \cup \reportedstominusi)]\\
    = & \beta_{-i}(\fullset) \cdot v(\fullset) +\sum_{\datasubset\subset \fullset, \stoi = D_i} \beta_{-i}(\datasubset) \cdot\! \mathbb{E}[v(\reportedstoi \cup \reportedstominusi)] +\sum_{\datasubset\subset \fullset, \stoi \neq D_i} \beta_{-i}(\datasubset) \cdot\! \mathbb{E}[v(\reportedstoi \cup \reportedstominusi)] \\
    = & \beta_{-i}(\fullset) \cdot v(\fullset) +\sum_{\datasubset\subset \fullset, \stoi = D_i} \beta_{-i}(\datasubset) \cdot\! \mathbb{E}[v(\reportedstoi \cup \reportedstominusi)] \\
    \geq &\beta_{-i}(\fullset) \cdot v(\fullset) +\sum_{\datasubset\subset \fullset, \stoi = D_i} \beta_{-i}(\datasubset) \cdot\! \mathbb{E}[v(\stoi \cup \reportedstominusi)] \\
    = &\beta_{-i}(\fullset) \cdot v(\fullset) +\sum_{\datasubset\subset \fullset, \stoi = D_i} \beta_{-i}(\datasubset) \cdot\! \mathbb{E}[v(\stoi \cup \reportedstominusi)] +\sum_{\datasubset\subset \fullset, \stoi \neq D_i} \beta_{-i}(\datasubset) \cdot\! \mathbb{E}[v(\stoi \cup \reportedstominusi)]\\
    = & \mathbb{E}[\empiricaldatavalue_{-i}(\fullset, v \mid \forall \datasubset \subset \fullset, \reportedstoi = \stoi)].
    \end{align*}
$\Leftarrow$: If $\exists \datasubset \subset \fullset$ such that if $\stoi = D_i$, we have $\beta_i(\datasubset) < 0$, or such that if $\stoi \neq D_i$, we have $\beta_i(\datasubset) \neq 0$, we can always construct a utility function $v$ such that $\mathbb{E}[\empiricaldatavalue_{i}(\fullset, v)] > \mathbb{E}[\empiricaldatavalue_{i}(\fullset, v \mid \forall \datasubset \subset \fullset, \reportedstoi = \stoi)]$. 
Also, if $\exists \datasubset \subset \fullset$ such that if $\stoi = D_i$, we have $\beta_{-i}(\datasubset) > 0$, or such that if $\stoi \neq D_i$, we have $\beta_{-i}(\datasubset) \neq 0$, we can always construct a utility function $v$ such that $\mathbb{E}[\empiricaldatavalue_{-i}(\fullset, v)] < \mathbb{E}[\empiricaldatavalue_{-i}(\fullset, v \mid \forall \datasubset \subset \fullset, \reportedstoi = \stoi)]$.
\end{proof}

\begin{proof}[Proof of Theorem \ref{thm:characterization2}]
    $\Rightarrow$: Under Assumption \ref{ass:optimal_dataset}, for any game $(\fullset, v)$, for any client $i$, and for any reported data subsets $\{\reportedstoi\mid \datasubset \subset \fullset, i \in \mathbb{N}(\datasubset) \}$, we have
    \begin{align*}
        &\mathbb{E}[\empiricaldatavalue_{i}(\fullset, v)] = \beta_i(\fullset) \cdot v(\fullset) +\sum_{\mathcal{C}\subset \mathbb{N}} \beta_i(D_{\mathcal{C}}) \cdot\! \mathbb{E}[v(\widehat{D}_{\mathcal{C}})]\\
        = & \beta_i(\fullset) \cdot v(\fullset) +\sum_{\mathcal{C}\subset \mathbb{N}, i\in \mathcal{C}} \beta_i(D_{\mathcal{C}}) \cdot\! \mathbb{E}[v(\widehat{D}_{\mathcal{C}})] +\sum_{\mathcal{C}\subset \mathbb{N}, i\notin \mathcal{C}} \beta_i(D_{\mathcal{C}}) \cdot\! \mathbb{E}[v(\widehat{D}_{\mathcal{C}})] \\
        \leq & \beta_i(\fullset) \cdot v(\fullset) +\sum_{\mathcal{C}\subset \mathbb{N}, i\in \mathcal{C}} \beta_i(D_{\mathcal{C}}) \cdot\! \mathbb{E}[v(D_i \cup (\cup_{i'\in \mathcal{C} \setminus \{i\}} \widehat{D}_{i'}))] +\sum_{\mathcal{C}\subset \mathbb{N}, i\notin \mathcal{C}} \beta_i(D_{\mathcal{C}}) \cdot\! \mathbb{E}[v(\widehat{D}_{\mathcal{C}})]\\
        = & \mathbb{E}[\empiricaldatavalue_{i}(\fullset, v \!\mid\! \forall \datasubset \!\subset \!\fullset,\! \reportedstoi\! =  \!\stoi)]
    \end{align*}
Because $\phi$ is efficient, we have $\beta_{-i}(D_{\mathcal{C}}) = - \beta_{i}(D_{\mathcal{C}})$ for all $\mathcal{C} \subset \mathbb{N}$.
Therefore, under Assumption \ref{ass:optimal_dataset}, for any game $(\fullset, v)$, for any client $i$, and for any reported data subsets $\{\reportedstoi\mid \datasubset \subset \fullset, i \in \mathbb{N}(\datasubset) \}$, we have
    \begin{align*}
        &\mathbb{E}[\empiricaldatavalue_{-i}(\fullset, v)] = \beta_{-i}(\fullset) \cdot v(\fullset) +\sum_{\mathcal{C}\subset \mathbb{N}} \beta_{-i}(D_{\mathcal{C}}) \cdot\! \mathbb{E}[v(\widehat{D}_{\mathcal{C}})]\\
        = & -\beta_i(\fullset) \cdot v(\fullset) -\sum_{\mathcal{C}\subset \mathbb{N}, i\in \mathcal{C}} \beta_i(D_{\mathcal{C}}) \cdot\! \mathbb{E}[v(\widehat{D}_{\mathcal{C}})] -\sum_{\mathcal{C}\subset \mathbb{N}, i\notin \mathcal{C}} \beta_i(D_{\mathcal{C}}) \cdot\! \mathbb{E}[v(\widehat{D}_{\mathcal{C}})] \\
        \geq & -\beta_i(\fullset) \cdot v(\fullset) -\sum_{\mathcal{C}\subset \mathbb{N}, i\in \mathcal{C}} \beta_i(D_{\mathcal{C}}) \cdot\! \mathbb{E}[v(D_i \cup (\cup_{i'\in \mathcal{C} \setminus \{i\}} \widehat{D}_{i'}))] -\sum_{\mathcal{C}\subset \mathbb{N}, i\notin \mathcal{C}} \beta_i(D_{\mathcal{C}}) \cdot\! \mathbb{E}[v(\widehat{D}_{\mathcal{C}})]\\
        = & \beta_{-i}(\fullset) \cdot v(\fullset) +\sum_{\mathcal{C}\subset \mathbb{N}, i\in \mathcal{C}} \beta_{-i}(D_{\mathcal{C}}) \cdot\! \mathbb{E}[v(D_i \cup (\cup_{i'\in \mathcal{C} \setminus \{i\}} \widehat{D}_{i'}))] +\sum_{\mathcal{C}\subset \mathbb{N}, i\notin \mathcal{C}} \beta_{-i}(D_{\mathcal{C}}) \cdot\! \mathbb{E}[v(\widehat{D}_{\mathcal{C}})]\\
        = & \mathbb{E}[\empiricaldatavalue_{-i}(\fullset, v \!\mid\! \forall \datasubset \!\subset \!\fullset,\! \reportedstoi\! =  \!\stoi)]
    \end{align*}
$\Leftarrow$: According to Theorem \ref{thm:characterization1}, if a data valuation metric $\phi$ is linear and BIC, we have
\begin{align*}
    \phi_i(\fullset, v) = \sum_{\datasubset \subseteq \fullset} \beta_i(\datasubset) \cdot v(\datasubset) = \sum_{\datasubset \subseteq \fullset, \stoi = D_i} \beta_i(\datasubset) \cdot v(\datasubset)
\end{align*}
Then, because $\phi$ is efficient, we have 
\begin{align*}
    &\phi_i(\fullset, v) = \sum_{\datasubset \subseteq \fullset, \stoi = D_i} \beta_i(\datasubset) \cdot v(\datasubset) \\
    = & \sum_{\mathcal{C} \subseteq \mathbb{N}} \beta_i(D_{\mathcal{C}}) \cdot v(D_{\mathcal{C}}) + \sum_{\datasubset \subseteq \fullset, \stoi = D_i, \exists i' \in \mathbb{N}(\datasubset)\setminus \{i\}, D^{\datasubset}_{i'} \neq D_{i'}} \beta_i(\datasubset) \cdot v(\datasubset) \\
    = &\sum_{\mathcal{C} \subseteq \mathbb{N}} \beta_i(D_{\mathcal{C}}) \cdot v(D_{\mathcal{C}}) - \sum_{\datasubset \subseteq \fullset, \stoi = D_i, \exists i' \in \mathbb{N}(\datasubset)\setminus \{i\}, D^{\datasubset}_{i'} \neq D_{i'}} \beta_{-i}(\datasubset) \cdot v(\datasubset) \\
    = &\sum_{\mathcal{C} \subseteq \mathbb{N}} \beta_i(D_{\mathcal{C}}) \cdot v(D_{\mathcal{C}}),
\end{align*}
where $\beta_i(D_{\mathcal{C}})$ must be non-negative for all $\forall \mathcal{C} \subset \mathbb{N}$ with $i\in \mathcal{C}$ to ensure BIC.
\end{proof}

\begin{proof}[Proof of Theorem \ref{thm:tsv_bic}]
    The client-level TSV can be written as:
    \begin{align*}
        & \tsv_i(\fullset, v) = \sum_{\mathcal{C} \subseteq \mathbb{N} \setminus \{i\}} \weightsv(\mathcal{C} \mid \mathbb{N})\big(v(\clientleveldatasubset \cup \{D_i\}) - v(\clientleveldatasubset)\big) =  \sum_{\mathcal{C} \subseteq \mathbb{N}} \beta^{TSV}_i(D_\mathcal{C}) \cdot v(D_\mathcal{C}),
    \end{align*}
    where $\beta^{TSV}_i(D_\mathcal{C}) = \begin{cases}
        \weightsv(\mathcal{C} \mid \mathbb{N}), & i\notin \mathcal{C},\\
        - \weightsv(\mathcal{C} \mid \mathbb{N}), & i\in \mathcal{C}.
    \end{cases}$. 
    Because $\beta^{TSV}_i(D_\mathcal{C}) \geq 0$ for all $\mathcal{C} \subset \mathbb{N}$ with $i\in \mathcal{C}$, according to Theorem \ref{thm:characterization2}, $\tsv$ satisfies BIC.
\end{proof}

\begin{proof}[Proof of Theorem \ref{thm:tsv_unique}]
    Because $\tsv_i(\fullset, v) = \sv_i(\clientlevelgranddataset, v)$, according to Theorem \ref{thm:sv_unique}, $\tsv_i$ uniquely satisfies LIN, DUM-C, SYM-C and EFF-C, where EFF-C means $\sum_{i\in \mathbb{N}} \phi_i(\fullset, v) = v(\fullset)$.
    Then, because $\tsv_{i,j}(\fullset, v) = \sv_j(D_i, v^{\tsv_i})$, according to Theorem \ref{thm:sv_unique}, $\tsv_{i,j}$ uniquely satisfies LIN, DUM-IB, SYM-IB and EFF-IB, where EFF-IB means $\forall i\in \mathbb{N}, \sum_{j\in [M_i]} \phi_{i,j}(\fullset, v) = \phi_i(\fullset, v)$.
    Therefore, $\tsv$ uniquely satisfies EFF, LIN, DUM-C, DUM-IB, SYM-C, and SYM-IB.
\end{proof}

\clearpage
\section{Additional Experiments}
\label{appendix:additional_exp}

\subsection{Statistical Significance of Main Results}
\label{appendix:exp_se}

\begin{table*}[h]
\centering
\caption{\textbf{Standard Errors for Results in Table~\ref{tab:reward_alloc_full_knowledge}.}}
\label{tab:reward_alloc_full_knowledge_se}
\resizebox{\textwidth}{!}
{
\begin{tabular}{cl|ccccc|ccccc}
\hline
\multicolumn{2}{c|}{\multirow{2}{*}{\textbf{\begin{tabular}[c]{@{}c@{}}CML \\ Setting\end{tabular}}}} & \multicolumn{5}{c|}{\textbf{Change in $\empiricaldatavalue_i$ (\%)}}   & \multicolumn{5}{c}{\textbf{Change in $\empiricaldatavalue_{-i}$ (\%)}} \\ \cline{3-12} 
\multicolumn{2}{c|}{}                                                                                 & \textbf{SV} & \textbf{TSV} & \textbf{LOO} & \textbf{BSV} & \textbf{BV} & \textbf{SV} & \textbf{TSV} & \textbf{LOO} & \textbf{BSV} & \textbf{BV} \\ \hline
\multirow{3}{*}{\textit{Bank}}                            & \textit{FedAvg}                           & 14.4        & 0.0          & 25.7         & 14.4         & 11.2        & 10.5        & 0.0          & 0.0          & 1.68         & 5.62        \\
                                                          & \textit{SplitFed}                         & 18.0        & 0.0          & 26.7         & 17.5         & 19.9        & 6.17        & 0.0          & 0.0          & 1.39         & 6.83        \\
                                                          & \textit{FedMD}                            & 15.6        & 0.0          & 35.0         & 13.9         & 13.0        & 4.11        & 0.0          & 0.0          & 1.67         & 2.98        \\ \hline
\multirow{3}{*}{\textit{Rent}}                            & \textit{FedAvg}                           & 1.29        & 0.0          & 1.99         & 2.90         & 2.64        & 8.48        & 0.0          & 0.00         & 0.43         & 16.2        \\
                                                          & \textit{SplitFed}                         & 1.29        & 0.0          & 2.33         & 2.93         & 2.78        & 13.9        & 0.0          & 0.0          & 0.27         & 8.89        \\
                                                          & \textit{FedMD}                            & 0.93        & 0.0          & 1.25         & 3.27         & 2.50        & 9.33        & 0.0          & 0.00         & 0.51         & 13.7        \\ \hline
\multirow{3}{*}{\textit{MNIST}}                           & \textit{FedAvg}                           & 2.13        & 0.0          & 10.9         & 3.29         & 1.18        & 1.54        & 0.0          & 0.0          & 0.50         & 1.28        \\
                                                          & \textit{SplitFed}                         & 3.48        & 0.0          & 12.3         & 3.32         & 4.08        & 3.82        & 0.0          & 0.0          & 0.29         & 4.28        \\
                                                          & \textit{FedMD}                            & 2.22        & 0.0          & 22.5         & 4.42         & 0.97        & 2.34        & 0.0          & 0.0          & 1.01         & 1.82        \\ \hline
\multirow{3}{*}{\textit{FMNIST}}                          & \textit{FedAvg}                           & 17.0        & 0.0          & 9.38         & 2.99         & 2.83        & 20.8        & 0.0          & 0.00         & 4.68         & 23.5        \\
                                                          & \textit{SplitFed}                         & 14.6        & 0.0          & 10.1         & 5.23         & 18.1        & 9.12        & 0.0          & 0.00         & 11.1         & 15.1        \\
                                                          & \textit{FedMD}                            & 1.34        & 0.0          & 4.68         & 3.62         & 1.23        & 5.51        & 0.0          & 0.00         & 7.80         & 2.03        \\ \hline
\multirow{3}{*}{\textit{CIFAR10}}                         & \textit{FedAvg}                           & 5.69        & 0.0          & 8.75         & 13.9         & 1.50        & 22.6        & 0.0          & 0.00         & 11.7         & 17.0        \\
                                                          & \textit{SplitFed}                         & 12.2        & 0.0          & 13.2         & 6.43         & 17.7        & 21.2        & 0.0          & 0.00         & 7.36         & 23.1        \\
                                                          & \textit{FedMD}                            & 1.25        & 0.0          & 10.1         & 3.65         & 0.67        & 10.2        & 0.0          & 0.00         & 4.47         & 5.08        \\ \hline
\multirow{3}{*}{\textit{AGNews}}                          & \textit{FedAvg}                           & 0.84        & 0.0          & 3.66         & 1.86         & 0.31        & 1.84        & 0.0          & 0.0          & 0.66         & 1.79        \\
                                                          & \textit{SplitFed}                         & 0.81        & 0.0          & 4.24         & 1.45         & 1.71        & 1.33        & 0.0          & 0.0          & 0.67         & 0.69        \\
                                                          & \textit{FedMD}                            & 1.16        & 0.0          & 14.0         & 3.53         & 0.32        & 2.26        & 0.0          & 0.0          & 1.70         & 1.59        \\ \hline
\multirow{3}{*}{\textit{Yahoo}}                           & \textit{FedAvg}                           & 2.30        & 0.0          & 5.99         & 2.61         & 1.08        & 16.1        & 0.0          & 0.0          & 4.35         & 0.56        \\
                                                          & \textit{SplitFed}                         & 1.84        & 0.0          & 5.82         & 1.41         & 2.40        & 1.88        & 0.0          & 0.00         & 10.7         & 8.17        \\
                                                          & \textit{FedMD}                            & 1.09        & 0.0          & 2.02         & 3.36         & 0.72        & 17.3        & 0.0          & 0.0          & 6.88         & 0.62        \\ \hline
\end{tabular}
}
\end{table*}

\begin{table*}[t]
\centering
\caption{\textbf{Standard Errors for Results in Table~\ref{tab:data_sel_decline_full}.} }
\label{tab:data_sel_decline_full_se}
\small
\resizebox{\textwidth}{!}
{
\begin{tabular}{cl|ccccccccccccccc}
\hline
\multicolumn{2}{c|}{\multirow{3}{*}{\textbf{\begin{tabular}[c]{@{}c@{}}CML \\ setting\end{tabular}}}} & \multicolumn{15}{c}{\textbf{Decline (\%) in model utility due to data overvaluation}}                                                                                                                                                                              \\ \cline{3-17} 
\multicolumn{2}{c|}{}                                                                                 & \multicolumn{5}{c|}{\textbf{Top-k Selection}}                                               & \multicolumn{5}{c|}{\textbf{Above-average Selection}}                                       & \multicolumn{5}{c}{\textbf{Above-median Selection}}                    \\ \cline{3-17} 
\multicolumn{2}{c|}{}                                                                                 & \textbf{SV} & \textbf{TSV} & \textbf{LOO} & \textbf{BSV} & \multicolumn{1}{c|}{\textbf{BV}} & \textbf{SV} & \textbf{TSV} & \textbf{LOO} & \textbf{BSV} & \multicolumn{1}{c|}{\textbf{BV}} & \textbf{SV} & \textbf{TSV} & \textbf{LOO} & \textbf{BSV} & \textbf{BV} \\ \hline
\multirow{3}{*}{\textit{Bank}}                            & \textit{FedAvg}                           & 0.72        & 0.0          & 0.60         & 0.05         & \multicolumn{1}{c|}{0.05}        & 0.19        & 0.0          & 1.03         & 0.18         & \multicolumn{1}{c|}{0.02}        & 0.05        & 0.0          & 0.21         & 0.04         & 0.04        \\
                                                          & \textit{SplitFed}                         & 0.64        & 0.0          & 0.50         & 0.07         & \multicolumn{1}{c|}{0.04}        & 0.10        & 0.0          & 0.49         & 0.09         & \multicolumn{1}{c|}{0.02}        & 0.15        & 0.0          & 0.05         & 0.15         & 0.05        \\
                                                          & \textit{FedMD}                            & 0.46        & 0.0          & 0.10         & 0.06         & \multicolumn{1}{c|}{0.05}        & 0.15        & 0.0          & 0.49         & 0.11         & \multicolumn{1}{c|}{0.04}        & 0.07        & 0.0          & 0.10         & 0.22         & 0.02        \\ \hline
\multirow{3}{*}{\textit{Rent}}                            & \textit{FedAvg}                           & 0.95        & 0.0          & 0.61         & 0.17         & \multicolumn{1}{c|}{1.24}        & 0.79        & 0.0          & 0.56         & 0.70         & \multicolumn{1}{c|}{0.74}        & 0.82        & 0.0          & 0.46         & 0.82         & 0.61        \\
                                                          & \textit{SplitFed}                         & 0.82        & 0.0          & 0.77         & 0.20         & \multicolumn{1}{c|}{1.02}        & 0.59        & 0.0          & 0.99         & 0.62         & \multicolumn{1}{c|}{0.47}        & 0.59        & 0.0          & 0.78         & 0.59         & 0.59        \\
                                                          & \textit{FedMD}                            & 1.27        & 0.0          & 1.35         & 0.14         & \multicolumn{1}{c|}{0.83}        & 0.63        & 0.0          & 1.23         & 0.66         & \multicolumn{1}{c|}{0.57}        & 0.61        & 0.0          & 0.60         & 0.61         & 0.61        \\ \hline
\multirow{3}{*}{\textit{MNIST}}                           & \textit{FedAvg}                           & 0.64        & 0.0          & 0.47         & 0.49         & \multicolumn{1}{c|}{0.52}        & 0.78        & 0.0          & 1.41         & 0.84         & \multicolumn{1}{c|}{1.33}        & 0.92        & 0.0          & 0.61         & 0.49         & 0.51        \\
                                                          & \textit{SplitFed}                         & 0.58        & 0.0          & 0.68         & 0.49         & \multicolumn{1}{c|}{0.99}        & 0.41        & 0.0          & 2.91         & 0.77         & \multicolumn{1}{c|}{0.31}        & 1.18        & 0.0          & 1.19         & 0.04         & 1.10        \\
                                                          & \textit{FedMD}                            & 0.74        & 0.0          & 0.66         & 0.43         & \multicolumn{1}{c|}{0.60}        & 0.86        & 0.0          & 1.15         & 0.62         & \multicolumn{1}{c|}{0.83}        & 1.79        & 0.0          & 1.16         & 0.87         & 1.43        \\ \hline
\multirow{3}{*}{\textit{FMNIST}}                          & \textit{FedAvg}                           & 0.66        & 0.0          & 1.29         & 3.65         & \multicolumn{1}{c|}{0.80}        & 0.04        & 0.0          & 0.04         & 7.95         & \multicolumn{1}{c|}{0.04}        & 0.09        & 0.0          & 0.05         & 8.92         & 0.05        \\
                                                          & \textit{SplitFed}                         & 3.99        & 0.0          & 1.36         & 3.75         & \multicolumn{1}{c|}{2.20}        & 7.22        & 0.0          & 4.08         & 0.0          & \multicolumn{1}{c|}{3.21}        & 6.73        & 0.0          & 0.09         & 9.04         & 2.19        \\
                                                          & \textit{FedMD}                            & 1.02        & 0.0          & 0.97         & 1.37         & \multicolumn{1}{c|}{0.61}        & 0.64        & 0.0          & 1.31         & 0.0          & \multicolumn{1}{c|}{1.76}        & 3.63        & 0.0          & 3.26         & 2.60         & 1.14        \\ \hline
\multirow{3}{*}{\textit{CIFAR10}}                         & \textit{FedAvg}                           & 0.39        & 0.0          & 0.20         & 0.39         & \multicolumn{1}{c|}{1.63}        & 0.93        & 0.0          & 0.0          & 0.52         & \multicolumn{1}{c|}{1.60}        & 0.68        & 0.0          & 0.0          & 2.64         & 0.0         \\
                                                          & \textit{SplitFed}                         & 2.12        & 0.0          & 0.12         & 0.44         & \multicolumn{1}{c|}{0.86}        & 5.28        & 0.0          & 0.0          & 0.0          & \multicolumn{1}{c|}{3.31}        & 4.69        & 0.0          & 0.0          & 2.52         & 0.75        \\
                                                          & \textit{FedMD}                            & 0.94        & 0.0          & 1.43         & 1.42         & \multicolumn{1}{c|}{0.29}        & 3.39        & 0.0          & 4.37         & 0.0          & \multicolumn{1}{c|}{3.44}        & 2.01        & 0.0          & 3.92         & 2.46         & 1.59        \\ \hline
\multirow{3}{*}{\textit{AGNews}}                          & \textit{FedAvg}                           & 0.14        & 0.0          & 0.31         & 0.15         & \multicolumn{1}{c|}{0.14}        & 0.97        & 0.0          & 1.18         & 0.16         & \multicolumn{1}{c|}{0.95}        & 0.72        & 0.0          & 0.04         & 0.49         & 0.78        \\
                                                          & \textit{SplitFed}                         & 0.34        & 0.0          & 0.29         & 0.13         & \multicolumn{1}{c|}{0.22}        & 0.97        & 0.0          & 0.23         & 0.70         & \multicolumn{1}{c|}{1.07}        & 0.88        & 0.0          & 0.04         & 0.64         & 0.90        \\
                                                          & \textit{FedMD}                            & 0.22        & 0.0          & 0.27         & 0.22         & \multicolumn{1}{c|}{0.17}        & 0.93        & 0.0          & 0.0          & 1.24         & \multicolumn{1}{c|}{1.22}        & 1.09        & 0.0          & 0.0          & 1.08         & 1.17        \\ \hline
\multirow{3}{*}{\textit{Yahoo}}                           & \textit{FedAvg}                           & 0.96        & 0.0          & 0.93         & 0.77         & \multicolumn{1}{c|}{0.25}        & 2.34        & 0.0          & 3.15         & 3.59         & \multicolumn{1}{c|}{2.29}        & 1.34        & 0.0          & 1.80         & 1.23         & 0.35        \\
                                                          & \textit{SplitFed}                         & 0.99        & 0.0          & 0.97         & 0.82         & \multicolumn{1}{c|}{0.84}        & 3.05        & 0.0          & 4.01         & 3.48         & \multicolumn{1}{c|}{3.16}        & 1.06        & 0.0          & 1.31         & 1.08         & 1.22        \\
                                                          & \textit{FedMD}                            & 0.85        & 0.0          & 0.81         & 0.81         & \multicolumn{1}{c|}{0.70}        & 2.82        & 0.0          & 2.59         & 3.62         & \multicolumn{1}{c|}{2.88}        & 0.85        & 0.0          & 1.04         & 0.91         & 0.86        \\ \hline
\end{tabular}
}
\end{table*}

\begin{table*}[t]
\centering
\caption{\textbf{Standard Errors for Results in Table~\ref{tab:data_sel_full}.}}
\label{tab:data_sel_full_se}
\small
\resizebox{\textwidth}{!}
{
\begin{tabular}{cl|ccccccccccccccc}
\hline
\multicolumn{2}{c|}{\multirow{3}{*}{\textbf{\begin{tabular}[c]{@{}c@{}}CML \\ setting\end{tabular}}}} & \multicolumn{15}{c}{\textbf{Model utility without data overvaluation}}                                                                                                                                                                                             \\ \cline{3-17} 
\multicolumn{2}{c|}{}                                                                                 & \multicolumn{5}{c|}{\textbf{Top-k Selection}}                                               & \multicolumn{5}{c|}{\textbf{Above-average Selection}}                                       & \multicolumn{5}{c}{\textbf{Above-median Selection}}                    \\ \cline{3-17} 
\multicolumn{2}{c|}{}                                                                                 & \textbf{SV} & \textbf{TSV} & \textbf{LOO} & \textbf{BSV} & \multicolumn{1}{c|}{\textbf{BV}} & \textbf{SV} & \textbf{TSV} & \textbf{LOO} & \textbf{BSV} & \multicolumn{1}{c|}{\textbf{BV}} & \textbf{SV} & \textbf{TSV} & \textbf{LOO} & \textbf{BSV} & \textbf{BV} \\ \hline
\multirow{3}{*}{\textit{Bank}}                            & \textit{FedAvg}                           & 0.14        & 0.17         & 0.17         & 0.14         & \multicolumn{1}{c|}{0.15}        & 0.19        & 0.20         & 0.22         & 0.19         & \multicolumn{1}{c|}{0.19}        & 0.20        & 0.18         & 0.22         & 0.18         & 0.20        \\
                                                          & \textit{SplitFed}                         & 0.28        & 0.25         & 0.27         & 0.28         & \multicolumn{1}{c|}{0.29}        & 0.22        & 0.23         & 0.21         & 0.24         & \multicolumn{1}{c|}{0.22}        & 0.23        & 0.23         & 0.25         & 0.23         & 0.20        \\
                                                          & \textit{FedMD}                            & 0.20        & 0.19         & 0.21         & 0.21         & \multicolumn{1}{c|}{0.19}        & 0.17        & 0.20         & 0.10         & 0.17         & \multicolumn{1}{c|}{0.11}        & 0.18        & 0.13         & 0.17         & 0.19         & 0.18        \\ \hline
\multirow{3}{*}{\textit{Rent}}                            & \textit{FedAvg}                           & 0.49        & 0.35         & 0.35         & 0.51         & \multicolumn{1}{c|}{0.46}        & 0.45        & 0.41         & 0.42         & 0.66         & \multicolumn{1}{c|}{0.41}        & 0.54        & 0.78         & 0.39         & 0.54         & 0.57        \\
                                                          & \textit{SplitFed}                         & 0.53        & 0.50         & 0.30         & 0.56         & \multicolumn{1}{c|}{0.54}        & 0.90        & 0.63         & 0.69         & 0.78         & \multicolumn{1}{c|}{0.77}        & 0.73        & 0.65         & 0.68         & 0.73         & 0.73        \\
                                                          & \textit{FedMD}                            & 0.20        & 0.14         & 0.10         & 0.08         & \multicolumn{1}{c|}{0.18}        & 0.16        & 0.27         & 0.36         & 0.32         & \multicolumn{1}{c|}{0.14}        & 0.16        & 0.16         & 0.10         & 0.34         & 0.09        \\ \hline
\multirow{3}{*}{\textit{MNIST}}                           & \textit{FedAvg}                           & 0.47        & 0.54         & 0.75         & 0.82         & \multicolumn{1}{c|}{0.68}        & 1.99        & 2.05         & 3.32         & 1.93         & \multicolumn{1}{c|}{0.95}        & 0.71        & 1.26         & 0.37         & 1.41         & 0.79        \\
                                                          & \textit{SplitFed}                         & 0.55        & 0.81         & 0.91         & 0.71         & \multicolumn{1}{c|}{0.95}        & 1.59        & 3.77         & 2.18         & 2.09         & \multicolumn{1}{c|}{0.59}        & 1.15        & 1.55         & 0.65         & 0.78         & 0.94        \\
                                                          & \textit{FedMD}                            & 0.91        & 0.39         & 0.33         & 0.95         & \multicolumn{1}{c|}{0.61}        & 1.06        & 1.97         & 2.18         & 1.36         & \multicolumn{1}{c|}{2.08}        & 1.06        & 0.53         & 1.29         & 1.08         & 0.68        \\ \hline
\multirow{3}{*}{\textit{FMNIST}}                          & \textit{FedAvg}                           & 0.38        & 0.72         & 0.24         & 0.25         & \multicolumn{1}{c|}{0.35}        & 0.14        & 0.14         & 0.14         & 0.84         & \multicolumn{1}{c|}{0.14}        & 0.10        & 0.10         & 0.09         & 0.28         & 0.10        \\
                                                          & \textit{SplitFed}                         & 0.30        & 0.70         & 0.25         & 0.37         & \multicolumn{1}{c|}{0.46}        & 0.10        & 0.10         & 0.52         & 1.03         & \multicolumn{1}{c|}{0.10}        & 0.08        & 0.07         & 0.13         & 0.26         & 0.06        \\
                                                          & \textit{FedMD}                            & 0.42        & 0.48         & 0.45         & 0.19         & \multicolumn{1}{c|}{0.48}        & 0.64        & 0.89         & 1.33         & 0.38         & \multicolumn{1}{c|}{0.25}        & 1.34        & 1.34         & 1.32         & 0.17         & 1.45        \\ \hline
\multirow{3}{*}{\textit{CIFAR10}}                         & \textit{FedAvg}                           & 0.68        & 0.32         & 0.19         & 0.06         & \multicolumn{1}{c|}{0.84}        & 1.39        & 1.08         & 0.32         & 0.16         & \multicolumn{1}{c|}{1.44}        & 0.23        & 0.23         & 0.32         & 0.45         & 0.38        \\
                                                          & \textit{SplitFed}                         & 0.52        & 0.34         & 0.15         & 0.26         & \multicolumn{1}{c|}{0.55}        & 1.88        & 1.88         & 0.18         & 0.14         & \multicolumn{1}{c|}{1.88}        & 0.74        & 0.74         & 0.18         & 0.30         & 0.74        \\
                                                          & \textit{FedMD}                            & 0.14        & 0.43         & 0.58         & 0.13         & \multicolumn{1}{c|}{0.20}        & 0.69        & 2.01         & 1.76         & 0.21         & \multicolumn{1}{c|}{1.42}        & 0.92        & 0.96         & 0.99         & 0.41         & 1.10        \\ \hline
\multirow{3}{*}{\textit{AGNews}}                          & \textit{FedAvg}                           & 0.30        & 0.48         & 0.41         & 0.91         & \multicolumn{1}{c|}{0.49}        & 2.13        & 0.37         & 2.27         & 1.38         & \multicolumn{1}{c|}{2.12}        & 0.51        & 0.49         & 0.69         & 0.78         & 0.65        \\
                                                          & \textit{SplitFed}                         & 0.47        & 0.22         & 0.75         & 0.65         & \multicolumn{1}{c|}{0.72}        & 2.39        & 0.83         & 1.01         & 1.34         & \multicolumn{1}{c|}{1.43}        & 0.55        & 0.79         & 0.56         & 1.49         & 0.44        \\
                                                          & \textit{FedMD}                            & 0.21        & 0.38         & 0.29         & 0.70         & \multicolumn{1}{c|}{0.45}        & 0.42        & 0.62         & 0.84         & 2.92         & \multicolumn{1}{c|}{0.30}        & 0.56        & 0.89         & 0.28         & 1.88         & 0.30        \\ \hline
\multirow{3}{*}{\textit{Yahoo}}                           & \textit{FedAvg}                           & 0.10        & 0.11         & 0.20         & 0.14         & \multicolumn{1}{c|}{0.11}        & 0.18        & 0.82         & 0.94         & 1.17         & \multicolumn{1}{c|}{0.82}        & 0.14        & 0.14         & 0.76         & 0.13         & 0.14        \\
                                                          & \textit{SplitFed}                         & 0.21        & 0.20         & 0.28         & 0.23         & \multicolumn{1}{c|}{0.18}        & 0.70        & 0.83         & 1.95         & 1.11         & \multicolumn{1}{c|}{0.76}        & 0.24        & 0.23         & 0.28         & 0.27         & 0.11        \\
                                                          & \textit{FedMD}                            & 0.11        & 0.12         & 0.19         & 0.16         & \multicolumn{1}{c|}{0.11}        & 0.66        & 1.12         & 0.77         & 1.02         & \multicolumn{1}{c|}{0.85}        & 0.08        & 0.06         & 0.10         & 0.11         & 0.08        \\ \hline
\end{tabular}
}
\end{table*}

\clearpage

\subsection{Varying Utility Metrics}
\label{appendix:exp_utility_metrics}
\begin{table*}[h]
\centering
\caption{\textbf{Varying Utility Metrics.} Reward allocation results under the data overvaluation attack, with \textbf{Improvement in Validation Loss} used as the utility metric.}
\resizebox{\textwidth}{!}
{
\begin{tabular}{cl|ccccc|ccccc|ccccc}
\hline
\multicolumn{2}{c|}{\multirow{2}{*}{\textbf{\begin{tabular}[c]{@{}c@{}}CML \\ Setting\end{tabular}}}} & \multicolumn{5}{c|}{\textbf{Change in $\phi_i$ (\%)}}                       & \multicolumn{5}{c|}{\textbf{Change in $\phi_{-i}$ (\%)}}                  & \multicolumn{5}{c}{\textbf{Client-level valuation error}}                    \\ \cline{3-17} 
\multicolumn{2}{c|}{}                                                                                 & \textbf{SV}   & \textbf{TSV} & \textbf{LOO}  & \textbf{BSV}  & \textbf{BV}  & \textbf{SV}   & \textbf{TSV} & \textbf{LOO} & \textbf{BSV} & \textbf{BV}  & \textbf{SV}   & \textbf{TSV} & \textbf{LOO}  & \textbf{BSV}  & \textbf{BV}   \\ \hline
\multirow{3}{*}{\textit{Bank}}                            & \textit{FedAvg}                           & +114          & 0.0          & +178          & +105          & +37          & -140          & 0.0          & 0.0          & -33          & -54          & 1.08          & 0.0          & 4.73          & 0.98          & 0.30          \\
                                                          & \textit{SplitFed}                         & +115          & 0.0          & +171          & +112          & +78          & -148          & 0.0          & 0.0          & -38          & -71          & 1.20          & 0.0          & 4.36          & 1.09          & 0.39          \\
                                                          & \textit{FedMD}                            & +112          & 0.0          & +143          & +121          & +70          & -140          & 0.0          & 0.0          & -52          & -34          & 1.37          & 0.0          & 7.29          & 1.23          & 0.35          \\ \hline
\multirow{3}{*}{\textit{Rent}}                            & \textit{FedAvg}                           & +70           & 0.0          & +177          & +47           & +69          & -80           & 0.0          & 0.0          & -9.0         & -110         & 4.44          & 0.0          & 16.1          & 1.71          & 3.48          \\
                                                          & \textit{SplitFed}                         & +84           & 0.0          & +183          & +56           & +93          & -124          & 0.0          & 0.0          & -12          & -153         & 5.34          & 0.0          & 8.97          & 2.13          & 4.86          \\
                                                          & \textit{FedMD}                            & +82           & 0.0          & +187          & +50           & +78          & -113          & 0.0          & 0.0          & -9.6         & -148         & 6.70          & 0.0          & 16.4          & 1.63          & 4.46          \\ \hline
\multirow{3}{*}{\textit{MNIST}}                           & \textit{FedAvg}                           & +86           & 0.0          & +100          & +188          & +17          & -116          & 0.0          & 0.0          & -76          & -18          & 2.41          & 0.0          & 0.97          & 1.81          & 0.92          \\
                                                          & \textit{SplitFed}                         & +107          & 0.0          & +101          & +189          & +40          & -180          & 0.0          & 0.0          & -80          & -43          & 3.15          & 0.0          & 0.92          & 2.01          & 1.65          \\
                                                          & \textit{FedMD}                            & +72           & 0.0          & +134          & +88           & +35          & -79           & 0.0          & 0.0          & -25          & -40          & 5.16          & 0.0          & 6.21          & 2.81          & 2.46          \\ \hline
\multirow{3}{*}{\textit{FMNIST}}                          & \textit{FedAvg}                           & +155          & 0.0          & +28           & +168          & +18          & -158          & 0.0          & 0.0          & -12          & -25          & 1.62          & 0.0          & 0.42          & 1.81          & 0.61          \\
                                                          & \textit{SplitFed}                         & +163          & 0.0          & +26           & +200          & +44          & -166          & 0.0          & 0.0          & -17          & -52          & 2.61          & 0.0          & 0.53          & 2.74          & 1.36          \\
                                                          & \textit{FedMD}                            & +82           & 0.0          & +92           & +146          & +36          & -113          & 0.0          & 0.0          & -133         & -45          & 3.18          & 0.0          & 3.97          & 4.04          & 1.39          \\ \hline
\multirow{3}{*}{\textit{CIFAR10}}                         & \textit{FedAvg}                           & +140          & 0.0          & +33           & +127          & +12          & -143          & 0.0          & 0.0          & -12          & -20          & 4.19          & 0.0          & 0.17          & 0.95          & 0.55          \\
                                                          & \textit{SplitFed}                         & +186          & 0.0          & +4.2          & +155          & +31          & -174          & 0.0          & 0.0          & -14          & -38          & 5.53          & 0.0          & 0.16          & 1.26          & 1.02          \\
                                                          & \textit{FedMD}                            & +77           & 0.0          & +88           & +125          & +44          & -122          & 0.0          & 0.0          & -90          & -49          & 2.03          & 0.0          & 1.22          & 22.3          & 0.60          \\ \hline
\multirow{3}{*}{\textit{AGNews}}                          & \textit{FedAvg}                           & +68           & 0.0          & +155          & +78           & +40          & -69           & 0.0          & 0.0          & -20          & -40          & 9.45          & 0.0          & 18.3          & 4.36          & 8.08          \\
                                                          & \textit{SplitFed}                         & +79           & 0.0          & +156          & +91           & +52          & -96           & 0.0          & 0.0          & -27          & -62          & 12.0          & 0.0          & 18.5          & 5.20          & 12.5          \\
                                                          & \textit{FedMD}                            & +81           & 0.0          & +135          & +86           & +54          & -103          & 0.0          & 0.0          & -24          & -77          & 25.3          & 0.0          & 7.41          & 6.74          & 29.6          \\ \hline
\multirow{3}{*}{\textit{Yahoo}}                           & \textit{FedAvg}                           & +74           & 0.0          & +144          & +103          & +30          & -96           & 0.0          & +0.0         & -63          & -28          & 2.78          & 0.0          & 6.59          & 1.47          & 1.04          \\
                                                          & \textit{SplitFed}                         & +125          & 0.0          & +144          & +174          & +81          & -200          & 0.0          & 0.0          & -200         & -165         & 6.93          & 0.0          & 6.47          & 7.45          & 3.73          \\
                                                          & \textit{FedMD}                            & +86           & 0.0          & +129          & +96           & +45          & -108          & 0.0          & 0.0          & -43          & -54          & 2.34          & 0.0          & 4.82          & 2.71          & 0.80          \\ \hline
\multicolumn{2}{c|}{\textbf{Average}}                                                                 & \textbf{+103} & \textbf{0.0} & \textbf{+119} & \textbf{+119} & \textbf{+48} & \textbf{-127} & \textbf{0.0} & \textbf{0.0} & \textbf{-47} & \textbf{-63} & \textbf{5.18} & \textbf{0.0} & \textbf{6.41} & \textbf{3.64} & \textbf{3.82} \\ \hline
\end{tabular}
}
\end{table*}

\begin{table*}[h]
\centering
\caption{\textbf{Varying Utility Metrics.} Data selection results under the data overvaluation attack, with \textbf{Improvement in Validation Loss} used as the utility metric.}
\small
\resizebox{\textwidth}{!}
{
\begin{tabular}{cl|ccccccccccccccc}
\hline
\multicolumn{2}{c|}{\multirow{3}{*}{\textbf{\begin{tabular}[c]{@{}c@{}}CML \\ setting\end{tabular}}}} & \multicolumn{15}{c}{\textbf{Decline (\%) in model utility due to data overvaluation}}                                                                                                                                                                                                \\ \cline{3-17} 
\multicolumn{2}{c|}{}                                                                                 & \multicolumn{5}{c|}{\textbf{Top-k Selection}}                                                     & \multicolumn{5}{c|}{\textbf{Above-average Selection}}                                             & \multicolumn{5}{c}{\textbf{Above-median Selection}}                          \\ \cline{3-17} 
\multicolumn{2}{c|}{}                                                                                 & \textbf{SV}   & \textbf{TSV} & \textbf{LOO}  & \textbf{BSV}  & \multicolumn{1}{c|}{\textbf{BV}}   & \textbf{SV}   & \textbf{TSV} & \textbf{LOO}  & \textbf{BSV}  & \multicolumn{1}{c|}{\textbf{BV}}   & \textbf{SV}   & \textbf{TSV} & \textbf{LOO}  & \textbf{BSV}  & \textbf{BV}   \\ \hline
\multirow{3}{*}{\textit{Bank}}                            & \textit{FedAvg}                           & 2.59          & 0.0          & 2.69          & 0.35          & \multicolumn{1}{c|}{0.02}          & 0.57          & 0.0          & 0.31          & 0.64          & \multicolumn{1}{c|}{0.16}          & 0.10          & 0.0          & 0.18          & 0.30          & 0.04          \\
                                                          & \textit{SplitFed}                         & 2.27          & 0.0          & 2.21          & 0.41          & \multicolumn{1}{c|}{0.07}          & 0.52          & 0.0          & 0.21          & 0.53          & \multicolumn{1}{c|}{0.22}          & 0.52          & 0.0          & 0.28          & 0.61          & 0.21          \\
                                                          & \textit{FedMD}                            & 2.38          & 0.0          & 2.63          & 0.40          & \multicolumn{1}{c|}{0.06}          & 0.39          & 0.0          & 1.12          & 0.49          & \multicolumn{1}{c|}{0.22}          & 0.33          & 0.0          & 0.58          & 0.21          & 0.09          \\ \hline
\multirow{3}{*}{\textit{Rent}}                            & \textit{FedAvg}                           & 10.3          & 0.0          & 11.2          & 1.09          & \multicolumn{1}{c|}{5.81}          & 2.69          & 0.0          & 0.42          & 1.17          & \multicolumn{1}{c|}{2.59}          & 2.75          & 0.0          & 1.79          & 2.75          & 1.77          \\
                                                          & \textit{SplitFed}                         & 10.3          & 0.0          & 11.1          & 0.94          & \multicolumn{1}{c|}{8.29}          & 2.15          & 0.0          & 2.65          & 1.73          & \multicolumn{1}{c|}{1.34}          & 1.69          & 0.0          & 1.94          & 1.69          & 1.69          \\
                                                          & \textit{FedMD}                            & 12.7          & 0.0          & 13.2          & 1.12          & \multicolumn{1}{c|}{8.49}          & 3.66          & 0.0          & 4.54          & 2.20          & \multicolumn{1}{c|}{2.26}          & 3.19          & 0.0          & 1.03          & 1.74          & 3.12          \\ \hline
\multirow{3}{*}{\textit{MNIST}}                           & \textit{FedAvg}                           & 2.90          & 0.0          & 1.99          & 2.62          & \multicolumn{1}{c|}{2.22}          & 0.95          & 0.0          & 1.99          & 0.37          & \multicolumn{1}{c|}{1.20}          & 0.51          & 0.0          & 0.55          & 0.53          & 0.61          \\
                                                          & \textit{SplitFed}                         & 3.15          & 0.0          & 2.41          & 2.93          & \multicolumn{1}{c|}{3.72}          & 2.41          & 0.0          & 3.51          & 1.56          & \multicolumn{1}{c|}{2.33}          & 1.79          & 0.0          & 3.04          & 1.14          & 2.55          \\
                                                          & \textit{FedMD}                            & 2.17          & 0.0          & 3.05          & 2.42          & \multicolumn{1}{c|}{1.72}          & 4.25          & 0.0          & 0.80          & 0.81          & \multicolumn{1}{c|}{3.10}          & 5.71          & 0.0          & 5.50          & 2.08          & 5.76          \\ \hline
\multirow{3}{*}{\textit{FMNIST}}                          & \textit{FedAvg}                           & 10.0          & 0.0          & 3.18          & 11.0          & \multicolumn{1}{c|}{2.93}          & 17.3          & 0.0          & 7.29          & 21.3          & \multicolumn{1}{c|}{5.23}          & 33.4          & 0.0          & 1.13          & 38.0          & 13.0          \\
                                                          & \textit{SplitFed}                         & 12.8          & 0.0          & 5.27          & 12.3          & \multicolumn{1}{c|}{11.5}          & 20.9          & 0.0          & 2.03          & 23.1          & \multicolumn{1}{c|}{19.0}          & 35.8          & 0.0          & 3.23          & 37.8          & 29.4          \\
                                                          & \textit{FedMD}                            & 6.46          & 0.0          & 9.76          & 5.75          & \multicolumn{1}{c|}{3.38}          & 0.0           & 0.0          & 2.52          & 3.33          & \multicolumn{1}{c|}{0.0}           & 2.66          & 0.0          & 0.54          & 6.45          & 0.0           \\ \hline
\multirow{3}{*}{\textit{CIFAR10}}                         & \textit{FedAvg}                           & 4.36          & 0.0          & 1.85          & 1.50          & \multicolumn{1}{c|}{2.31}          & 0.19          & 0.0          & 0.0           & 0.69          & \multicolumn{1}{c|}{1.36}          & 7.72          & 0.0          & 0.0           & 0.69          & 4.22          \\
                                                          & \textit{SplitFed}                         & 3.52          & 0.0          & 0.93          & 0.93          & \multicolumn{1}{c|}{2.86}          & 3.39          & 0.0          & 0.0           & 0.17          & \multicolumn{1}{c|}{4.74}          & 3.21          & 0.0          & 0.0           & 0.0           & 6.35          \\
                                                          & \textit{FedMD}                            & 4.73          & 0.0          & 5.05          & 5.68          & \multicolumn{1}{c|}{0.72}          & 0.0           & 0.0          & 12.9          & 0.0           & \multicolumn{1}{c|}{7.07}          & 11.9          & 0.0          & 8.65          & 12.4          & 4.15          \\ \hline
\multirow{3}{*}{\textit{AGNews}}                          & \textit{FedAvg}                           & 0.80          & 0.0          & 0.28          & 0.46          & \multicolumn{1}{c|}{0.44}          & 0.0           & 0.0          & 1.88          & 0.0           & \multicolumn{1}{c|}{3.87}          & 1.49          & 0.0          & 0.08          & 0.0           & 1.52          \\
                                                          & \textit{SplitFed}                         & 3.20          & 0.0          & 0.46          & 0.35          & \multicolumn{1}{c|}{1.78}          & 0.0           & 0.0          & 0.72          & 0.0           & \multicolumn{1}{c|}{4.81}          & 3.21          & 0.0          & 0.0           & 0.0           & 3.61          \\
                                                          & \textit{FedMD}                            & 1.79          & 0.0          & 2.52          & 0.80          & \multicolumn{1}{c|}{1.63}          & 2.19          & 0.0          & 8.78          & 0.0           & \multicolumn{1}{c|}{2.35}          & 4.93          & 0.0          & 0.27          & 0.0           & 4.93          \\ \hline
\multirow{3}{*}{\textit{Yahoo}}                           & \textit{FedAvg}                           & 5.58          & 0.0          & 5.21          & 2.40          & \multicolumn{1}{c|}{3.41}          & 38.0          & 0.0          & 29.7          & 9.07          & \multicolumn{1}{c|}{11.6}          & 3.40          & 0.0          & 3.37          & 3.73          & 3.92          \\
                                                          & \textit{SplitFed}                         & 6.71          & 0.0          & 6.52          & 4.19          & \multicolumn{1}{c|}{4.16}          & 34.7          & 0.0          & 29.8          & 29.1          & \multicolumn{1}{c|}{32.8}          & 5.18          & 0.0          & 4.39          & 4.20          & 4.61          \\
                                                          & \textit{FedMD}                            & 4.75          & 0.0          & 5.70          & 4.06          & \multicolumn{1}{c|}{1.58}          & 35.4          & 0.0          & 43.2          & 38.6          & \multicolumn{1}{c|}{4.91}          & 4.39          & 0.0          & 5.87          & 5.55          & 0.58          \\ \hline
\multicolumn{2}{c|}{\textbf{Average}}                                                                 & \textbf{5.40} & \textbf{0.0} & \textbf{4.63} & \textbf{2.94} & \multicolumn{1}{c|}{\textbf{3.20}} & \textbf{8.08} & \textbf{0.0} & \textbf{7.35} & \textbf{6.42} & \multicolumn{1}{c|}{\textbf{5.29}} & \textbf{6.38} & \textbf{0.0} & \textbf{2.02} & \textbf{5.71} & \textbf{4.39} \\ \hline
\end{tabular}
}
\end{table*}

\begin{table*}[h]
\centering
\caption{\textbf{Varying Utility Metrics.} Reward allocation results under the data overvaluation attack, with the \textbf{Class-wise Accuracy} as the utility metric~\cite{schoch2022cs}. Since this utility metric is designed for classification tasks, the Rent dataset for tabular data regression is missing.}
\resizebox{\textwidth}{!}
{
\begin{tabular}{cl|ccccc|ccccc|ccccc}
\hline
\multicolumn{2}{c|}{\multirow{2}{*}{\textbf{\begin{tabular}[c]{@{}c@{}}CML \\ Setting\end{tabular}}}} & \multicolumn{5}{c|}{\textbf{Change in $\phi_i$ (\%)}}                     & \multicolumn{5}{c|}{\textbf{Change in $\phi_{-i}$ (\%)}}                 & \multicolumn{5}{c}{\textbf{Client-level valuation error}}                    \\ \cline{3-17} 
\multicolumn{2}{c|}{}                                                                                 & \textbf{SV}  & \textbf{TSV} & \textbf{LOO} & \textbf{BSV}  & \textbf{BV}  & \textbf{SV}  & \textbf{TSV} & \textbf{LOO} & \textbf{BSV} & \textbf{BV}  & \textbf{SV}   & \textbf{TSV} & \textbf{LOO}  & \textbf{BSV}  & \textbf{BV}   \\ \hline
\multirow{3}{*}{\textit{Bank}}                            & \textit{FedAvg}                           & +76          & 0.0          & +115         & +56           & +47          & -66          & 0.0          & 0.0          & -14          & -30          & 0.79          & 0.0          & 3.18          & 0.52          & 0.41          \\
                                                          & \textit{SplitFed}                         & +88          & 0.0          & +87          & +81           & +71          & -63          & 0.0          & 0.0          & -16          & -38          & 0.81          & 0.0          & 1.78          & 0.66          & 0.41          \\
                                                          & \textit{FedMD}                            & +76          & 0.0          & +5.5         & +65           & +45          & -36          & 0.0          & 0.0          & -12          & -18          & 0.54          & 0.0          & 0.40          & 0.49          & 0.29          \\ \hline
\multirow{3}{*}{\textit{MNIST}}                           & \textit{FedAvg}                           & +56          & 0.0          & +56          & +87           & +23          & -48          & 0.0          & 0.0          & -25          & -21          & 2.29          & 0.0          & 0.66          & 3.21          & 0.84          \\
                                                          & \textit{SplitFed}                         & +83          & 0.0          & +65          & +107          & +54          & -109         & 0.0          & 0.0          & -39          & -62          & 3.61          & 0.0          & 0.52          & 4.90          & 1.93          \\
                                                          & \textit{FedMD}                            & +59          & 0.0          & +28          & +97           & +22          & -54          & 0.0          & 0.0          & -32          & -23          & 2.71          & 0.0          & 0.84          & 3.15          & 1.11          \\ \hline
\multirow{3}{*}{\textit{FMNIST}}                          & \textit{FedAvg}                           & +111         & 0.0          & +41          & +134          & +21          & -139         & 0.0          & 0.0          & -93          & -134         & 0.88          & 0.0          & 0.31          & 4.25          & 0.28          \\
                                                          & \textit{SplitFed}                         & +129         & 0.0          & +43          & +160          & +106         & -193         & 0.0          & 0.0          & -186         & -163         & 2.05          & 0.0          & 0.33          & 7.63          & 0.90          \\
                                                          & \textit{FedMD}                            & +69          & 0.0          & +87          & +140          & +31          & -77          & 0.0          & 0.0          & -112         & -34          & 3.28          & 0.0          & 2.20          & 4.75          & 1.87          \\ \hline
\multirow{3}{*}{\textit{CIFAR10}}                         & \textit{FedAvg}                           & +95          & 0.0          & +62          & +84           & +29          & 132          & 0.0          & 0.0          & -67          & -106         & 1.03          & 0.0          & 0.32          & 3.93          & 0.39          \\
                                                          & \textit{SplitFed}                         & +119         & 0.0          & +68          & +113          & +105         & -144         & 0.0          & 0.0          & -71          & -136         & 1.61          & 0.0          & 0.34          & 5.45          & 0.80          \\
                                                          & \textit{FedMD}                            & +63          & 0.0          & +102         & +107          & +28          & -76          & 0.0          & 0.0          & -50          & -35          & 1.27          & 0.0          & 0.72          & 4.96          & 0.53          \\ \hline
\multirow{3}{*}{\textit{AGNews}}                          & \textit{FedAvg}                           & +71          & 0.0          & +151         & +96           & +36          & -76          & 0.0          & 0.0          & -31          & -37          & 13.1          & 0.0          & 7.05          & 12.4          & 4.07          \\
                                                          & \textit{SplitFed}                         & +84          & 0.0          & +153         & +111          & +50          & -113         & 0.0          & 0.0          & -44          & -61          & 17.3          & 0.0          & 5.36          & 10.8          & 6.11          \\
                                                          & \textit{FedMD}                            & +66          & 0.0          & +57          & +101          & +42          & -67          & 0.0          & 0.0          & -35          & -40          & 7.44          & 0.0          & 0.61          & 4.45          & 4.88          \\ \hline
\multirow{3}{*}{\textit{Yahoo}}                           & \textit{FedAvg}                           & +60          & 0.0          & +103         & +91           & +25          & -76          & 0.0          & 0.0          & -37          & -22          & 2.68          & 0.0          & 4.71          & 3.39          & 0.90          \\
                                                          & \textit{SplitFed}                         & +116         & 0.0          & +103         & +164          & +84          & -200         & 0.0          & 0.0          & -200         & -172         & 7.29          & 0.0          & 3.51          & 15.3          & 4.06          \\
                                                          & \textit{FedMD}                            & +67          & 0.0          & +120         & +100          & +29          & -89          & 0.0          & 0.0          & -50          & -29          & 2.56          & 0.0          & 3.61          & 2.76          & 1.08          \\ \hline
\multicolumn{2}{c|}{\textbf{Average}}                                                                 & \textbf{+83} & \textbf{0.0} & \textbf{+80} & \textbf{+105} & \textbf{+47} & \textbf{-83} & \textbf{0.0} & \textbf{0.0} & \textbf{-62} & \textbf{-65} & \textbf{3.96} & \textbf{0.0} & \textbf{2.03} & \textbf{5.17} & \textbf{1.71} \\ \hline
\end{tabular}
}
\end{table*}

\begin{table*}[h]
\centering
\caption{\textbf{Varying Utility Metrics.} Data selection results under the data overvaluation attack, with the \textbf{Class-wise Accuracy} as the utility metric~\cite{schoch2022cs}. Since this utility metric is designed for classification tasks, the Rent dataset for tabular data regression is missing.}
\small
\resizebox{\textwidth}{!}
{
\begin{tabular}{cl|ccccccccccccccc}
\hline
\multicolumn{2}{c|}{\multirow{3}{*}{\textbf{\begin{tabular}[c]{@{}c@{}}CML \\ setting\end{tabular}}}} & \multicolumn{15}{c}{\textbf{Decline (\%) in model utility due to data overvaluation}}                                                                                                                                                                                                \\ \cline{3-17} 
\multicolumn{2}{c|}{}                                                                                 & \multicolumn{5}{c|}{\textbf{Top-k Selection}}                                                     & \multicolumn{5}{c|}{\textbf{Above-average Selection}}                                             & \multicolumn{5}{c}{\textbf{Above-median Selection}}                          \\ \cline{3-17} 
\multicolumn{2}{c|}{}                                                                                 & \textbf{SV}   & \textbf{TSV} & \textbf{LOO}  & \textbf{BSV}  & \multicolumn{1}{c|}{\textbf{BV}}   & \textbf{SV}   & \textbf{TSV} & \textbf{LOO}  & \textbf{BSV}  & \multicolumn{1}{c|}{\textbf{BV}}   & \textbf{SV}   & \textbf{TSV} & \textbf{LOO}  & \textbf{BSV}  & \textbf{BV}   \\ \hline
\multirow{3}{*}{\textit{Bank}}                            & \textit{FedAvg}                           & 2.31          & 0.0          & 2.60          & 0.18          & \multicolumn{1}{c|}{0.19}          & 0.15          & 0.0          & 0.91          & 0.14          & \multicolumn{1}{c|}{0.02}          & 0.10          & 0.0          & 0.37          & 0.06          & 0.04          \\
                                                          & \textit{SplitFed}                         & 1.70          & 0.0          & 1.87          & 0.36          & \multicolumn{1}{c|}{0.09}          & 0.26          & 0.0          & 1.26          & 0.36          & \multicolumn{1}{c|}{0.03}          & 0.35          & 0.0          & 0.08          & 0.39          & 0.10          \\
                                                          & \textit{FedMD}                            & 0.93          & 0.0          & 0.46          & 0.12          & \multicolumn{1}{c|}{0.11}          & 0.15          & 0.0          & 0.41          & 0.09          & \multicolumn{1}{c|}{0.04}          & 0.12          & 0.0          & 0.24          & 0.11          & 0.04          \\ \hline
\multirow{3}{*}{\textit{MNIST}}                           & \textit{FedAvg}                           & 3.69          & 0.0          & 1.67          & 2.55          & \multicolumn{1}{c|}{2.39}          & 1.74          & 0.0          & 1.95          & 0.89          & \multicolumn{1}{c|}{1.50}          & 2.40          & 0.0          & 2.28          & 0.49          & 1.14          \\
                                                          & \textit{SplitFed}                         & 3.52          & 0.0          & 2.42          & 2.76          & \multicolumn{1}{c|}{4.17}          & 0.40          & 0.0          & 5.52          & 1.32          & \multicolumn{1}{c|}{0.49}          & 2.35          & 0.0          & 2.75          & 0.77          & 2.28          \\
                                                          & \textit{FedMD}                            & 3.00          & 0.0          & 1.99          & 2.44          & \multicolumn{1}{c|}{2.44}          & 1.15          & 0.0          & 2.52          & 1.01          & \multicolumn{1}{c|}{2.34}          & 4.02          & 0.0          & 1.70          & 2.08          & 2.55          \\ \hline
\multirow{3}{*}{\textit{FMNIST}}                          & \textit{FedAvg}                           & 2.70          & 0.0          & 6.32          & 12.8          & \multicolumn{1}{c|}{2.51}          & 0.06          & 0.0          & 0.06          & 20.5          & \multicolumn{1}{c|}{0.06}          & 0.25          & 0.0          & 0.08          & 25.4          & 0.10          \\
                                                          & \textit{SplitFed}                         & 12.9          & 0.0          & 6.72          & 12.5          & \multicolumn{1}{c|}{10.6}          & 6.16          & 0.0          & 0.18          & 0.0           & \multicolumn{1}{c|}{5.42}          & 17.4          & 0.0          & 0.11          & 26.2          & 6.55          \\
                                                          & \textit{FedMD}                            & 5.51          & 0.0          & 3.75          & 7.36          & \multicolumn{1}{c|}{2.05}          & 9.67          & 0.0          & 1.85          & 7.61          & \multicolumn{1}{c|}{5.36}          & 17.6          & 0.0          & 10.1          & 6.70          & 1.40          \\ \hline
\multirow{3}{*}{\textit{CIFAR10}}                         & \textit{FedAvg}                           & 0.36          & 0.0          & 0.34          & 1.28          & \multicolumn{1}{c|}{2.95}          & 0.24          & 0.0          & 0.0           & 0.49          & \multicolumn{1}{c|}{5.65}          & 1.22          & 0.0          & 0.0           & 7.68          & 0.45          \\
                                                          & \textit{SplitFed}                         & 6.72          & 0.0          & 0.30          & 2.19          & \multicolumn{1}{c|}{2.49}          & 9.67          & 0.0          & 0.0           & 0.0           & \multicolumn{1}{c|}{9.00}          & 12.1          & 0.0          & 0.0           & 12.8          & 2.31          \\
                                                          & \textit{FedMD}                            & 3.46          & 0.0          & 5.28          & 5.27          & \multicolumn{1}{c|}{1.05}          & 11.5          & 0.0          & 13.8          & 0.0           & \multicolumn{1}{c|}{5.06}          & 8.01          & 0.0          & 8.40          & 11.2          & 4.56          \\ \hline
\multirow{3}{*}{\textit{AGNews}}                          & \textit{FedAvg}                           & 0.48          & 0.0          & 0.23          & 0.85          & \multicolumn{1}{c|}{0.29}          & 2.87          & 0.0          & 0.0           & 0.20          & \multicolumn{1}{c|}{0.51}          & 0.33          & 0.0          & 0.08          & 0.79          & 0.56          \\
                                                          & \textit{SplitFed}                         & 1.30          & 0.0          & 0.58          & 0.76          & \multicolumn{1}{c|}{0.91}          & 4.15          & 0.0          & 0.23          & 1.03          & \multicolumn{1}{c|}{3.19}          & 4.42          & 0.0          & 0.06          & 1.03          & 4.44          \\
                                                          & \textit{FedMD}                            & 1.41          & 0.0          & 1.17          & 1.18          & \multicolumn{1}{c|}{0.91}          & 5.06          & 0.0          & 0.0           & 2.99          & \multicolumn{1}{c|}{4.00}          & 3.04          & 0.0          & 0.0           & 1.99          & 4.81          \\ \hline
\multirow{3}{*}{\textit{Yahoo}}                           & \textit{FedAvg}                           & 5.16          & 0.0          & 5.51          & 4.34          & \multicolumn{1}{c|}{2.45}          & 35.1          & 0.0          & 19.7          & 31.1          & \multicolumn{1}{c|}{6.61}          & 5.26          & 0.0          & 8.06          & 5.78          & 2.39          \\
                                                          & \textit{SplitFed}                         & 6.29          & 0.0          & 5.62          & 4.48          & \multicolumn{1}{c|}{4.55}          & 35.1          & 0.0          & 24.1          & 33.3          & \multicolumn{1}{c|}{32.9}          & 7.08          & 0.0          & 7.47          & 5.46          & 5.89          \\
                                                          & \textit{FedMD}                            & 5.02          & 0.0          & 4.44          & 4.46          & \multicolumn{1}{c|}{4.10}          & 38.7          & 0.0          & 41.4          & 27.7          & \multicolumn{1}{c|}{15.7}          & 5.40          & 0.0          & 6.26          & 5.25          & 5.44          \\ \hline
\multicolumn{2}{c|}{\textbf{Average}}                                                                 & \textbf{3.69} & \textbf{0.0} & \textbf{2.85} & \textbf{3.66} & \multicolumn{1}{c|}{\textbf{2.46}} & \textbf{9.01} & \textbf{0.0} & \textbf{6.33} & \textbf{7.15} & \multicolumn{1}{c|}{\textbf{5.44}} & \textbf{5.08} & \textbf{0.0} & \textbf{2.67} & \textbf{6.34} & \textbf{2.50} \\ \hline
\end{tabular}
}
\end{table*}

\clearpage

\subsection{Varying Reward Functions}
\label{appendix:exp_reward}

\begin{table*}[h]
\centering
\caption{\textbf{Varying Reward Functions.}}
\resizebox{\textwidth}{!}
{
\begin{tabular}{cl|ccccc|ccccc}
\hline
\multicolumn{2}{c|}{\multirow{2}{*}{\textbf{\begin{tabular}[c]{@{}c@{}}CML \\ Setting\end{tabular}}}} & \multicolumn{5}{c|}{\textbf{Change in balanced reward $R^{in}_i$ (\%)}}     & \multicolumn{5}{c}{\textbf{Change in proportional reward $R^{ex}_i$ (\%)}}  \\ \cline{3-12} 
\multicolumn{2}{c|}{}                                                                                 & \textbf{SV}   & \textbf{TSV} & \textbf{LOO}  & \textbf{BSV}  & \textbf{BV}   & \textbf{SV}  & \textbf{TSV} & \textbf{LOO} & \textbf{BSV} & \textbf{BV}  \\ \hline
\multirow{3}{*}{\textit{Bank}}                            & \textit{FedAvg}                           & +77           & 0.0          & +124          & +102          & +55           & +66          & 0.0          & +53          & +53          & +33          \\
                                                          & \textit{SplitFed}                         & +103          & 0.0          & +109          & +118          & +58           & +75          & 0.0          & +51          & +70          & +50          \\
                                                          & \textit{FedMD}                            & +69           & 0.0          & +18           & +92           & +41           & +65          & 0.0          & +22          & +56          & +33          \\ \hline
\multirow{3}{*}{\textit{Rent}}                            & \textit{FedAvg}                           & +150          & 0.0          & +197          & +160          & +165          & +28          & 0.0          & +93          & +32          & +22          \\
                                                          & \textit{SplitFed}                         & +168          & 0.0          & +194          & +170          & +178          & +47          & 0.0          & +94          & +38          & +37          \\
                                                          & \textit{FedMD}                            & +142          & 0.0          & +196          & +160          & +179          & +24          & 0.0          & +100         & +33          & +37          \\ \hline
\multirow{3}{*}{\textit{MNIST}}                           & \textit{FedAvg}                           & +164          & 0.0          & +119          & +176          & +132          & +38          & 0.0          & +34          & +45          & +15          \\
                                                          & \textit{SplitFed}                         & +169          & 0.0          & +118          & +179          & +163          & +57          & 0.0          & +42          & +54          & +38          \\
                                                          & \textit{FedMD}                            & +163          & 0.0          & +77           & +171          & +147          & +42          & 0.0          & +9.6         & +54          & +16          \\ \hline
\multirow{3}{*}{\textit{FMNIST}}                          & \textit{FedAvg}                           & +67           & 0.0          & +58           & +178          & +34           & +90          & 0.0          & +13          & +78          & +14          \\
                                                          & \textit{SplitFed}                         & +138          & 0.0          & +58           & +183          & +122          & +107         & 0.0          & +12          & +92          & +77          \\
                                                          & \textit{FedMD}                            & +132          & 0.0          & +174          & +181          & +162          & +27          & 0.0          & +50          & +79          & +28          \\ \hline
\multirow{3}{*}{\textit{CIFAR10}}                         & \textit{FedAvg}                           & +69           & 0.0          & +61           & +118          & +45           & +34          & 0.0          & +10          & +45          & +57          \\
                                                          & \textit{SplitFed}                         & +89           & 0.0          & +61           & +186          & +87           & +90          & 0.0          & +28          & +63          & +80          \\
                                                          & \textit{FedMD}                            & +65           & 0.0          & +113          & +186          & +59           & +29          & 0.0          & +73          & +63          & +26          \\ \hline
\multirow{3}{*}{\textit{AGNews}}                          & \textit{FedAvg}                           & +194          & 0.0          & +190          & +194          & +180          & +47          & 0.0          & +84          & +54          & +21          \\
                                                          & \textit{SplitFed}                         & +194          & 0.0          & +190          & +194          & +184          & +61          & 0.0          & +84          & +61          & +32          \\
                                                          & \textit{FedMD}                            & +192          & 0.0          & +124          & +189          & +185          & +47          & 0.0          & +42          & +57          & +24          \\ \hline
\multirow{3}{*}{\textit{Yahoo}}                           & \textit{FedAvg}                           & +140          & 0.0          & +187          & +193          & +169          & +15          & 0.0          & +61          & +56          & +24          \\
                                                          & \textit{SplitFed}                         & +188          & 0.0          & +187          & +197          & +185          & +92          & 0.0          & +61          & +88          & +75          \\
                                                          & \textit{FedMD}                            & +93           & 0.0          & +193          & +192          & +173          & +7.2         & 0.0          & +71          & +60          & +29          \\ \hline
\multicolumn{2}{c|}{\textbf{Average}}                                                                 & \textbf{+132} & \textbf{0.0} & \textbf{+131} & \textbf{+168} & \textbf{+129} & \textbf{+52} & \textbf{0.0} & \textbf{+52} & \textbf{+59} & \textbf{+37} \\ \hline
\end{tabular}
}
\end{table*}

\subsection{Data Overvaluation with Incomplete Knowledge}
\label{appendix:exp_incomplete_knowledge}

\begin{table*}[h]
\centering
\caption{\textbf{Incomplete Knowledge.} Reward allocation results under the data overvaluation attack with incomplete knowledge of the block subset $\datasubset$.}
\resizebox{\textwidth}{!}
{
\begin{tabular}{cl|ccccc|ccccc}
\hline
\multicolumn{2}{c|}{\multirow{2}{*}{\textbf{\begin{tabular}[c]{@{}c@{}}CML \\ Setting\end{tabular}}}} & \multicolumn{5}{c|}{\textbf{Change in balanced reward $R^{in}_i$ (\%)}}     & \multicolumn{5}{c}{\textbf{Change in proportional reward $R^{ex}_i$ (\%)}}  \\ \cline{3-12} 
\multicolumn{2}{c|}{}                                                                                 & \textbf{SV}   & \textbf{TSV} & \textbf{LOO}  & \textbf{BSV}  & \textbf{BV}   & \textbf{SV}  & \textbf{TSV} & \textbf{LOO} & \textbf{BSV} & \textbf{BV}  \\ \hline
\multirow{3}{*}{\textit{Bank}}                            & \textit{FedAvg}                           & +77           & 0.0          & +124          & +102          & +55           & +66          & 0.0          & +53          & +53          & +33          \\
                                                          & \textit{SplitFed}                         & +103          & 0.0          & +109          & +118          & +58           & +75          & 0.0          & +51          & +70          & +50          \\
                                                          & \textit{FedMD}                            & +69           & 0.0          & +18           & +92           & +41           & +65          & 0.0          & +22          & +56          & +33          \\ \hline
\multirow{3}{*}{\textit{Rent}}                            & \textit{FedAvg}                           & +150          & 0.0          & +197          & +160          & +165          & +28          & 0.0          & +93          & +32          & +22          \\
                                                          & \textit{SplitFed}                         & +168          & 0.0          & +194          & +170          & +178          & +47          & 0.0          & +94          & +38          & +37          \\
                                                          & \textit{FedMD}                            & +142          & 0.0          & +196          & +160          & +179          & +24          & 0.0          & +100         & +33          & +37          \\ \hline
\multirow{3}{*}{\textit{MNIST}}                           & \textit{FedAvg}                           & +164          & 0.0          & +119          & +176          & +132          & +38          & 0.0          & +34          & +45          & +15          \\
                                                          & \textit{SplitFed}                         & +169          & 0.0          & +118          & +179          & +163          & +57          & 0.0          & +42          & +54          & +38          \\
                                                          & \textit{FedMD}                            & +163          & 0.0          & +77           & +171          & +147          & +42          & 0.0          & +9.6         & +54          & +16          \\ \hline
\multirow{3}{*}{\textit{FMNIST}}                          & \textit{FedAvg}                           & +67           & 0.0          & +58           & +178          & +34           & +90          & 0.0          & +13          & +78          & +14          \\
                                                          & \textit{SplitFed}                         & +138          & 0.0          & +58           & +183          & +122          & +107         & 0.0          & +12          & +92          & +77          \\
                                                          & \textit{FedMD}                            & +132          & 0.0          & +174          & +181          & +162          & +27          & 0.0          & +50          & +79          & +28          \\ \hline
\multirow{3}{*}{\textit{CIFAR10}}                         & \textit{FedAvg}                           & +69           & 0.0          & +61           & +118          & +45           & +34          & 0.0          & +10          & +45          & +57          \\
                                                          & \textit{SplitFed}                         & +89           & 0.0          & +61           & +186          & +87           & +90          & 0.0          & +28          & +63          & +80          \\
                                                          & \textit{FedMD}                            & +65           & 0.0          & +113          & +186          & +59           & +29          & 0.0          & +73          & +63          & +26          \\ \hline
\multirow{3}{*}{\textit{AGNews}}                          & \textit{FedAvg}                           & +194          & 0.0          & +190          & +194          & +180          & +47          & 0.0          & +84          & +54          & +21          \\
                                                          & \textit{SplitFed}                         & +194          & 0.0          & +190          & +194          & +184          & +61          & 0.0          & +84          & +61          & +32          \\
                                                          & \textit{FedMD}                            & +192          & 0.0          & +124          & +189          & +185          & +47          & 0.0          & +42          & +57          & +24          \\ \hline
\multirow{3}{*}{\textit{Yahoo}}                           & \textit{FedAvg}                           & +140          & 0.0          & +187          & +193          & +169          & +15          & 0.0          & +61          & +56          & +24          \\
                                                          & \textit{SplitFed}                         & +188          & 0.0          & +187          & +197          & +185          & +92          & 0.0          & +61          & +88          & +75          \\
                                                          & \textit{FedMD}                            & +93           & 0.0          & +193          & +192          & +173          & +7.2         & 0.0          & +71          & +60          & +29          \\ \hline
\multicolumn{2}{c|}{\textbf{Average}}                                                                 & \textbf{+132} & \textbf{0.0} & \textbf{+131} & \textbf{+168} & \textbf{+129} & \textbf{+52} & \textbf{0.0} & \textbf{+52} & \textbf{+59} & \textbf{+37} \\ \hline
\end{tabular}
}
\end{table*}

\begin{table*}[h]
\centering
\caption{\textbf{Incomplete Knowledge.} Reward allocation results under the data overvaluation attack with incomplete knowledge of the block subset $\datasubset$.}
\label{tab:reward_alloc_incomplete}
\resizebox{\textwidth}{!}
{
\begin{tabular}{cl|ccccc|ccccc|ccccc}
\hline
\multicolumn{2}{c|}{\multirow{2}{*}{\textbf{\begin{tabular}[c]{@{}c@{}}CML \\ Setting\end{tabular}}}} & \multicolumn{5}{c|}{\textbf{Change in $\phi_i$ (\%)}}                    & \multicolumn{5}{c|}{\textbf{Change in $\phi_{-i}$ (\%)}}                 & \multicolumn{5}{c}{\textbf{Client-level valuation error}}                    \\ \cline{3-17} 
\multicolumn{2}{c|}{}                                                                                 & \textbf{SV}  & \textbf{TSV} & \textbf{LOO} & \textbf{BSV} & \textbf{BV}  & \textbf{SV}  & \textbf{TSV} & \textbf{LOO} & \textbf{BSV} & \textbf{BV}  & \textbf{SV}   & \textbf{TSV} & \textbf{LOO}  & \textbf{BSV}  & \textbf{BV}   \\ \hline
\multirow{3}{*}{\textit{Bank}}                            & \textit{FedAvg}                           & +66          & 0.0          & +119         & +69          & +81          & -27          & 0.0          & 0.0          & -7.0         & +46          & 0.48          & 0.0          & 3.18          & 0.57          & 1.12          \\
                                                          & \textit{SplitFed}                         & +75          & 0.0          & +95          & +86          & +99          & -35          & 0.0          & 0.0          & -9.9         & +45          & 0.58          & 0.0          & 1.79          & 0.70          & 1.03          \\
                                                          & \textit{FedMD}                            & +65          & 0.0          & +6.5         & +72          & +75          & -24          & 0.0          & 0.0          & -5.5         & +43          & 0.45          & 0.0          & 0.41          & 0.53          & 0.90          \\ \hline
\multirow{3}{*}{\textit{Rent}}                            & \textit{FedAvg}                           & +28          & 0.0          & +177         & +46          & +109         & -17          & 0.0          & 0.0          & -5.3         & +83          & 1.31          & 0.0          & 15.4          & 1.66          & 7.26          \\
                                                          & \textit{SplitFed}                         & +47          & 0.0          & +183         & +57          & +127         & -35          & 0.0          & 0.0          & -6.6         & +83          & 2.33          & 0.0          & 9.09          & 2.17          & 8.12          \\
                                                          & \textit{FedMD}                            & +24          & 0.0          & +187         & +49          & +109         & -15          & 0.0          & 0.0          & -4.6         & +60          & 1.25          & 0.0          & 17.5          & 1.59          & 6.39          \\ \hline
\multirow{3}{*}{\textit{MNIST}}                           & \textit{FedAvg}                           & +38          & 0.0          & +58          & +70          & +37          & -25          & 0.0          & 0.0          & -7.2         & +14          & 1.69          & 0.0          & 0.70          & 2.44          & 1.13          \\
                                                          & \textit{SplitFed}                         & +57          & 0.0          & +66          & +84          & +64          & -49          & 0.0          & 0.0          & -10          & +3.2         & 2.61          & 0.0          & 0.56          & 3.38          & 1.45          \\
                                                          & \textit{FedMD}                            & +42          & 0.0          & +29          & +85          & +34          & -30          & 0.0          & 0.0          & -7.6         & +10          & 1.83          & 0.0          & 0.87          & 2.20          & 1.22          \\ \hline
\multirow{3}{*}{\textit{FMNIST}}                          & \textit{FedAvg}                           & +90          & 0.0          & +44          & +117         & +37          & -49          & 0.0          & 0.0          & -44          & +27          & 0.48          & 0.0          & 0.36          & 3.37          & 0.26          \\
                                                          & \textit{SplitFed}                         & +107         & 0.0          & +44          & +141         & +115         & -157         & 0.0          & 0.0          & -82          & +5.0         & 1.27          & 0.0          & 0.38          & 4.79          & 0.68          \\
                                                          & \textit{FedMD}                            & +27          & 0.0          & +86          & +129         & +42          & -17          & 0.0          & 0.0          & -27          & -3.6         & 0.93          & 0.0          & 1.99          & 3.59          & 1.27          \\ \hline
\multirow{3}{*}{\textit{CIFAR10}}                         & \textit{FedAvg}                           & +34          & 0.0          & +62          & +77          & +82          & -66          & 0.0          & 0.0          & -34          & +24          & 0.64          & 0.0          & 0.36          & 2.87          & 0.55          \\
                                                          & \textit{SplitFed}                         & +90          & 0.0          & +65          & +100         & +119         & -101         & 0.0          & 0.0          & -32          & +30          & 0.90          & 0.0          & 0.38          & 3.88          & 0.98          \\
                                                          & \textit{FedMD}                            & +29          & 0.0          & +107         & +106         & +40          & -17          & 0.0          & 0.0          & -12          & +0.4         & 0.45          & 0.0          & 0.80          & 3.67          & 0.36          \\ \hline
\multirow{3}{*}{\textit{AGNews}}                          & \textit{FedAvg}                           & +47          & 0.0          & +153         & +87          & +59          & -37          & 0.0          & 0.0          & -9.8         & +27          & 7.09          & 0.0          & 7.23          & 7.06          & 6.87          \\
                                                          & \textit{SplitFed}                         & +61          & 0.0          & +154         & +98          & +68          & -56          & 0.0          & 0.0          & -13          & +17          & 8.92          & 0.0          & 5.53          & 6.35          & 6.38          \\
                                                          & \textit{FedMD}                            & +47          & 0.0          & +58          & +92          & +64          & -36          & 0.0          & 0.0          & -10          & +28          & 3.86          & 0.0          & 0.62          & 2.84          & 8.21          \\ \hline
\multirow{3}{*}{\textit{Yahoo}}                           & \textit{FedAvg}                           & +15          & 0.0          & +110         & +96          & +44          & -8.3         & 0.0          & 0.0          & -3.1         & +5.5         & 0.44          & 0.0          & 5.42          & 2.95          & 0.85          \\
                                                          & \textit{SplitFed}                         & +92          & 0.0          & +110         & +155         & +93          & -147         & 0.0          & 0.0          & -37          & -60          & 4.29          & 0.0          & 4.06          & 9.63          & 3.09          \\
                                                          & \textit{FedMD}                            & +7.2         & 0.0          & +129         & +106         & +48          & -3.8         & 0.0          & 0.0          & -1.4         & +1.1         & 0.18          & 0.0          & 4.21          & 2.36          & 0.81          \\ \hline
\multicolumn{2}{c|}{\textbf{Average}}                                                                 & \textbf{+52} & \textbf{0.0} & \textbf{+97} & \textbf{+92} & \textbf{+74} & \textbf{-45} & \textbf{0.0} & \textbf{0.0} & \textbf{-18} & \textbf{+23} & \textbf{2.00} & \textbf{0.0} & \textbf{3.85} & \textbf{3.27} & \textbf{2.81} \\ \hline
\end{tabular}
}
\end{table*}

\begin{table*}[h]
\centering
\caption{\textbf{Incomplete Knowledge.} Data selection results under the data overvaluation attack with incomplete knowledge of the block subset $\datasubset$.}
\label{tab:data_sel_decline_incomplete}
\small
\resizebox{\textwidth}{!}
{
\begin{tabular}{cl|ccccccccccccccc}
\hline
\multicolumn{2}{c|}{\multirow{3}{*}{\textbf{\begin{tabular}[c]{@{}c@{}}CML \\ setting\end{tabular}}}} & \multicolumn{15}{c}{\textbf{Decline (\%) in model utility due to data overvaluation}}                                                                                                                                                                                                \\ \cline{3-17} 
\multicolumn{2}{c|}{}                                                                                 & \multicolumn{5}{c|}{\textbf{Top-k selection}}                                                     & \multicolumn{5}{c|}{\textbf{Above-average selection}}                                             & \multicolumn{5}{c}{\textbf{Above-median selection}}                          \\ \cline{3-17} 
\multicolumn{2}{c|}{}                                                                                 & \textbf{SV}   & \textbf{TSV} & \textbf{LOO}  & \textbf{BSV}  & \multicolumn{1}{c|}{\textbf{BV}}   & \textbf{SV}   & \textbf{TSV} & \textbf{LOO}  & \textbf{BSV}  & \multicolumn{1}{c|}{\textbf{BV}}   & \textbf{SV}   & \textbf{TSV} & \textbf{LOO}  & \textbf{BSV}  & \textbf{BV}   \\ \hline
\multirow{3}{*}{\textit{Bank}}                            & \textit{FedAvg}                           & 0.25          & 0.0          & 2.61          & 0.29          & \multicolumn{1}{c|}{0.37}          & 0.09          & 0.0          & 1.03          & 0.20          & \multicolumn{1}{c|}{0.20}          & 0.09          & 0.0          & 0.35          & 0.06          & 0.13          \\
                                                          & \textit{SplitFed}                         & 0.34          & 0.0          & 1.94          & 0.97          & \multicolumn{1}{c|}{0.27}          & 0.19          & 0.0          & 1.26          & 0.33          & \multicolumn{1}{c|}{0.14}          & 0.37          & 0.0          & 0.08          & 1.46          & 0.28          \\
                                                          & \textit{FedMD}                            & 0.71          & 0.0          & 0.40          & 1.36          & \multicolumn{1}{c|}{0.32}          & 0.19          & 0.0          & 0.49          & 0.09          & \multicolumn{1}{c|}{0.06}          & 0.10          & 0.0          & 0.21          & 0.30          & 0.13          \\ \hline
\multirow{3}{*}{\textit{Rent}}                            & \textit{FedAvg}                           & 1.37          & 0.0          & 11.0          & 2.04          & \multicolumn{1}{c|}{11.4}          & 0.49          & 0.0          & 1.05          & 1.80          & \multicolumn{1}{c|}{8.39}          & 1.23          & 0.0          & 1.28          & 2.75          & 3.89          \\
                                                          & \textit{SplitFed}                         & 2.14          & 0.0          & 11.1          & 2.10          & \multicolumn{1}{c|}{11.9}          & 1.28          & 0.0          & 2.92          & 1.73          & \multicolumn{1}{c|}{5.64}          & 1.61          & 0.0          & 1.73          & 1.69          & 4.52          \\
                                                          & \textit{FedMD}                            & 1.94          & 0.0          & 12.7          & 4.10          & \multicolumn{1}{c|}{13.6}          & 0.98          & 0.0          & 4.56          & 2.79          & \multicolumn{1}{c|}{5.80}          & 2.84          & 0.0          & 2.20          & 2.39          & 3.02          \\ \hline
\multirow{3}{*}{\textit{MNIST}}                           & \textit{FedAvg}                           & 3.32          & 0.0          & 1.56          & 3.38          & \multicolumn{1}{c|}{3.35}          & 1.72          & 0.0          & 1.95          & 1.13          & \multicolumn{1}{c|}{1.84}          & 2.40          & 0.0          & 1.38          & 0.49          & 1.96          \\
                                                          & \textit{SplitFed}                         & 3.92          & 0.0          & 2.39          & 3.72          & \multicolumn{1}{c|}{5.23}          & 0.83          & 0.0          & 5.52          & 1.24          & \multicolumn{1}{c|}{0.82}          & 2.63          & 0.0          & 2.75          & 0.04          & 3.88          \\
                                                          & \textit{FedMD}                            & 2.74          & 0.0          & 2.00          & 3.40          & \multicolumn{1}{c|}{3.65}          & 1.43          & 0.0          & 2.52          & 1.74          & \multicolumn{1}{c|}{2.52}          & 4.03          & 0.0          & 1.70          & 1.84          & 4.38          \\ \hline
\multirow{3}{*}{\textit{FMNIST}}                          & \textit{FedAvg}                           & 4.46          & 0.0          & 6.34          & 13.7          & \multicolumn{1}{c|}{3.93}          & 0.06          & 0.0          & 0.06          & 21.0          & \multicolumn{1}{c|}{0.06}          & 0.24          & 0.0          & 0.08          & 25.4          & 0.10          \\
                                                          & \textit{SplitFed}                         & 12.9          & 0.0          & 6.70          & 12.2          & \multicolumn{1}{c|}{11.7}          & 17.9          & 0.0          & 4.77          & 0.0           & \multicolumn{1}{c|}{4.91}          & 19.2          & 0.0          & 0.11          & 26.2          & 5.62          \\
                                                          & \textit{FedMD}                            & 5.74          & 0.0          & 3.55          & 7.15          & \multicolumn{1}{c|}{1.46}          & 3.88          & 0.0          & 1.93          & 8.40          & \multicolumn{1}{c|}{6.51}          & 12.0          & 0.0          & 8.33          & 7.47          & 1.14          \\ \hline
\multirow{3}{*}{\textit{CIFAR10}}                         & \textit{FedAvg}                           & 1.37          & 0.0          & 0.32          & 2.21          & \multicolumn{1}{c|}{4.65}          & 8.52          & 0.0          & 0.0           & 0.49          & \multicolumn{1}{c|}{3.56}          & 1.65          & 0.0          & 0.0           & 11.6          & 0.0           \\
                                                          & \textit{SplitFed}                         & 7.15          & 0.0          & 0.22          & 2.23          & \multicolumn{1}{c|}{6.27}          & 19.7          & 0.0          & 0.0           & 0.0           & \multicolumn{1}{c|}{8.94}          & 14.1          & 0.0          & 0.0           & 12.8          & 5.00          \\
                                                          & \textit{FedMD}                            & 1.86          & 0.0          & 5.45          & 9.33          & \multicolumn{1}{c|}{1.09}          & 3.95          & 0.0          & 15.8          & 0.0           & \multicolumn{1}{c|}{3.07}          & 5.36          & 0.0          & 8.18          & 11.2          & 4.56          \\ \hline
\multirow{3}{*}{\textit{AGNews}}                          & \textit{FedAvg}                           & 0.28          & 0.0          & 0.85          & 0.91          & \multicolumn{1}{c|}{1.32}          & 2.01          & 0.0          & 2.33          & 0.20          & \multicolumn{1}{c|}{2.84}          & 0.0           & 0.0          & 0.08          & 0.41          & 3.91          \\
                                                          & \textit{SplitFed}                         & 1.23          & 0.0          & 0.58          & 1.27          & \multicolumn{1}{c|}{2.50}          & 3.67          & 0.0          & 0.23          & 1.03          & \multicolumn{1}{c|}{4.18}          & 4.03          & 0.0          & 0.06          & 1.47          & 4.44          \\
                                                          & \textit{FedMD}                            & 0.72          & 0.0          & 1.33          & 1.27          & \multicolumn{1}{c|}{1.77}          & 0.39          & 0.0          & 0.0           & 3.06          & \multicolumn{1}{c|}{6.45}          & 0.09          & 0.0          & 0.0           & 0.85          & 6.05          \\ \hline
\multirow{3}{*}{\textit{Yahoo}}                           & \textit{FedAvg}                           & 2.16          & 0.0          & 5.50          & 5.36          & \multicolumn{1}{c|}{3.59}          & 5.58          & 0.0          & 19.7          & 31.1          & \multicolumn{1}{c|}{7.84}          & 2.38          & 0.0          & 9.33          & 5.04          & 4.64          \\
                                                          & \textit{SplitFed}                         & 6.10          & 0.0          & 5.93          & 5.75          & \multicolumn{1}{c|}{5.93}          & 35.1          & 0.0          & 25.8          & 33.3          & \multicolumn{1}{c|}{32.9}          & 8.43          & 0.0          & 6.63          & 6.18          & 5.89          \\
                                                          & \textit{FedMD}                            & 0.92          & 0.0          & 4.51          & 4.84          & \multicolumn{1}{c|}{4.54}          & 7.48          & 0.0          & 41.4          & 29.4          & \multicolumn{1}{c|}{14.9}          & 1.45          & 0.0          & 6.26          & 5.25          & 5.56          \\ \hline
\multicolumn{2}{c|}{\textbf{Average}}                                                                 & \textbf{2.93} & \textbf{0.0} & \textbf{4.14} & \textbf{4.17} & \multicolumn{1}{c|}{\textbf{4.71}} & \textbf{5.50} & \textbf{0.0} & \textbf{6.35} & \textbf{6.62} & \multicolumn{1}{c|}{\textbf{5.79}} & \textbf{4.01} & \textbf{0.0} & \textbf{2.42} & \textbf{5.95} & \textbf{3.29} \\ \hline
\end{tabular}
}
\end{table*}

\end{document}